\newif\ifArxiv
\Arxivtrue
\newif\ifPagesLimited
\PagesLimitedfalse
\newif\ifDraft
\Draftfalse
\ifDraft
\PassOptionsToClass{draft}{article}
\fi
\ifArxiv
\documentclass[12pt,a4paper]{article}
\else
\documentclass[conference,letterpaper]{IEEEtran}
\fi

\usepackage{geometry}
\usepackage[dvipsnames]{xcolor}
\usepackage{xspace}
\usepackage[english]{babel}
\usepackage{amsthm}
\usepackage[T1]{fontenc}
\usepackage{amsmath}
\usepackage{amssymb}
\usepackage{bm} % Bold math

\usepackage[utf8]{inputenc}
\usepackage{fontawesome} \let\warning\faBell \let\force\faHandORight% \warning
\usepackage{tikz}
\usetikzlibrary{automata}
\usetikzlibrary{calc}
\usetikzlibrary{shapes}
\usetikzlibrary{arrows}
\usetikzlibrary{decorations.pathmorphing}
\usepackage{enumitem}
\usepackage{hyperref}
\usepackage{comment}
\usepackage{colortbl}
\usepackage{array}
\usepackage{cases}
\usepackage{multirow}
\usepackage{mathpartir}
\usepackage{rotating}
\usepackage{subcaption}
\usepackage{multicol}

\renewcommand{\Game}{\mathcal{G}}
\newcommand{\Gameb}{\mathcal{H}}

\newcommand{\Mach}{\mathcal{M}}

\newcommand{\<}{\langle}
\renewcommand{\>}{\rangle}

\newcommand{\bsigma}{\bar{\sigma}}

\renewcommand{\|}{\upharpoonright}

\makeatletter
\newcommand{\m@ken@medcal}[2][\@error]{\expandafter\def\csname #1\endcsname{\ensuremath{\mathcal{#2}}}}
\newcommand{\m@kecal}{\@dblarg\m@ken@medcal}
\newcommand{\m@ken@medbb}[2][\@error]{\expandafter\def\csname #1\endcsname{{\ensuremath{\mathbb{#2}}}}}
\newcommand{\m@kebb}{\@dblarg\m@ken@medbb}
\newcommand{\m@ken@medsf}[2][\@error]{\expandafter\def\csname #1\endcsname{{\ensuremath{\mathsf{#2}}}}}
\newcommand{\m@kesf}{\@dblarg\m@ken@medsf}
\newcommand{\m@ken@medmathtxt}[2][\@error]{\expandafter\def\csname #1\endcsname{{\ensuremath{\textit{#2}}}}}
\newcommand{\m@kemathtxt}{\@dblarg\m@ken@medmathtxt}
\newcommand{\m@ken@medtt}[2][\@error]{\expandafter\def\csname #1\endcsname{{\ensuremath{\texttt{#2}}}}}
\newcommand{\m@kett}{\@dblarg\m@ken@medtt}

\m@kecal[Scen]{S}
\m@kecal[pc]{C}
\m@kecal[colors]{C}
\m@kecal[agents]{A}
\m@kecal[auto]{A} % Not used in the same part
\m@kecal[Lang]{L}
\m@kecal[behaviors]{B}
\m@kesf{coll}
\m@kesf{Ind}
\m@kesf{Sane}
\m@kesf{Wrong}
\m@kesf{Scam}
\m@kesf{Corr}
\m@kesf[noAct]{skip}
\m@kesf[actions]{Act}
\m@kesf{Streett}

\m@kebb[alice]{A}
\m@kebb[bob]{B}
\m@kebb[channel]{C}
\m@kebb[charlie]{C}
\m@kebb[intruder]{I}
\m@kebb[ttp]{T}
\m@kebb[player]{P}
\m@kebb[ghost]{G}
\m@kebb[WW]{W}
\m@kebb[solver]{S}
\m@kebb[opponent]{O}
\m@kebb[adam]{A}
\m@kebb[eve]{E}

\m@kesf{EOO}
\m@kesf{EOR}
\m@kesf{NRO}
\m@kesf{NRR}
\m@kesf{Sub}
\m@kesf{Con}
\m@kesf{Promise}
\m@kesf{Ack}
\m@kesf{Resolve}
\m@kesf{Aborted}
\m@kesf{Signed}
\m@kesf{Curse}
\m@kesf{Boast}
\m@kesf{Traces}
\m@kesf{tr}
\m@kesf{CDev}
\m@kesf{a}
\m@kesf[cd]{d}
\m@kesf[ecd]{e}
\m@kesf{EDev}
\m@kesf{Parity}
\m@kesf{D}
\m@kesf{B}
\m@kesf{Com}

\m@kesf{s}
\m@kesf{r}

\m@kecal[untrusted]{U}
\m@kecal[trusted]{T}
\m@kecal[nature]{N}

\newcommand{\Occ}{\mathsf{Occ}}
\newcommand{\Inf}{\mathsf{Inf}}
\newcommand{\Fin}{\mathsf{Fin}}

\newcommand{\obj}[2]{\mbox{\tikz{\node[font=\tiny,inner sep=1.5pt,draw,shape=rectangle,solid,minimum width=0pt,fill=white,rounded corners=0pt] {$#1\rightarrow#2$};}}}
\newcommand{\iab}{\obj{\alice}{\bob}}
\newcommand{\iba}{\obj{\bob}{\alice}}
\newcommand{\iac}{\obj{\alice}{\charlie}}
\newcommand{\ica}{\obj{\charlie}{\alice}}
\newcommand{\ibc}{\obj{\bob}{\charlie}}
\newcommand{\icb}{\obj{\charlie}{\bob}}
\newcommand{\iij}{\obj{i}{j}}

\newcommand{\iki}{\obj{k}{i}}

\newcommand{\mmid}{\ \middle|\ }
\makeatother

\newcommand{\NN}{\mathbb{N}}

\newcommand{\RR}{\mathbb{R}}

\renewcommand{\SS}{\mathbb{S}}

\newcommand{\PP}{\mathbb{P}}
\newcommand{\CC}{\mathbb{C}}

\newcommand{\II}{\mathbb{I}}

\newcommand{\bt}{\blacktriangle}
\newcommand{\btd}{\blacktriangledown}

\renewcommand{\l}{\ell}
\renewcommand{\epsilon}{\varepsilon}
\renewcommand{\phi}{\varphi}

\newcommand{\Hist}{\mathsf{Hist}}
\newcommand{\Dev}{\mathsf{Dev}}
\newcommand{\noDev}{\Delta_\emptyset}

\newcommand{\Comp}{\mathsf{Comp}}

\newcommand{\M}{\mathcal{M}}

\newcommand{\Poly}{\ensuremath{\mathsf{P}}}
\newcommand{\NP}{\ensuremath{\mathsf{NP}}}
\newcommand{\ExpTime}{\ensuremath{\mathsf{EXPTIME}}}
\newcommand{\ExpSpace}{\ensuremath{\mathsf{EXPSPACE}}}

\newcommand{\coNP}{\ensuremath{\mathsf{coNP}}}

\newcommand{\PSpace}{\ensuremath{\mathsf{PSPACE}}}

\newcommand{\W}{\ensuremath{\mathsf{W}}}
\newcommand{\XP}{\ensuremath{\mathsf{XP}}}
\newcommand{\FPT}{\ensuremath{\mathsf{FPT}}}

\newcommand{\bx}{\bar{x}}
\newcommand{\by}{\bar{y}}

\newcommand{\balpha}{\bar{\alpha}}
\newcommand{\bbeta}{\bar{\beta}}

\newcommand{\bs}{\bar{s}}

\newcommand{\dchi}{\dot{\chi}}

\newcommand{\dalpha}{\dot{\alpha}}
\newcommand{\dbeta}{\dot{\beta}}

\newcommand{\autoMR}{\bgroup\hat{\auto}\egroup}

\newcommand{\eventually}{\mathbf{F}\,}
\newcommand{\globally}{\mathbf{G}\,}
\newcommand{\eqdef}{\ensuremath{\mathrel{\stackrel{\text{def}}{=}}}}
\newcommand{\procPar}{\mathrel{||}}

\makeatletter
\newcommand{\Biggg}{\bBigg@{4}}
\newcommand{\Bigggg}{\bBigg@{5}}
\newcommand{\Biggggg}{\bBigg@{6}}
\makeatother

\newtheorem{thm}{Theorem}

\newtheorem{lm}{Lemma}

\theoremstyle{definition}
\newtheorem{defi}{Definition}
\newtheorem{pb}{Problem}

\theoremstyle{remark}

\makeatletter
\renewcommand\@endtheorem{\vvv@endmarker\endtrivlist\@endpefalse}
\newcommand\vvv@endmarker{%
  {\unskip\nobreak\hfil\penalty50
  \hskip2em\vadjust{}\nobreak\hfil$\triangleleft$
  \parfillskip=0pt \finalhyphendemerits=0 \par
  \penalty 10000 \parskip=0pt\noindent}\ignorespaces}
\makeatother
\newtheorem{rk}{Remark}

\newtheorem{exa}{Example}

\tikzstyle{abstractstate}=[state,shape=rectangle,rounded corners=7pt]
\tikzstyle{aliceMove}=[solid,red,thick]
\tikzstyle{bobMove}=[dashed,blue,thick]
\tikzstyle{ttpMove}=[dotted,green!60!black,thick]
\tikzstyle{infraMove}=[dashdotted,orange,thick]

\tikzstyle{player1}=[regular polygon,regular polygon sides=4,draw,inner sep=0.5pt]
\tikzstyle{player2}=[circle,draw,inner sep=2pt]
\tikzstyle{player3}=[regular polygon,regular polygon sides=6,,draw,inner sep=0.5pt]
\tikzstyle{initial}+=[initial by arrow,initial text=]

\newcommand{\playerOneSymb}{\bgroup\mathchoice{%
        \tikz{\node[player1,minimum size=9pt]{~};}% Display
    }{%
        \tikz{\node[player1,minimum size=9pt]{~};}% Text
    }{%
        \tikz{\useasboundingbox (-3pt,-3pt) rectangle (3pt,4pt); \node[player1,minimum size=6pt]{~};}% Sub
    }{%
        \tikz{\node[player1,minimum size=3pt]{~};}% Subsub
    }\egroup}

\newcommand{\playerTwoSymb}{\bgroup\mathchoice{%
        \tikz{\useasboundingbox (-3.5pt,-3pt) rectangle (3.5pt,4pt);\node[player2,minimum size=7pt]{~};}% Display
    }{%
        \tikz{\useasboundingbox (-3.5pt,-3pt) rectangle (3.5pt,4pt);\node[player2,minimum size=7pt]{~};}% Text
    }{%
        \tikz{\useasboundingbox (-2pt,-3pt) rectangle (2pt,4pt);\node[player2,minimum size=4pt,inner sep=0pt]{};}% Sub
    }{%
        \tikz{\useasboundingbox (-1.5pt,-2pt) rectangle (1.5pt,2.5pt);\node[player2,inner sep=0pt,minimum size=3pt]{};}% Subsub
    }\egroup}
\newcommand{\playerThreeSymb}{\bgroup\mathchoice{%
        \tikz{\node[player3,minimum size=8pt]{~};}% Display
    }{%
        \tikz{\node[player3,minimum size=8pt]{~};}% Text
    }{%
        \tikz{\useasboundingbox (-2pt,-2pt) rectangle (2.5pt,4pt);\node[player3,minimum size=5pt]{~};}% Sub
    }{%
        \tikz{\useasboundingbox (-1.25pt,-2pt) rectangle (1.5pt,2.5pt);\node[player3,minimum size=3.5pt]{~};}% Subsub
    }\egroup}

\tikzstyle{initial}+=[initial by arrow,initial text={}]
\tikzstyle{every state}+=[minimum size=5pt,fill, top color=white, bottom color=black!20]
\tikzstyle{mealyTrans}=[->,font=\scriptsize]

    \tikzstyle{nobodyState}=[state,rectangle,inner sep=3pt,minimum size=5pt]
    \tikzstyle{innocentState}=[state,regular polygon,regular polygon sides=6,inner sep=1.5pt]
    \tikzstyle{literalState}[$i$]=[nobodyState,%
    prefix after command={\pgfextra{\tikzset{every label/.style={
                    fill=black,text=white,rounded corners=3pt,outer sep=0pt,font=\scriptsize,label distance=-3pt,inner sep=2pt
    }}}},
    label=-60:{#1}]
    \tikzstyle{opponentState}=[state,regular polygon,regular polygon sides=6,inner sep=0pt,minimum size=5pt]
    \tikzstyle{solverState}=[state,circle,inner sep=1pt,minimum size=5pt]
    \tikzstyle{loseState}=[shape=cloud,inner sep=1pt,fill=red!80!black,text=white,font=\scriptsize,cloud ignores aspect]
    \newcommand{\cobuchiLabel}[3][]{\node[anchor=west,loseState,xshift=-2pt,#1] at (#2.east) {#3};}

    \newcommand{\cobuchiLabelLeft}[3][]{\node[anchor=east,loseState,xshift=2pt,#1] at (#2.west) {#3};}

    \tikzstyle{automatonState}=[state,circle,inner sep=1pt,minimum size=5pt,fill, bottom color=white, top color=black!20]

    \tikzstyle{vanishingPattern}=[dash pattern=on 5pt off 2pt on 2pt off 2pt on 2pt off 2pt]
    \newcommand{\outVanishingArrows}[3][]{%
        \foreach \a in {#3} {
            \coordinate (#2Out\a) at ($(#2.\a)+(\a:15pt)$);
            \path[draw,vanishingPattern,#1] (#2) -- (#2Out\a);
        }
    }
    \newcommand{\inVanishingArrows}[3][]{%
        \foreach \a in {#3} {
            \coordinate (#2In\a) at ($(#2.\a)+(\a:15pt)$);
            \path[draw,vanishingPattern,<-,#1] (#2) -- (#2In\a);
        }
    }
    
    \pgfdeclarelayer{arrows}
    \pgfsetlayers{arrows,main}
    
\makeatletter
\def\myvdots{\vbox{\baselineskip4\p@ \lineskiplimit\z@
    \hbox{.}\hbox{.}\hbox{.}}}
\def\myddots{\mathinner{\mkern1mu\raise7\p@\vbox{\hbox{.}}\mkern2mu
    \raise4\p@\hbox{.}\mkern2mu\raise\p@\hbox{.}\mkern1mu}}
\def\myiddots{\mathinner{\mkern1mu\raise-6\p@\vbox{\hbox{.}}\mkern2mu
    \raise-3\p@\hbox{.}\mkern2mu\raise\p@\hbox{.}\mkern1mu}}
\makeatother

\newcommand{\og}{\guillemetleft\,}
\newcommand{\fg}{\,\guillemetright\xspace}
\ifDraft
\newcommand{\Ms}[1]{\textcolor{BurntOrange}{\textit{Mathieu:} \og#1\fg}}
\newcommand{\Le}[1]{\textcolor{Magenta}{\textit{L\'eonard:} \og#1\fg}}
\newcommand{\JF}[1]{\textcolor{NavyBlue}{\textit{Jean-Fran\c cois:} \og#1\fg}}
\newcommand{\Ma}[1]{\textcolor{PineGreen}{\textit{Marie:} \og#1\fg}}
\newcommand{\Gu}[1]{\textcolor{BlueViolet}{\textit{Guillaume:} \og#1\fg}}
\else
\makeatletter
\newcommand{\@untreatedComment}[1]{\@latex@warning{Untreated comment from #1}}
\newcommand{\Ms}[1]{\@untreatedComment{Mathieu}}
\newcommand{\Le}[1]{\@untreatedComment{Léonard}}
\newcommand{\JF}[1]{\@untreatedComment{Jean-François}}
\newcommand{\Ma}[1]{\@untreatedComment{Marie}}
\newcommand{\Gu}[1]{\@untreatedComment{Guillaume}}
\makeatother
\fi

\title{Pessimism of the Will, Optimism of the Intellect:\ifArxiv\\\else\ \fi Fair Protocols with Malicious but Rational Agents}

\ifArxiv
\author{Léonard Brice$^1$ \and Jean-François Raskin$^1$ \and Mathieu Sassolas$^1$ \and Guillaume Scerri$^2$ \and Marie Van Den Bogaard$^3$
\and$^1$Université libre de Bruxelles, Belgium \and $^2$Université Paris-Saclay, ENS Paris Saclay \& CNRS, LMF, Gif-sur-Yvette, France \and $^3$Univ. Gustave Eiffel, CNRS, LIGM, Marne-la-Vallée, France}
\else

\author{
\IEEEauthorblockN{Léonard Brice,\\Jean-François Raskin,\\ Mathieu Sassolas}
\IEEEauthorblockA{Université libre de Bruxelles\\Bruxelles, Belgium}
\and
\IEEEauthorblockN{Guillaume Scerri}
\IEEEauthorblockA{Université Paris-Saclay,\\ ENS Paris-Saclay \& CNRS, LMF\\ Gif-sur-Yvette, France}
\and
\IEEEauthorblockN{Marie Van Den Bogaard}
\IEEEauthorblockA{Univ. Gustave Eiffel, CNRS, LIGM\\Marne-la-Vallée, France}
}
\fi

\begin{document}

	\maketitle

\begin{abstract}
    Fairness is a desirable and crucial property of many protocols that handle, for instance, exchanges of message.
    It states that if at least one agent engaging in the protocol is honest, then either the protocol will unfold correctly and fulfill its intended goal for all participants, or it will fail for everyone.
    In this work, we present a game-based framework for the study of fairness protocols, that does not define a priori an attacker model.
    It is based on the notion of \emph{strong secure equilibria}, and leverages the conceptual and algorithmic toolbox of game theory.
    In the case of finite games, we provide decision procedures with tight complexity bounds for determining whether a protocol is immune to nefarious attacks from a coalition of participants, and whether such a protocol could exist based on the underlying graph structure and objectives.
\end{abstract}

\section{Introduction}

% \emph{De tous temps les hommes...}
\emph{Fairness} is a desirable and crucial property of many protocols that handle, for instance, exchanges of messages. 
It states that if at least one agent engaging in the protocol is honest, then either the protocol will unfold correctly and fulfill its intended goal for all participants (and this is clearly the most desired outcome), or it will fail for everyone.
In other words, for exchange protocols, it means that either every agent involved in the exchange will send their message and receive the one intended for them, or that nobody will receive any message.
Traditionally, such a fairness property is formalized by expressing it as a \emph{trace} property of the protocol : one must make sure that for {\em every possible execution} of the protocol, the property holds. 
While verifying trace properties is well-understood, formalizing fairness that way imposes a very stringent requirement, as it constrains every single trace or execution, including those associated with irrational and non-credible behavior.
Reasoning on traces under the assumption that these non-credible behaviors will not occur means reasoning on sets of traces (i.e. using hyperproperties), which require a different approach to the verification of these protocols and more elaborated modeling and verification techniques.
Furthermore, imposing such a strong notion of fairness usually requires  the involvement of a \emph{trusted third party} (TTP for short)~\cite{pagnia1999impossibility}.
On the other hand, while a number of general formalisms for hyperproperties exist, restrictions on what behaviors are credible or not are not generally focused on, at least in the context of fairness. Additionally, general algorithms for deciding hyperproperties typically have very high complexity making them somewhat impractical.
Here we therefore restrict ourselves to finding the right notion of acceptable behavior and focus on this notion in order to obtain a reasonable decision procedure.

In this work, we introduce a novel semantic framework based on {\em multi-player games} and an appropriate notion of {\em equilibrium}. This framework allows for reasoning about fair exchange protocols without enforcing the excessively strong fairness constraint based on traces, while still providing guarantees that are relevant and practical. Additionally, we explore the algorithmic tools required to develop verification procedures within this semantic framework.

%a relaxation of this strong fairness constraint, and explore its implications in terms of protocol design and verification algorithms. When  
%We proceed by modeling protocols with \emph{multi-player games on finite graphs}. 
%This allows us to leverage all the tools of game theory, both the conceptual and algorithmic ones.

Relaxation of the fairness property has previously been attempted.
In fact, some previous works already refer to a notion of \emph{rational fairness}. 
See the related works paragraph below~\ref{par:related_works} for a more detailed account of those attempts.
Often, the works in this direction rely on modeling protocols via two-player games, and choosing a solution concept (usually a variant of Nash equilibria) to describe the desired outcomes of the game/protocol and its robustness against attacks. 
In our approach, we assert that protocols should be represented as games with more than two players, allowing for the modeling of malignant coalitions that could potentially derail the protocol. We propose a solution concept that remains robust against such deviations, operating under reasonable (albeit pessimistic) rationality assumptions about the participants.

Before presenting our contributions in further details, let us give a few insights on solution concepts in game theory and their relation to rationality.

\subsection{Rationality in games}
In fair exchange protocols, participants are not purely adversarial, as each agent has an individual goal (e.g., receiving the messages they expect from others) that is not necessarily merely being a nuisance to other players.
Consequently, these protocols cannot be faithfully formalized as \emph{zero-sum} games, and we need to use the more general framework of \emph{non zero-sum} games to model them accurately.
In this setting, each agent has a goal that they pursue through making choices of actions, known as \emph{strategies}, and every outcome of the collective strategies of all agents (a \emph{profile}) results in a potentially different payoff for each player. The central notion for these games is \emph{equilibrium}: a profile of strategies from which it is not worthwhile to deviate.
The notion of equilibrium was first introduced by Nash in the 1950s~\cite{Nas50}, for the case where a single agent deviating would not increase their payoff. Since then, several notions of equilibria have been defined to capture different \emph{models of rationality}: Aumann~\cite{Aumann60} defined Strong Nash Equilibria, where several players deviating conjointly cannot increase their payoff. Other notions include Immune Equilibria~\cite{DBLP:conf/fossacs/Brenguier16}, where players should not be able to decrease the payoff of others by deviating, and Secure Equilibria~\cite{ChatterjeeHenzingerJurdzinski05,DBLP:conf/csl/BruyereMR14}, in which players first aim at maximizing their payoff and, as a secondary objective, try to minimize the payoff of the other players.

Although most of these initial definitions originated in economic theory, game models have also been studied by computer scientists for addressing underlying algorithmic questions: how difficult it is to construct an equilibrium, and how challenging it is to decide the existence of an equilibrium with desirable properties, but also for modeling purposes. In the application context we consider (e.g., fair exchange protocols), the most appropriate instantiation of these game models is games played on directed finite graphs~\cite{GaleStewart53}. Each vertex of a graph represents a state of affairs in the execution of the protocol, and each state belongs to a given player who can choose, based on the sequence of previously visited vertices, how to move to the next vertex. This transition models the execution of one more step in the protocol, thus producing a path that, in turn, models a complete execution of the protocol.
The payoff for each player can be determined by a Boolean condition on this path: does the execution meet the goal of the player? 

This general framework allows us to reason about the rationality of players: if a player has a strategy that enables them to achieve their objective given the strategies of other players, they may have no reason to deviate from it. Therefore, when reasoning about the correctness of a protocol, it should be unnecessary to consider irrational behaviors. A protocol can be deemed adequate if it is correct for rational behaviors, even if irrational behaviors could lead to violations of the strong fairness condition. Our framework also accommodates a degree of pessimism: if a player can deviate to preserve their own goal while harming others, they might choose to do so. Additionally, it allows us to reason about coalitions that agents might form to gain an advantage or cause harm through deviation.

%), making it a qualitative game, or assigned a numerical value as in quantitative games such as mean-payoff~\cite{EhrenfeuchtMycielski79} or discounted payoff games~\cite{Parthasarathy73}.
% We refer the interested reader to the recent survey~\cite{GamesOnGraphs23} for more information about these classes of games.

% Zero-sum (winning strategy) -> Non zero-sum (equilibria) 

% Nash -> Secure / Strong / coalitions 

\subsection{Approach and contributions}

Taking into account the modeling needs induced by exchange protocols with fairness objectives, we introduce the concept of \emph{strong secure equilibria} (SSE) and argue that it is the appropriate solution concept. Indeed, SSE ensures robustness against both coalition deviations (strong equilibria, \cite{Aumann60}) and nuisance-only deviations (secure equilibria, \cite{ChatterjeeHenzingerJurdzinski05}). We propose a modeling approach for protocols, termed \emph{the protocol challenge}, which is expressive enough to encode classical exchange protocols, including those with TTPs and infrastructure as regular players. Additionally, our notion of protocol challenge allows for fine-tuning the levels of rationality and fairness among agents, potentially eliminating the need for a TTP. This is achieved using relatively simple objectives (\emph{i.e.}, Boolean), for which decision procedures can be devised (for any $\omega$-regular objectives).

In our approach, and contrary to most approaches in the literature, the study of fairness protocols does not rely on a predefined attacker model. In our framework, the attacker is implicit and can consist of any coalition of agents, with their power being a parameter of the instance at hand. The types of attacks captured include any deviation that is considered rational under the notion of Strong Secure Equilibria.

Our technical contribution focuses on the case of games played on finite directed graphs with $\omega$-regular objectives (which subsume Linear Temporal Logics objectives, for instance), which are well-suited to the study of fairness protocols.
%\Le{relying on exchanges of messages, as we will show through various examples}.
The finiteness assumption is adequate as the set of messages that can be exchanged in these protocols is usually finite ---~up to nonces or other secret data that can be abstracted by a finite number of messages.
The expressive power of $\omega$-regularity allows to capture a wide range of desired properties on finite or repeated exchanges of messages, as illustrated in various examples throughout this paper.

%\subsubsection*{Contributions}
%In this work, we present a game-based framework for the study of fairness protocols, that does not rely on a predifined attacker model.
%Indeed, in our framework the attacker is implicte and is composed of any coalition of agents, and their power is a parameter of the instance at hand.
%The type of attacks that are captured is any deviation that is considered rational under the notion of the Strong Secure Equilibria. To the best of our knowledge, this notion of equilibrium in non-zero sum multiplayer games that had not been studied before.
%The paper focuses on the case of games on finite graphs, which are well suited to the study of fairness protocols as the set of messages that can be exchanged in these protocols is usually finite ---~up to nonces or other secret data that can be abstracted by a finite number of messages as it is done in other formal models of those protocols. 

We provide decision procedures for determining whether a protocol is immune to nefarious (but rational) attacks from a coalition of participants, and whether such a protocol could exist based on the underlying graph structure and objectives. Most of the computational complexity bounds for our algorithms are tight, and we include an analysis of the case with a fixed number of agents, which is often the dominant factor in the complexity of multiplayer games and which is usually small in practical applications.

\subsection{Related Works}\label{par:related_works}
Several works rely on notions of equilibria to argue for the resilience of the provided protocol, such as the RatFish peer-to-peer protocol~\cite{BackesCiobotaruKrohmer10}, although they rely on 2-player Nash equilibria.
In~\cite{AbrahamDolevGonenHalpern06}, the authors use resilient strong Nash equilibria to show that their protocol does not require a TTP as long as the agents are all rational, but they do not provide a general framework.
A more general formalism to handle rational fairness is laid out in~\cite{ButtyanHubauxCapkun04}, where the authors point out the limit of fairness as a strong trace property.
They also choose the game modeling approach and propose a variant of Nash equilibria as the appropriate solution concept to capture rationality in protocols.
However, their notion of rationality of the agents is \emph{optimistic}: they are not assumed to be willing to deviate from the protocol only to affect negatively the others, it must at the very least impact them strictly positively. 
Furthermore, since the model they choose is strictly two-player, there is no notion of robustness against malignant coalitions. Thus, they favor a sort of \emph{best} Nash equilibria as a rational solution concept.
Finally, while they propose a relevant general framework to model protocols in terms of multiplayer games, they provide no implementations or decision procedure.
% \Ms{\emph{Hic seunt hyperproperties}: cool parce que général avec procédures de décision mais non élémentaire. Lien à faire dans un sens ou l'autre avec notre contribution au paragraphe au dessus.}

At the other end of the spectrum, a very general approach uses \emph{hyperproperties}~\cite{ClarksonSchneider10,RakotonirinaBartheSchneidewind24}.
These allow to express properties of sets of traces, and are therefore well suited to the verification of protocols.
The drawback of such a general approach is that decision problems are often undecidable or non-elementary for relevant fragments; in comparison, the complexity of the most difficult problems in our setting is only $\ExpTime$.

Although a wide array of notions of rationality in games have been defined, few works focus on the notion of secure equilibria.
Two notable exceptions are~\cite{ChatterjeeHenzingerJurdzinski05} which actually also motivates the use of secure equilibria with the need to model malicious agents, and~\cite{DBLP:conf/csl/BruyereMR14}, for quantitative objectives.
Both however only consider the case of two player games.
Two player games are also considered in~\cite{KremerRaskin2003} to evaluate the fairness of an exchange protocol, through the logic ATL.
The authors hint at using coalitions, as they solve several two player games to take into account the fact that the infrastructure may side with one agent or the other to facilitate their cheating.
In contrast, SSEs syntactically consider all possible coalitions.

% \bigskip
% \begin{itemize}
%     \item  Hubeaux~\cite{ButtyanHubauxCapkun04} \Ma{cf plus haut}
    
%     \item Chatterjee~\cite{ChatterjeeHenzingerJurdzinski05}
%     \item peer to peer Ratfish~\cite{BackesCiobotaruKrohmer10}
%     \item \Ms{Un autre papier vu après (cité dans Ratfish je crois): \cite{AbrahamDolevGonenHalpern06}): SNE et résilience pour protocoles sans TTP.}
%     \item Raskin Kremer~\cite{DBLP:journals/jcs/KremerR03}
%     \item \textbf{Hyperproperties}~\cite{RakotonirinaBartheSchneidewind24} \Ma{je ne sais pas s'il faut commencer avec ça dans les related works ?}
    
% \end{itemize}

\section{The challenges of fair exchange}\label{sec:motivatingExample}

The aim of designing a protocol is to establish a set of behavior rules that yield a specific outcome if thoroughly followed, and that are self-enforceable: every participating agent must adhere to these rules, lest they do not get anything out of the interactions.

Before moving on the more formal definitions and results of the paper, let us run through the obstacles and possible solutions to the design a good fair exchange protocol. 
Towards this, let us consider the simple setting of only two agents, Alice and Bob, that wish to exchange messages or currency online. (For instance, Alice could want to purchase digital art from Bob.)
How to ensure that the exchange is fair?
A first naive approach to designing such a protocol could be ``let both send their part to the other one''. 
\ifPagesLimited\relax\else(See \figurename~\ref{fig:ongoingExampleTwoPlayers} in Appendix~\ref{app:first_examples} for a depiction of the possible sequences of actions following this basic guideline.)\fi
If both of them do this, even asynchronously, then the exchange is performed. 
Of course, there could be a malicious third party involved here, that could derail the conduction of the protocol for any reason. 
This is what is traditionally called an \emph{external} attacker, and protocols being robust to external attacks is the subject of many works in the security community.
Another source of trouble could be the failure of the communication channels. 
In this work, we assume that the channels are \emph{resilient}: they deliver every message perfectly, albeit maybe with an arbitrary but finite delay (every message is \emph{eventually} delivered).

However, another critical aspect to take into account is what assumptions we make about the motivations of the agents. 
Indeed, one can argue that there is no reason for an agent to trust the other one. 
In fact, both may even prefer an outcome where they receive something from the other as intended, but do not send their share.
On the other hand, the agents wish in priority to avoid being \emph{wronged}, i.e. sending their share without receiving anything.
Depending on such assumptions, the solutions for designing a fair protocol may vary: clearly, the naive approach sketched above fails as soon as one agent is not fully trustworthy.

Let us briefly describe two approaches to ensure robustness against these possibly untrustworthy agents.

\iffalse  
Consider the problem of the swap of two items between two agents through a postal service.
        The agents performing this exchange may not trust each other, and neither wants to fulfill their end of the deal if they do not obtain the other in exchange.
        On the other hand, getting the other agent's item while still keeping their own is not only acceptable but even preferred (albeit morally dubious).
        
        The normal protocol to follow would be that they both send the item to the other and they should both end up with their end of the bargain fulfilled.
        This however is open to cheating from either party who can not send their item and be rewarded.
        This simple setting can be represented with the graph of \figurename~\ref{fig:ongoingExampleTwoPlayers}.
        From each state each agent can choose to perform actions they have control over: Alice's actions are the solid red arrows while Bob's one are dashed and blue.
        While the intended protocol is to have items exchanged fully, it is clear that the agent having received an item first can just deviate from it and never send the other item; the other agent has no control whatsoever in that situation.

\fi

\subsection{With a trusted third party}\label{sec:motivatingExampleWithTTP}

A first possible solution to ensure fairness while making no assumption on the trustworthiness of the participants is to involve a \emph{trusted third party} (TTP) in the protocol.
This particular agent is assumed to be trustworthy, and to have no other interest than the completion of the exchange. 
Furthermore, it is supposed to have sufficient leverage on the agents to enforce the completion of the exchange, if one of them tries to deflect from the intended course of actions.
The TTP is thus considered to be an agent with absolute \emph{authority}.
One could consider performing the exchange only via the TTP, sending both agents' parts to the TTP that would, upon full reception, redistribute them appropriately. This, however foolproof, is extremely costly, as the TTP is used in every instance of the protocol.
Usually, this is avoided by amending the naive protocol to add the rule ``if an agent does not receive anything while he has sent his share, he contacts and alerts the TTP, who then makes sure the other agent comply''. 
This is what is called an \emph{optimistic} version of a protocol involving a TTP. It is assumed that the participating agents desire the intended outcome of the protocol, and thus that in most traces of it, the TTP will not need to intervene at all.
\ifPagesLimited\relax\else A depiction of a such a TTP moderated protocol exchange can be found in \figurename~\ref{fig:ongoingExampleWithTTP} in Appendix~\ref{app:first_examples}.\fi
This is also the approach of the Zhou-Gollmann protocol that uses a TTP only in case of a (suspected or real) cheating attempt (and is encompassed by our framework, see Section~\ref{sec:ZhouGollmann} later on).

\subsection{Beyond trust: without a trusted third party}\label{sec:motivatingExampleInfinite}

The TTP approach is well-known and widely accepted as a necessity (as well-justified by the impossibility result of~\cite{pagnia1999impossibility}).
What can be done when this approach is not applicable?
Indeed, there are many cases where the trustworthiness of a third party is not clear enough to reasonably rely upon it.
For those cases, can we design protocols which rules are self-enforceable and that eliminate the need of a TTP?
In other terms, what if the incentives to complete the protocol properly were stronger than the ones to deviate from the intended behaviors?
An alternative, and the approach we choose in this work, is to rely, not on the trustworthiness of agents, but instead on their \emph{rationality}.
While not amenable to every possible context where fair exchange is needed, it can provide a viable solution in relevant situations.
We will therefore consider that a protocol is safe if no alliance of agents can wrong another agent \emph{without wronging a member of the alliance itself}. 
The security of the protocol stems from the assumption that all agents, even though they might be malicious, would nonetheless behave rationally.

For instance, in peer-to-peer settings, no one can be absolutely trusted in theory, but all agents have an incentive to be honest, since keeping a good reputation enable them to perform further exchanges.
As another example, a ride-share application can be used as a trusted third party between a driver and a traveler without being trusted \emph{per se}, but because both the driver and the traveler know that if the application had dishonest behaviors in some situations, then no one would trust it anymore, and the company that holds it would stop their profits.

In both these examples, important common features are their \emph{multi-agent} nature and the \emph{repeatability} of the exchange, albeit with different agents. 
To illustrate further and start going towards a more abstract model, consider the following setting, depicted by Figure~\ref{fig:noTTPcircularExchange}.
Alice, Bob and Charlie are three agents who wish to exchange an infinity of items: each agent has infinitely many items to send to both other agents.
Again, as shown in~\cite{pagnia1999impossibility}, there is no protocol that could satisfy the usual definitions of fairness. 
Indeed, there is an agent (say, for instance, Alice) that can \emph{scam} another agent (say Bob), i.e. that can send only finitely many items to Bob and receive at least one more item from him.
That is the case if the agents exchange their items directly, but also if they use the third agent (Charlie) as an intermediary, because Charlie is not a TTP and can therefore choose to help Alice, forming a (malicious) \emph{coalition}, to scam Bob (or the reverse).

\iffalse
In the real world, however, the existence of an agent that can always be absolutely trusted by all the agents is a very strong assumption; but trust can arise from the knowledge of the interests of that party.
A ride-share application, for example, can be used as a trusted third party between a driver and a traveler without being trusted \emph{per se}, but because both the driver and the traveler know that if the application had dishonest behaviors in some situations, then no one would trust it anymore, and the company that holds it would stop their profits.
Similarly, in peer-to-peer settings, no one can be absolutely trusted in theory, but all agents have an incentive to be honest, since keeping a good reputation enable them to perform further exchanges.

Consider, more precisely, the following setting, depicted by Figure~\ref{fig:noTTPcircularExchange}.
Alice, Bob and Charlie are three agents who wish to exchange an infinity of items: each agent has infinitely many items to send to both other agents.
For the same reason as before, in every protocol one could design in that setting, the usual definitions of \emph{fairness} are not satisfied: there is an agent (say, for instance, Alice) that can \emph{scam} another agent (say Bob), i.e. that can send only finitely many items to Bob and receive at least one more item from him.
That is the case if the agents exchange their items directly, but also if they use the third agent (Charlie) as an intermediary, because Charlie is not a TTP and can therefore choose to help Alice scam Bob (or the reverse).
\fi
However, since the process is repeating for further exchanges, if Charlie does so, Bob may react by stopping exchanges not only with Alice, but also with Charlie.
Alice would then be satisfied (she actually scammed Bob), but Charlie would be worse off (she did not scam anyone herself, and can no longer exchange with Bob).
So if Charlie behaves rationally, she has no incentive to help another agent scam the third.
In practice, therefore, such a protocol can be considered as reasonably \emph{safe}: Alice, Bob, and Charlie all take, alternatively, the role of the TTP and enable an exchange of items between the two other agents; and if an agent deviates from the protocol, the other ones stop all exchanges with them.

        % A more subtle approach does not require the third party to be trusted, but only to have her best interest to act honestly as a third party to Alice and Bob.
        % In this setting, Alice and Bob wish to keep on exchanging items $(\iab_n)_{n\in\NN}$ from Alice to Bob and $(\iba_n)_{n\in\NN}$ from Bob to Alice.
        % The third party is a regular agent Charlie ($\charlie$) who also wishes to exchange with Alice and Bob: Alice has a sequence of objects $(\iac_n)_{n\in\NN}$ to give to Charlie in exchange for objects $(\ica_n)_{n\in\NN}$, and Bob has a sequence of objects $(\ibc_n)_{n\in\NN}$ to give to Charlie in exchange for objects $(\icb_n)_{n\in\NN}$.
        % It is assumed they would all rather unduly keep an item in an exchange, but not if that means not being able to continue exchanging with others.
        % The will to continue exchanging is what keeps all agents from trying to cheat others.

        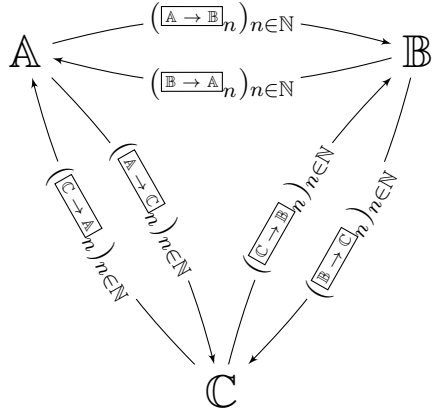
\begin{figure}
            \centering
            \begin{tikzpicture}
                % \tikzstyle{toexchange}=[->,>=latex',draw,bend left=15,decorate,decoration={snake, segment length=7mm, amplitude=2pt,post length=2pt}]
                \tikzstyle{toexchange}=[->,>=latex',draw,bend left=15]
                \tikzstyle{exlab}=[sloped,fill=white,inner sep=2pt]
                \node[font=\Large] (a) at (150:3cm) {$\alice$};
                \node[font=\Large] (b) at (30:3cm) {$\bob$};
                \node[font=\Large] (c) at (-90:3cm) {$\charlie$};

                \path (a) edge[toexchange] node[exlab] {$(\iab_n)_{n\in\NN}$} (b);
                \path (b) edge[toexchange] node[exlab] {$(\iba_n)_{n\in\NN}$} (a);
                \path (a) edge[toexchange] node[exlab] {$(\iac_n)_{n\in\NN}$} (c);
                \path (c) edge[toexchange] node[exlab] {$(\ica_n)_{n\in\NN}$} (a);
                \path (b) edge[toexchange] node[exlab] {$(\ibc_n)_{n\in\NN}$} (c);
                \path (c) edge[toexchange] node[exlab] {$(\icb_n)_{n\in\NN}$} (b);
            \end{tikzpicture}
            \caption{Expected exchanged items between Alice, Bob, and Charlie.}
            \label{fig:noTTPcircularExchange}
        \end{figure}

    \section{General formalisms}\label{sec:formalism}

In this section, we present our model and argue that, in the game-theoretic setting, safe protocols are characterized by strong secure equilibria.
We then provide definitions of games on finite graphs, in order to handle them algorithmically.
In particular, we present techniques to model the infrastructure (\emph{e.g.}, communication channels) as a player among others in these finite games.

        \subsection{Framework}

We consider the following general scenario: several \emph{agents} (for instance: programs, or human users) interact through \emph{behaviors} (for instance: sending some message after having received some other), and aim to satisfy a \emph{correctness condition} (a subset of the possible traces, for instance: exchanging a cryptographic key, or signing a contract).
But, some of those agents may be tempted to deviate from that objective, in order to \emph{wrong} some other agent; as long as they are not themselves wronged.
We thus make the hypothesis that an attack on the protocol comes from the same agent(s) that engage in it.
In other terms, the participating agents are assumed rational but malicious. If we want to model an external attacker in the classical sense, we can by introducing him as an explicit new agent in our framework.
In such a scenario, we wish to design, or to verify, \emph{protocols} that all those agents should follow, so that the correctness condition is achieved, and that no alliance of agents can cheat other agent and wrong them.

In that setting, the \emph{Infrastructure} is seen as an agent among the others, with several possible behaviors but no wrongness condition, so that it can help anyone to wrong anyone.

Before defining protocols, we need a notion of \emph{protocol challenge}, to capture the context in which one wants to define a protocol.

\begin{defi}[Protocol challenge]\label{def:protocolProblemNoInfrastructure}
    A \emph{protocol challenge} is a tuple:
    \[\pc = \left(\agents, (\behaviors_a)_{a \in \agents}, (\Wrong_a)_{a \in \agents}, \Corr\right),\] where:

    \begin{itemize}
\item $\agents$ is the set of \emph{agents};
    
    \item each agent $a \in \agents$ has a set $\behaviors_a\neq \emptyset$ of available \emph{behaviors} (and for $A \subseteq \agents$, we write $\behaviors_A = \prod_{a \in A} \behaviors_a$);

    \item for each agent $a$, the set $\Wrong_a \subseteq \behaviors_\agents$ is the \emph{wronging condition} of agent $a$;

    \item the set $\Corr \in \behaviors_\agents$ is called the \emph{correctness condition}, and is such that for each agent $a$, we have $\Corr \cap \Wrong_a = \emptyset$;

    %\item the agent $\II \in \agents$, called \emph{Infrastructure}, is such that $\Wrong_\II = \emptyset$.
\end{itemize}
\end{defi}

In the rest of this article, we often use the bar $\bbeta_A$ to denote tuples.
The subscript $A$ denotes the set of indices on which the tuple is defined.
It will be omitted when that set is complete and clear from the context --- for example, we will write $\bbeta$ for a tuple of behaviors for all the agents, which we call a \emph{profile} of behaviors.
\\ % change to empty line after removing comment
A \emph{protocol} is, then, a subset of profiles of behaviors.

\begin{defi}[Protocol, implementation]
    A \emph{protocol} is a nonempty set $P \subseteq \behaviors_\agents$.
    An \emph{implementation} of $P$ is an element $\bbeta$ of $P$.
\end{defi}

The reason why a protocol is not a single behavior, but a set, is that protocols are often not defined in a deterministic way: they may be ambiguous about minor aspects, and leave some room for interpretation.
In such a case, all of the possible implementations must be verified.

\begin{exa} \label{ex:protocol}

 Consider an item exchange protocol with a TTP (such as the one in Section~\ref{sec:motivatingExampleWithTTP}), in which we explicitly take into account the fact that when Alice or Bob sends an item, it is not immediately received by the recipient, but spends some time within the \emph{communication infrastructure} ($\II$).
        It is reasonable to assume that this Infrastructure, while offering no guarantee on delays, always delivers the item (this is the standard model of \emph{resilient channels}).
        \ifArxiv The interested reader will find an illustration in Appendix~\ref{app:ongoingExampleWithTTPandInfra}, \figurename~\ref{fig:ongoingExampleWithTTPandInfra}.\fi

 Within the above formalism, there are four agents: Alice, Bob, the TTP, and the Infrastructure $\II$.
The available behaviors for each agent are given by the choice, at any given moment, to perform (or not) any of their own action whenever possible.
The behaviors of the Infrastructure are further restricted by prohibiting behaviors that infinitely choose not to perform any action while infinitely enabled, to model the resilience of the channels.

Alice (or Bob) is wronged if she (he) remains indefinitely in a state where she (he) has neither items.
The TTP is wronged if either Alice or Bob is, which guarantees that both can trust it.
Regarding the correctness condition, it could consist simply in exchanging the items.
However, we choose here to add the condition that the exchange goes on without any intervention of the TTP, as such an intervention is in practice costly and should be used parsimoniously.

An example of a protocol is the set of behaviors defined as the product of the sets of following behaviors:
\begin{description}
    \item[For Alice:] Wait between 1 and 7 time units before sending $\textit{item}_1$. Then wait between 22 and 28 time units.
    If $\textit{item}_2$ has not been received by then, alert the TTP.
    \item[For Bob:] Wait between 3 and 12 time units before sending $\textit{item}_2$. Then wait between 20 and 30 time units.
    If $\textit{item}_1$ has not been received by then, alert the TTP.
    \item[For the TTP:] When alerted, immediately force the other party to send the item.
    \item[For the Infrastructure:] Each received item is transmitted within 6 time units.
\end{description}
One can check that this allows the exchange to be performed, without any TTP intervention, in under 18 time units (in the worst case, 7 + 6 time units for $\textit{item}_1$ to be delivered to Bob, and 12 + 6 for $\textit{item}_2$ to be delivered to Alice), so the correctness condition is achieved for any implementation of this protocol.
Note that that is true only if the Infrastructure behaves as planned by the protocol: in practice, no one controls it, and an item delivery may last (much) more than 6 time units.
In such cases, the correctness condition is no longer guaranteed: the TTP may be called.
But that is acceptable in light of another requirement to protocols, immunity against attacks, which must be guaranteed even if the Infrastructure helps the attacking agents.
\end{exa}

While the usual attacking models are fully pessimistic, i.e. assume that malicious agents may do anything to wrong other agents, our model enable us to mitigate that pessimism by assuming no agent will wrong themself in order to perform an attack.

\begin{defi}[Attack, safe protocol] \label{def:attack}
    An \emph{attack} of the \emph{alliance} $A \subseteq \agents$ to the protocol $P$ is a tuple of behaviors $\balpha$, such that:
    \begin{itemize}
        \item there exists an implementation $\bbeta \in P$ such that, for each $a \in \agents \setminus A$, we have $\beta_a = \alpha_a$ (the agents that do not belong to $A$ follow the protocol $P$);
    
        \item for every $a \in A$, we have $\balpha \not\in \Wrong_a$ (none of the attacking agents is wronged);
        
        \item and there exists $a \not\in A$ such that $\balpha \in \Wrong_a$ (at least one agent is wronged).
    \end{itemize}

    A protocol $P$ is \emph{safe} if there is no attack to $P$.
\end{defi}

\begin{rk}
    If the protocol $P$ is safe, then in particular, the empty alliance has no attack.
    Therefore, no agent is wronged in any implementation $\bbeta \in P$.
\end{rk}

Thus, modeling a concrete situation by a protocol challenge can be done by defining appropriate wronging conditions, depending on the type of attacks one wants to avoid: it defines simultaneously what a successful attack is, and what the agents cannot accept when they are involved in an attack.

In addition to being immune to attacks, a desirable property of a protocol is to conform to the correctness condition.
Therefore, a positive answer to the challenge is defined as follows:
\begin{defi}
    A \emph{positive answer} to a protocol challenge is a safe protocol $P$ such that $P \subseteq \Corr$.
\end{defi}

\begin{exa}\label{ex:protocolWithTTPgood}
    If we consider the protocol described in Example~\ref{ex:protocol}, it appears that every implementation satisfies the correctness condition.
    Moreover, there is no attack to the protocol: the only ways for Alice to satisfy her individual objective are to be helped by the TTP or by Bob; but then, that agent is wronged ---~and symmetrically for Bob.
\end{exa}
\begin{exa}\label{ex:protocolWithoutTTPbad}
    On the other hand, consider again the protocol challenge in the context without the TTP.
    Then, there is no safe protocol satisfying the correctness condition in that case: if the correctness condition is satisfied, then either the alliance $\{\alice\}$ or the alliance $\{\bob\}$ has an attack, by refusing to send the item.
\end{exa}

        \subsection{Safe protocols as strong secure equilibria}

We will now see how such protocol challenges can be modeled as games.
Let us first define games.

\begin{defi}[Game]\label{def:generalGame}
    A \emph{game} is a tuple $\Game = \left(\Pi, (\Sigma_i)_{i \in \Pi}, (\mu_i)_{i \in \Pi}\right)$, where:
    \begin{itemize}
        \item $\Pi$ is a set of \emph{players};

        \item for each player $i \in \Pi$, the set $\Sigma_i \neq \emptyset$ gathers the \emph{strategies} available for player $i$ (when $P \subseteq \Pi$, we also write $\Sigma_P = \prod_{i \in P} \Sigma_i$);

        \item for each player $i$, the function $\mu_i: \Sigma_\Pi \to \RR$ is the \emph{payoff function} of player $i$.
    \end{itemize}
\end{defi}

Intuitively, strategies are functions that map a sequence of past events to the next action a player must choose.
But here, strategies are only abstract objects, as behaviors are in the context of a protocol challenge.
Later, we will define strategies in the context of games played on graphs.

Given a strategy profile $\bsigma$ and a subset of players $C\subseteq \Pi$, $\bsigma_C \in \Sigma_C$ denotes the profile of strategies from $\bsigma$ for players of $C$.
The set $-C$ denotes $\Pi\setminus C$, the set of all players not in $C$.
For two disjoint sets of players $C_1,C_2$ with $C_1 \cap C_2=\emptyset$, two strategy profiles $\bsigma_{C_1} \in \Sigma_{C_1}$ and $\bsigma_{C_2}' \in \Sigma_{C_2}$ can be combined into a profile for all players in $C_1\cup C_2$, written $( \bsigma_{C_1}, \bsigma'_{C_2} ) \in \Sigma_{C_1\cup C_2}$.
In particular, given two profiles $\bsigma,\bsigma'\in\Sigma_\Pi$ and a set of players $C$, $( \bsigma_{-C}, \bsigma'_C )\in \Sigma_\Pi$ is the profile made from strategies of $\bsigma$ for players not in $C$ and in $\bsigma'$ for players in $C$.

In such games, we are interested in a new notion of equilibrium, that does not allow deviations of any coalition (\emph{strong} equilibrium) that would harm another player without harming a member of the coalition (\emph{secure} equilibrium).

\begin{defi}[Strong secure equilibrium]
    Let $\Game$ be a game.
    Let $\bsigma \in \Sigma_\Pi$ be a strategy profile in $\Game$: a \emph{harmful deviation} of the \emph{coalition} $C \subseteq \Pi$ is a strategy profile $\bsigma'_C \in \Sigma_C$ such that for every player $i \in C$, we have $\mu_i( \bsigma_{-C}, \bsigma'_C ) \geq \mu_i( \bsigma )$, and for some player $i \not\in C$, we have $\mu_i( \bsigma_{-C}, \bsigma'_C ) < \mu_i( \bsigma )$.
    A \emph{strong secure equilibrium} in $\Game$, or \emph{SSE} for short, is a strategy profile that admits no harmful deviation.
\end{defi}

That notion captures our notion of safe protocol: in an SSE, players cannot collude to harm another player, as long as we assume that they will not sacrifice themselves to do so.
Formally:

\begin{thm}[Safe protocols and SSEs]\label{thm:linkSNEsafeNoInfrastructure}
    Let $\pc$ be a protocol challenge, and let us consider the game $\Game$ defined as follows: the players are the agents of $\pc$, the strategies available for each player $i$ are the behaviors available for agent $a$ in $\pc$, and the payoff functions are defined as follows:

    $$\mu_a: \bsigma \mapsto \begin{cases}
        0 & \text{if } \bsigma \in \Wrong_a \\
        1 & \text{if } \bsigma \not\in \Wrong_a
    \end{cases}.$$

    Then, the protocol $P$ is safe if and only if each $\bsigma \in P$ is an SSE satisfying $\mu_i(\bsigma) = 1$ for every player~$i$.
\end{thm}

\begin{proof}
    That result will be an immediate consequence of the stronger Theorem~\ref{thm:linkSNEsafeWithSanity}, in Subsection~\ref{ssec:compliant}.
\end{proof}

Note that, beyond the link between safety and SSEs, this theorem also shows that games with Boolean objectives are sufficient to capture our notion of safe protocols.
This (possibly surprising) result can also be seen as a consequence of the following property of SSEs: if the strategy profile $\bsigma$ is an SSE in the game $\Game$, with payoff functions $(\mu_i)_{i \in \Pi}$, then it is also an SSE in the same game where each $\mu_i$ has been replaced by the Boolean payoff function $\mu'_i$ that maps each play $\pi$ to the payoff $1$ if $\mu_i(\pi) \geq \mu_i\< \bsigma \>$, and to the payoff $0$ otherwise.

\subsection{Finite models of protocols}

Modeling and verifying protocols through games is only possible when available behaviors are generated from a finite object; the most frequent underlying game model being \emph{games on graph}, where players navigate a finite graph, the player owning the vertex choosing the next edge.
In what follows, given a set $S$, we write $S^*$ (resp. $S^\omega$) the set of finite (resp. infinite) sequences $s_1\dots s_n$ (resp. $s_1 s_2 \dots$) of elements of $S$.

\begin{defi}[Game on graph]\label{def:gameGraph}
    A \emph{game on graph} is a tuple $\Game=\left(\Pi,V,E,\left(\mu_i\right)_{i\in\Pi}\right)$ where:
    \begin{itemize}
        \item $\Pi$ is a set of players;
        \item $V = \biguplus_{i\in\Pi} V_i$ is the set of vertices, partitioned according to which player owns each vertex;
        \item $E\subseteq V\times V$ is the set of edges, which is assumed to create no sink vertex: $\forall v\in V, \exists v'\in V, (v,v')\in E$;
        \item for each player $i$, the function $\mu_i:V^\omega\rightarrow\RR$ is the payoff function of player $i$.
    \end{itemize}
    A game is initialized by specifying an initial vertex $v_0\in V$.
    Game $\Game$ initialized in $v_0$ is denoted $\Game_{\|v_0}$.
\end{defi}

We write $\mu: V^\omega\rightarrow\RR^\Pi$ the payoff vector function defined by $\mu(\rho)=\left(\mu_i(\rho)\right)_{i\in\Pi}$ for $\rho\in V^\omega$.
A path in the graph $(V, E)$ is called \emph{history} when it is finite, and \emph{play} when it is infinite.

An initialized game on graph is a particular case of a game, as defined in Definition~\ref{def:generalGame}.
Indeed, strategies for a given player $i$ can be defined as mappings $\sigma_i: V^*V_i \to V$ such that $(v,\sigma_i(wv))\in E$ for any $w\in V^*$.
Then, from an initial vertex $v_0$, a strategy profile $\bsigma=\left(\sigma_i\right)_{i\in\Pi}$ (one strategy per player) generates a single infinite path $\< \bsigma\> = \pi_0\pi_1\cdots$ with $\pi_0 = v_0$, such that for every $j\in \NN$, if $\pi_j\in V_i$, then $\pi_{j+1} = \sigma_i(\pi_0\cdots \pi_j)$.
One can then define the payoff $\mu_i(\bsigma)$ as $\mu_i(\< \bsigma \>)$.

% The strategies available for player $i$ is the set of functions $\sigma_i:V^*V_i\rightarrow V$ such that for any $w\in V^*$, $(v,\sigma_i(wv))\in E$.
% From an initial vertex $v_0$, a strategy profile $\bsigma=\left(\sigma_1,\dots,\sigma_{|\Pi|}\right)$ (a strategy for each player) generates a single infinite path $\rho \in V^\omega$: $\rho=v_0v_1\cdots$ such that for every $j\in \NN$, if $v_j\in V_i$, then $v_{j+1} = \sigma_i(v_0\cdots v_j)$.
% Path $\rho$, called the \emph{outcome} of profile $\bsigma$ and written $\rho=\outcome{\bsigma}$ defines the payoff vector: $\mu(\bsigma)=\mu(\outcome{\bsigma})$.
% Hence an initialized game on graph is a particular case of a game.
% \bigskip

\begin{exa}\label{ex:finiteModelTripartiteExchange}
    Recall the item exchange of Section~\ref{sec:motivatingExampleInfinite}.
A blunt representation of the possible states of the system is in means remembering, for each of the three agents, the set of items it currently possesses.
That means an infinite (and even non-enumerable!) set of vertices.
Instead, if we add the assumption that no player $i$ will send the item $\iij_{n+1}$ before they have received both items $\iki_n$, one can model the infinite exchange as a single round of the exchange, repeating infinitely.
Such an assumption is reasonable: indeed, a player $i$ that sends an item before the previous round is finished could be easily scammed.
% \Le{Note that, on the other hand, this reasoning leads us to see this infinite exchange as a succession of finite exchanges, justifying that the wrongness and correctness conditions are $\omega$-regular.}
%
% \Le{Note that this reasoning, in a more general frame, justifies that $\omega$-regular objectives are sufficient for the cases we study: all the infinite-duration protocols we consider can actually be abstracted as repetitions of finite (regular) protocols.}
%
\end{exa}

Some strategies, called \emph{finite memory strategies}, can themselves be encoded by finite objects, called \emph{Mealy machines}.
That will in particular be the case of the strategy profiles representing the protocols we wish to verify.
In such a machine, the states represent the memory of the player.
At each turn, the current vertex in the game is read, firing a transition to a new state; and if the current vertex is controlled by the player using the machine, then the transition outputs a new vertex, to which the player goes.
We do not impose determinism on Mealy machines, since we represent a protocol as a set of strategy profiles and not necessarily as a singleton.
We will see in Section~\ref{sec:algorithms} that non-determinism does not bring additional complexity to the decision problems we study.

\begin{defi}[Mealy machine]
    A \emph{Mealy machine for player $i$} on a game $\Game$ is a tuple $\Mach = (Q, q_0, \Delta)$, where $Q$ is a finite set of \emph{states}, where $q_0 \in Q$ is the \emph{initial state}, and where $\Delta \subseteq (Q \times V_{-i} \times Q) \cup (Q \times V_i \times Q \times V)$ is a finite set of \emph{transitions}, such that for every $(p, u, q, v) \in \Delta$, we have $uv \in E$, and such that for every $p \in Q$ and $u \in V$, there exists a transition $(p, u, q)$ or $(p, u, q, v) \in \Delta$.
        When for each $p$ and $u$, there exists exactly one such transition, we say that the machine $\Mach$ is \emph{deterministic}.
		
		A strategy $\sigma_i$ in $\Game_{\|v_0}$ is \emph{compatible} with $\Mach$ if there exists a mapping $h \mapsto q_h$ that maps every history $h$ in $\Game_{\|v_0}$ to a state $q_h \in Q$, such that for every $hv \in \Hist_{-i} \Game_{\|v_0}$, we have $(q_h, v, q_{hv}) \in \Delta$, and for every $hv \in \Hist_i \Game_{\|v_0}$, we have $(q_h, v, q_{hv}, \sigma_i(hv)) \in \Delta$.
		The set of strategies in $\Game_{\|v_0}$ compatible with $\Mach$ is written $\Comp_{\|v_0}(\Mach)$; these strategies are called \emph{finite-memory}.
		If $\Mach$ is deterministic, then there is exactly one strategy compatible with $\Mach$.

We define analogously Mealy machines capturing strategy profiles for several players, and in particular for the whole set $\Pi$.
In such a case, we simply call it \emph{machine on the game $\Game$}, with no mention of a player.
\end{defi}

\subsection{About compliant infrastructures} \label{ssec:compliant}

The proposed model relies on the assumption that agents are rational and therefore would always prefer to satisfy their individual objective and to avoid being wronged.
However, it is often necessary to consider a special agent whose case is slightly different: the \emph{Infrastructure}, written $\II$, which may do anything and help any player, as long as it satisfies some compliance hypothesis.
For example, it is a common assumption that (neutral) communication channels can delay messages for an arbitrary duration, but cannot drop a message altogether.
In the case of protocol challenges and games specified abstractly, this can be done explicitly: the set $\behaviors_\II$ does not contain the behaviors we choose to exclude.
But when we want to see a protocol challenge as a game played on a finite graph, some behaviors that are assumed to be impossible cannot be excluded by the game structure itself: a finite arena that enables arbitrary delays of messages necessarily contains a path in which they are never delivered.
As a result, such spurious behaviors must be included in the model, but the payoff function must be designed in a way that guarantees that they will never be used.

This is why, although the Infrastructure is not an agent endowed with reason and acting according to its interest, hence we always have $\Wrong_\II = \emptyset$, we define a \emph{compliance hypothesis} that the Infrastructure will always guarantee.

The definition of what constitutes a protocol challenge and an attack can then be updated as follows:

\begin{defi}[Protocol challenge \emph{with compliance hypothesis}]
    A \emph{protocol challenge} is a tuple:
    \[\pc = \left(\agents, (\behaviors_a)_{a \in \agents}, (\Wrong_a)_{a \in \agents}, \Com, \Corr\right),\] where:

    \begin{itemize}
    \item $\left(\agents, (\behaviors_a)_{a \in \agents}, (\Wrong_a)_{a \in \agents}, \Corr\right)$ is a protocol challenge;

    \item the agent \emph{Infrastructure}, written $\II \in \agents$, has an empty wronging condition $\Wrong_\II = \emptyset$;

    \item the \emph{compliance hypothesis} $\Com \subseteq \behaviors_\agents$ is such that the Infrastructure has a least one \emph{compliant behavior}, i.e. a behavior $\beta_\II$ such that for every behavior $\bbeta_{-\II}$, we have $(\beta_\II, \bbeta_{-\II}) \in \Com$.
\end{itemize}
\end{defi}

\begin{defi}[Attack, safe protocol \emph{under compliance hypothesis}]
    An \emph{attack} of the \emph{alliance} $A \subseteq \agents$ to the protocol $P$ in the protocol challenge $\pc$ is a tuple of behaviors $\balpha$, such that:
    \begin{itemize}
        \item $\balpha$ is an attack on $P$, in the sense of Definition~\ref{def:attack}, in the protocol challenge $\left(\agents, (\behaviors_a)_{a \in \agents},\right.$ $\left.(\Wrong_a)_{a \in \agents}, \Corr\right)$;
    
        \item and $\balpha$ satisfies the compliance hypothesis $\Com$.
    \end{itemize}

    The protocol $P$ is \emph{safe} if there is no attack to $P$, and if in every implementation of $P$, the Infrastructure's behavior is compliant.

    A \emph{positive answer} to a protocol challenge with compliance hypothesis is a protocol $P \subseteq \Corr$ that is safe under compliance hypothesis.
\end{defi}

\begin{exa}
    In the item exchange with explicit Infrastructure described in Example~\ref{ex:protocol}, consider the following deviation to our protocol: the Infrastructure delivers $\textit{item}_2$ to Alice, but keeps $\textit{item}_1$ forever.
    Then, Alice gets her individual goal satisfied (without even effectively deviating), and Bob is wronged: even when he calls the TTP, she cannot arrange the situation, since Alice has already sent $\textit{item}_1$.
    However, in such a scenario, the compliance hypothesis is no longer satisfied, hence it is not an attack under compliance hypothesis.
\end{exa}

\begin{rk}\emph{Why is the compliance hypothesis not encoded in the Infrastructure's wronging condition?}    
Intuitively, it may seem a natural abstraction to define the Infrastructure's wronging condition as the negation of its compliance hypothesis: it can then help any agent who tries to attack others, as long as the compliance hypothesis is maintained.
However, the compliance hypothesis cannot fit such a formalism: the wronging condition is what an agent \emph{cannot accept} to help other agents, while non-compliance is what the Infrastructure \emph{cannot do}.
A protocol can perfectly be safe if, in a scenario in which some agent, say Bob, deviates, some other agent, say Alice, accepts to be wronged in order to wrong Bob and, thus, punish him.
There can also be scenarii in which Alice cannot avoid to be wronged.
In contrast, the Infrastructure must always guarantee its compliance hypothesis whatever the other agents do.
\end{rk}

The addition of the compliance hypothesis does not hinder the conversion of the problem into a game.
The payoffs are defined as follows:
\begin{itemize}
    \item for $a \neq \II$:
    $$\mu_i: \bsigma \mapsto \begin{cases}
        0 & \text{if } \bsigma \in \Wrong_a \cap \Com \\
        1 & \text{if } \bsigma \in \overline{\Wrong_a} \cup \overline{\Com}
    \end{cases};$$

    \item for the Infrastructure:
    $$\mu_\II: \begin{cases}
        0 & \text{if } \bsigma \in \overline{\Com} \\
        1 & \text{if } \bsigma \in \Com
    \end{cases}.$$
\end{itemize}

As we will see in the formal proof, the fact that non-compliance cancels the wronging conditions of the non-Infrastructure players will ensure that it is in the best interest of the Infrastructure to play according to a compliant strategy.

\begin{thm}[Safe protocols and SSEs, with compliance hypothesis]\label{thm:linkSNEsafeWithSanity}
    Let $P$ be a protocol in the protocol challenge with compliance hypothesis $\pc$.
    The protocol $P$ is safe if and only if each $\bsigma \in P$ is an SSE satisfying $\mu_i(\bsigma) = 1$ for every player~$i$.
\end{thm}

\begin{proof}~
    \begin{itemize}
        \item \emph{If $P$ is safe under compliance hypothesis.}
        
        Let $\bsigma \in P$.
        From the compliance hypothesis, $\bsigma\in\Com$, so $\mu_\II(\bsigma)=1$.
        Since $P$ has no attack from the empty alliance, we have $\bsigma \not\in \Wrong_a$, i.e. $\mu_a(\bsigma) = 1$, for every agent $a\neq \II$.

        Let us now prove that $\bsigma$ is an SSE.
        Assume that there exists a coalition $C \subseteq \agents$, that has a harmful deviation $\bsigma'_C$.
        Let $\balpha=(\bsigma_{-C},\bsigma_C')$.
        We have:
        
        \begin{itemize}
            \item $\mu_i(\balpha) = 1$ for all $i \in C$;
            
            \item and there exists $j \not\in C$ such that $\mu_i(\balpha) = 0$.
        \end{itemize}

        Then, if we have $\II \in C$, then $\mu_i(\balpha) = 1$, i.e. $\balpha$ satisfies the compliance hypothesis; and if $\II \not\in C$, then $\alpha_\II = \beta_\II$ is a compliant behavior, hence we also have $\balpha \in \Com$ and $\mu_\II(\balpha) = 1$.
        Therefore, we have $j \neq \II$, which means that agent $j$ is wronged in $\balpha$, while no agent belonging to $C$ is: as a consequence, the behavior $\balpha$ is an attack of the alliance $C$ on $P$, which contradicts its safety.

        \item \emph{If every $\bsigma \in P$ is an SSE that yields payoff $1$ to every player.}

        First, let us assume that there exists $\bsigma \in P$ such that $\sigma_\II$ is not a compliant behavior.
        Then, there exists $\bsigma'_{-\II}$ such that $(\bsigma'_{-\II}, \sigma_\II) \not\in \Com$, and therefore such that $\mu_\II(\bsigma'_{-\II}, \sigma_\II) = 0$, and $\mu_i(\bsigma'_{-\II}, \sigma_\II) = 1$ for each $i \neq \II$.
        In other words, the coalition $\agents \setminus \{\II\}$ has a harmful deviation from the strategy profile $\bsigma$, which is therefore not an SSE: contradiction.

        Let us now assume that there is an attack $\balpha$ of the alliance $A \subseteq \agents$ on the protocol $P$.
        Let $\bsigma$ be an implementation of $P$ such that $\balpha_{\agents \setminus A} = \bsigma_{\agents \setminus A}$.
        By assumption, the strategy profile $\bsigma$ is an SSE that gives every player a payoff equal to $1$.
        
        Consider the coalition $A$: for each agent $i \in A \setminus \{\II\}$, we have $\balpha \notin \Wrong_i$, hence $\mu_i(\balpha) = 1$; and since the compliance hypothesis is satisfied in $\balpha$, we also have $\mu_\II(\balpha) = 1$.
        Moreover, there exists $j \not\in A$ such that $\balpha \in \Wrong_j$, i.e. $\mu_j(\balpha) = 0$.
        As a consequence, the strategy profile $\balpha_A$ is a harmful deviation of $\bsigma_A$, which is impossible since $\bsigma$ is an SSE: there is therefore no attack to $P$, which is then safe under compliance hypothesis.\qedhere
    \end{itemize}
\end{proof}

    \section{Modeling message exchanges in games}\label{sec:examples}

In this section, we argue that finite games are, despite their relative simplicity, expressive enough to capture a wide range of fair exchange protocols. We first provide a general technique to encode these problems, and apply it to the Zhou-Gollmann optimistic protocol.
We show in\ifArxiv\ Appendix~\ref{app:BaumWaidner}\else~\cite{arxivVersion}\fi\ that it can also be used to model (several variations of) the Baum-Waidner multi-party contract signing protocol.
\ifArxiv We \else Below, we \fi also provide a more formal description of the tripartite exchange of Section~\ref{sec:motivatingExampleInfinite}, where safety relies on agents' rationality.

\subsection{Protocol challenges for message exchange}

Not every protocol analysis can be effectively translated into a \emph{finite} game on graphs.
One main limitation is protocols that are defined using an a priori unbounded variety of messages.
On the other hand, fair exchange protocols such as contract signing (whether two-party or multi-party), rely on a finite variety of messages: messages that do not conform in structure to what is expected can be discarded (or even trigger the termination of the protocol).
As a result, by treating cryptographic primitives (which are assumed to be perfect) in a symbolic way, one can limit the state-space to encode the set of messages that each participant has already received.

% Protocols are described as the sequence of messages that ought to be exchanged between parties to reach the common goal of the protocol.
% Several assumptions are usually done on these messages, for example \emph{perfect cryptography} is assumed: a message \emph{signature} cannot be faked, so only said agent can create such a message.
% Similarly, ciphered messages are assumed to be only decipherable with the appropriate key.

% The channels on which the messages are exchanged can be modeled as more or less reliable.
% Perfect channels transmit the message immediately, reliable channels may delay messages but will always do so, unreliable channels may drop messages altogether.
% Perfect channels need not appear as agents as they have no choice to perform.
% Reliable and unreliable channels are modeled by the Infrastructure agent, with the wrongness condition specifying how unreliable these can be.

The only actions that an agent can perform is to send a message to another agent.
The set of possible actions (the sending of a messages) is denoted $\actions$.
From the set of messages currently known, agents have the possibility to send all messages they are allowed to based on rules describing the protocol.
When the agent is the Infrastructure, the action is the relaying of a message to its intended recipient; this correspond to the action of said recipient receiving the message.
Recall that in order to assume that the communication channels are reliable, \emph{i.e.} that they cannot delay a message infinitely, it is sufficient to include that in the Infrastructure's compliance hypothesis.

% Because protocols rely on messages being known by agents, the current state of the protocol run is fully determined by which a actions have been performed, that correspond to which messages are known by which agents.
Usually, protocol models are concurrent.
Here, in order to fit in the framework of turn-based graph games, we bypass concurrency by including a round-robin mechanism: agents take turn in sending messages (or just skip their turn); as a result the state space also includes the identity of the active agent.

This can be modeled by a game on graph $2^\actions\times\agents$.
From a state $(v,a)$, the sending of a set of messages $M\subseteq\actions$ from by $a$ means choosing as successor state $(v\cup M,a')$, where $a'$ is the next agent in the round-robin.
Abusively, we say that an agent ``sends messages $x,y,z$'' when $v'\setminus v=\{x,y,z\}$.
It is sometimes useful to refer to the absence of action ($v=v'$) as a special action $\noAct$.
The behavior of each agent is the choice of the message to be emitted when in a given state.

        \subsection{The Zhou-Gollmann optimistic protocol}\label{sec:ZhouGollmann}

The Zhou-Gollmann optimistic protocol~\cite{ZhouGollmann97} allows for the sending of a message $M$ from Alice ($\alice$) to Bob ($\bob$) along with the exchange of \emph{non-repudiation} proofs between them.
The protocol only relies on a \emph{trusted third party} (TTP, $\ttp$) as a recovery process, hence the \emph{optimistic} nature of the protocol.
The TTP is connected to both $\alice$ and $\bob$ by resilient communication channels that ensure immediate transmission of messages.
Alice and Bob are, however, connected through channels that can delay messages for an arbitrarily long time; it is often assumed that messages are always indeed delivered and not delayed forever.
Messages are never altered by the channels.

The protocol is performed as follows.
Alice first sends the message $M$ ciphered with a fresh key $k$, along with the \emph{evidence of origin} of the message: the ciphered message and some metadata, signed with $\alice$'s private key.
This whole message will be denoted as $\EOO$.
Bob replies with an \emph{evidence of receipt} message $\EOR$: the ciphered  $M$ and some metadata, signed with $\bob$'s private key.
Alice then sends the key $k$, again with evidence of origin in the form of a privately signed string containing $k$ and the metadata; altogether this message is denoted $\EOO_k$.
Bob replies with an evidence of receipt of $k$, denoted $\EOR_k$, made of a privately signed $k$ and some metadata.
Remark that having both the ciphered message and the key, Bob is now able to obtain $M$.
The non-repudiation proof of origin ($\NRO$) is made of the evidence of origin of both $M$ and $k$.
Conversely, the non-repudiation proof of receipt ($\NRR$) is made of the evidence of receipt of both $M$ and $k$.

In the case where $\alice$ never receives the $\EOR$, she stops the protocol and the transmission of $M$ is aborted.
If Alice has, however, received the $\EOR$ but not the $\EOR_k$, a situation in which she is at risk of being cheated since Bob now has $M$ and the $\NRO$, Alice can initiate a recovery protocol with the TTP.
In this case she sends key $k$ along with a signed evidence of origin $\Sub_k$ to $\ttp$.
The TTP can then publish $k$ along with a signed confirmation $\Con_k$ that acts as the evidence of receipt for Alice.
This confirmation can also be used by Bob as a confirmation of origin; in the case Alice had initiated the recovery process without first actually sending $\EOO_k$, Bob can still recover it from the TTP.

As the Infrastructure may delay messages, it makes sense for this protocol to have a commonly agreed time after which it is considered that messages that were not received were actually never sent and the protocol should be stopped.
So $\alice$ indicates a deadline $t$ in her metadata: after $t$, the recovery protocol cannot be initialized.
As a result, after $t$, Bob can ask the TTP for the key and $\Con_k$ if he has not received $\EOO_k$ yet.
If the TTP does not have it, then the protocol is aborted.
This means that Alice must initiate the recovery protocol before $t$; the TTP would just refuse to provide $\Con_k$ after $t$.

The rest of the metadata in this protocol allows each participant to identify protocol runs and prevents man-in-the-middle type attacks, but is not relevant to the setting of this work, and will be ignored henceforth.

\bigskip

In our framework, both $\alice$, $\bob$, and $\ttp$ are agents, and $\II$ is the Infrastructure.
The actions available to the agents as the messages they can send, and a special $\noAct$ action that amounts to not doing anything.
When $X\in\{\alice,\bob\}$ sends a message $m$ to the other, $Y$, they actually sends it to $\II$ which relays it (maybe later, if ever) to its recipient.
This is modeled by two actions $m^\s$, a message from $X$ to $\II$, and $m^\r$, from $\II$ to $Y$.
As communication with the TTP is perfect, the Infrastructure is not involved in that case.
Except for $\Con_k$, each message has a given recipient which will not be specified in the corresponding action:
\begin{itemize}
    \item $\alice$ can execute $\EOO^\s$, $\EOO_k^\s$, $\Sub_k$, and $\noAct_\alice$.
    \item $\bob$ can execute $\EOR^\s$, $\EOR_k^\s$, and $\noAct_\bob$.
    \item $\ttp$ can execute $\Con_k^\alice$ (to Alice), $\Con_k^\bob$ (to Bob), and $\noAct_\ttp$.
    \item $\II$ can execute $m^\r$ for $m\in\{\EOO,\EOR,\EOO_k,\EOR_k\}$, and $\noAct_\II$.
\end{itemize}
The rules prevent the sending of a message containing elements unknown to an agent yet.
So the Infrastructure cannot send $m^\r$ before $m^\s$ has occurred.
And as evidence of receipt contains the message received (in some form), Bob cannot execute $\EOR^\s$ (resp. $\EOR_k^\s$) before $\EOO^\r$ (resp. $\EOO_k^\r$) occurred.
Similarly, $\ttp$ cannot execute $\Con_k^\alice$ or $\Con_k^\bob$ before $\Sub_k$.

% The wrongness conditions of the agents are defined as follows.
As a notation, we identify the fact that an action (except $\noAct$s) has been executed with a propositional variable, and collectively denote them as the set $\actions$.
The variable becomes true once the action has been executed, and remains so.
That amounts to define the arena as the turn-based game on vertices $2^\actions\times\{\alice,\bob,\ttp,\II\}$.
\ifPagesLimited
The set of transitions $E$ that ensures the above rules are respected and uses a round-robin to rotate through players is formally given in Appendix~\ref{app:ZhouGollmannTransitions}.
\else

\ifPagesLimited
The set of transitions $E$ is formally defined as follows, from the rules described in Section~\ref{sec:ZhouGollmann} and using a round-robin to rotate through players.
\fi
For $v\in 2^\actions$, we can see it as a predicate: for $x\in\actions$, $v(x) \stackrel{\text{\tiny def}}{=} x\in v$.
% Then $E$ is the largest subset of $(2^\actions,\alice)\times(2^\actions,\bob)\cup (2^\actions,\bob)\times(2^\actions,\ttp)\cup (2^\actions,\ttp)\times(2^\actions,\II)\cup (2^\actions,\II)\times(2^\actions,\alice)$:
$E$ is defined as the largest subset of $\left(2^\actions\times\{\alice,\bob,\ttp,\II\}\right)^2$ such that:
\begin{itemize}
    \item If \((v,X)(v',Y)\in E\) then \((X,Y) \in \{(\alice,\bob),(\bob,\ttp),(\ttp,\II),(\II,\alice)\}\).
    \item If \((v,X)(v',Y) \in E\) then \(v\subseteq v'\).
    \item If \((v,\alice)(v',\bob) \in E\) then \((v'\setminus v) \subseteq\{\EOO^\s,\EOO_k^\s,\Sub_k\}\).
    \item If \((v,\bob)(v',\ttp)\in E\) then \((v'\setminus v) \subseteq\{\EOR^\s,\EOR_k^\s\}\), \(v'(\EOR^\s)\Rightarrow v(\EOO^\r)\), and \(v'(\EOR_k^\s) \Rightarrow v(\EOO_k^\r)\).
    \item If \((v,\ttp)(v',\II)\in E\) then \((v'\setminus v) \subseteq \{\Con_k^\alice,\Con_k^\bob\}\) and \(v'\neq v \Rightarrow \Sub_k \in v\).
    \item If \((v,\II)(v',\alice)\in E\) then \((v'\setminus v) \subseteq\left\{m^\r\mid m\in\{\EOO,\EOO_k,\EOR,\EOR_k\}\right\}\) and for \(m \in \{\EOO,\EOO_k,\EOR,\EOR_k\}\), \(v'(m^\r) \Rightarrow v(m^\s)\).
\end{itemize}
\fi

All conditions can then be expressed as Linear Temporal Logic (LTL) formulas~\cite{Pnueli77}.
This logic allows reasoning over traces by using time modalities.
The \emph{eventually} modality, written $\eventually \varphi$, means that the formula $\varphi$ must be true at some time in the future.
Dually, the \emph{globally} modality, written $\globally \varphi$, means that $\varphi$ must be true at all time from now on.

We first define (boolean) formulas for the non-repudiation proofs:
\begin{mathpar}
\NRO = \EOO^\r \wedge (\EOO_k^\r \vee \Con_k^\bob)
\and
\NRR = \EOR^\r \wedge (\EOR_k^\r \vee \Con_k^\alice)
\end{mathpar}
Then we can define the wrongness condition of Alice and Bob as wanting their non-repudiation proof to be obtained by the other without obtaining a non-repudiation proof themselves.
Otherwise said, Alice is wronged if $\NRO$ is eventually true and $\NRR$ is never (\emph{i.e.} always not) true.
In LTL syntax, this is written as:
\begin{mathpar}
\Wrong_\alice = \eventually \NRO \wedge \globally \neg\NRR
\and
\Wrong_\bob = \eventually \NRR \wedge \globally \neg\NRO
\end{mathpar}
The wrongness condition of the TTP is defined to correspond to situation where at least one of the other agents is wronged, since its goal is to ensure a fair exchange of said proofs:
\[\Wrong_\ttp = \Wrong_\alice \cup \Wrong_\bob\]
The compliance hypothesis models the assumption that the Infrastructure cannot delay messages forever.
\[\Com = \bigwedge_{m\in\{\EOO,\EOR,\EOO_k,\EOR_k\}} \globally (m^\s \Rightarrow \eventually m^\r)\]
Note that there is a way for the Infrastructure to ensure compliance regardless of the behaviors of others, for example by immediately transmitting any message it receives.
Finally, the correctness of the protocol is the full exchange of both non-repudiation proofs:
\[\Corr = \eventually \NRR \wedge \eventually \NRO\]
% Note that for every trace $\tau\in\Com$, the strategy of $\II$ to transmit messages as it does in $\tau$ for every prefix $\tau'$ of $\tau$ and immediately transmit the message in any other case is a compliant behavior. \Ms{Not sure we need this anymore}

Remark that all these conditions can be encoded with relatively simple parity conditions (see Definition~\ref{def:parityCondition} below) over the arena.
For example for $\Wrong_\alice$, give priority $1$ to vertices where $\NRR$ holds, $2$ to other vertices where $\NRO$ holds, and $3$ elsewhere.
In the corresponding game, the winning condition combines the negation of the wrongness condition with the negation of the compliance condition.
As the arena is a directed acyclic graph, this can still be done by a simple parity condition.
Namely, vertices where  $\NRR$ holds or there is a sent message that has not been transmitted get priority $0$, other vertices where $\NRO$ holds have priority $1$, and all remaining vertices get priority $2$.

The protocol as described above where the Infrastructure chooses to transmit messages without delay and Alice chooses a deadline that is $\theta$ time units after the start of the protocol can be modeled by the Mealy machine $\M_\theta$ of \figurename~\ref{fig:mealyZGprotocol}\ifPagesLimited{ in Appendix~\ref{app:ZhouGollmannMealy}}\else.

Transitions in of \figurename~\ref{fig:mealyZGprotocol} are read as follows: $\varphi,X \,|\, a,Y$ means if the game is in a state $(v,X)$ where $v\models \varphi$, then go to state $(v\cup a,Y)$.
To provide a deterministic and complete machine, all conditions $\varphi$ (associated with a given agent) should be mutually exclusive and their disjunction amount to $\top$.
Transitions $m^\s, \II \,|\, \{m^\r\},\alice$ is actually the abbreviation of $16=2^4$ transitions.
For each subset $Z\subseteq\{\EOO,\EOR,\EOO_k,\EOR_k\}$, we have the transition
\[
\left. \bigwedge_{m\in Z} m^\s \wedge \bigwedge_{m\notin Z} \neg m^\s, \II ~\right|~ \{m^\r | m\in Z\},\alice
\]
in the machine.
This means that every message that has been sent must be received.

 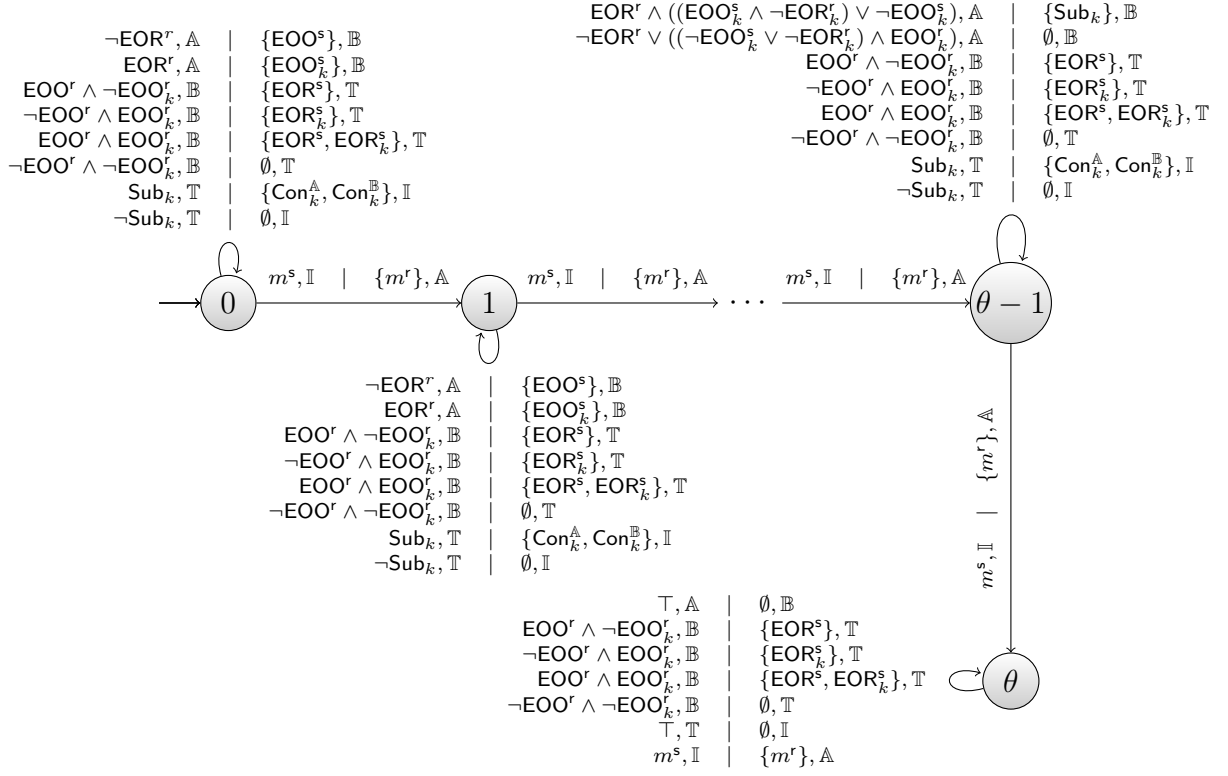
\begin{figure*}
     \centering
     \begin{tikzpicture}[auto,node distance=3.45cm]
     \node[state,initial] (s0) {0};
     \node[state,right of=s0] (s1) {1};
     \node[right of=s1] (dots) {$\cdots$};
     \node[state,right of=dots,inner sep=1pt] (sTauMinus1) {$\theta-1$};
     \node[state,below of=sTauMinus1,node distance=5cm] (sTau) {$\theta$};

     \path[mealyTrans] (s0) edge[loop above] node[xshift=-0.15cm] {%
        \(\begin{array}{rcl}
        \neg\EOR^r,\alice &|& \{\EOO^\s\}, \bob \\
        \EOR^\r, \alice &|& \{\EOO_k^\s\},\bob \\
        \EOO^\r\wedge\neg\EOO_k^\r, \bob &|& \{\EOR^\s\}, \ttp \\
        \neg\EOO^\r\wedge\EOO_k^\r, \bob &|& \{\EOR_k^\s\}, \ttp \\
        \EOO^\r\wedge\EOO_k^\r, \bob &|& \{\EOR^\s,\EOR_k^\s\}, \ttp \\
        \neg\EOO^\r\wedge\neg\EOO_k^\r, \bob &|& \emptyset, \ttp \\
        \Sub_k, \ttp &|& \{\Con_k^\alice, \Con_k^\bob\}, \II \\
        \neg\Sub_k, \ttp &|& \emptyset, \II
        \end{array}\)%
     } (s0);
     \path[mealyTrans] (s0) edge node  {%
        \(\begin{array}{rcl}
        m^\s, \II &|& \{m^\r\},\alice
        \end{array}\)%
     } (s1);
     \path[mealyTrans] (s1) edge[loop below] node[xshift=-0.15cm] {%
        \(\begin{array}{rcl}
        \neg\EOR^r,\alice &|& \{\EOO^\s\}, \bob \\
        \EOR^\r, \alice &|& \{\EOO_k^\s\},\bob \\
        \EOO^\r\wedge\neg\EOO_k^\r, \bob &|& \{\EOR^\s\}, \ttp \\
        \neg\EOO^\r\wedge\EOO_k^\r, \bob &|& \{\EOR_k^\s\}, \ttp \\
        \EOO^\r\wedge\EOO_k^\r, \bob &|& \{\EOR^\s,\EOR_k^\s\}, \ttp \\
        \neg\EOO^\r\wedge\neg\EOO_k^\r, \bob &|& \emptyset, \ttp \\
        \Sub_k, \ttp &|& \{\Con_k^\alice, \Con_k^\bob\}, \II \\
        \neg\Sub_k, \ttp &|& \emptyset, \II
        \end{array}\)%
     } (s1);
     \path[mealyTrans] (s1) edge node  {%
        \(\begin{array}{rcl}
        m^\s, \II &|& \{m^\r\},\alice
        \end{array}\)%
     } (dots);
     \path[mealyTrans] (dots) edge node  {%
        \(\begin{array}{rcl}
        m^\s, \II &|& \{m^\r\},\alice
        \end{array}\)%
     } (sTauMinus1);
     \path[mealyTrans] (sTauMinus1) edge[loop above] node[xshift=-1.57cm] {%
        \(\begin{array}{rcl}
        \EOR^\r\wedge((\EOO_k^\s\wedge\neg\EOR_k^\r) \vee \neg\EOO_k^\s),\alice &|& \{\Sub_k\}, \bob \\
        \neg\EOR^\r\vee((\neg\EOO_k^\s\vee\neg\EOR_k^\r)\wedge\EOO_k^\r),\alice &|& \emptyset, \bob \\
        \EOO^\r\wedge\neg\EOO_k^\r, \bob &|& \{\EOR^\s\}, \ttp \\
        \neg\EOO^\r\wedge\EOO_k^\r, \bob &|& \{\EOR_k^\s\}, \ttp \\
        \EOO^\r\wedge\EOO_k^\r, \bob &|& \{\EOR^\s,\EOR_k^\s\}, \ttp \\
        \neg\EOO^\r\wedge\neg\EOO_k^\r, \bob &|& \emptyset, \ttp \\
        \Sub_k, \ttp &|& \{\Con_k^\alice, \Con_k^\bob\}, \II \\
        \neg\Sub_k, \ttp &|& \emptyset, \II
        \end{array}\)%
     } (sTauMinus1);
     \path[mealyTrans] (sTauMinus1) edge node[anchor=south,rotate=90]  {%
        \(\begin{array}{rcl}
        m^\s, \II &|& \{m^\r\},\alice
        \end{array}\)%
     } (sTau);
     \path[mealyTrans] (sTau) edge[loop left] node {%
        \(\begin{array}{rcl}
        \top,\alice &|& \emptyset, \bob \\
        \EOO^\r\wedge\neg\EOO_k^\r, \bob &|& \{\EOR^\s\}, \ttp \\
        \neg\EOO^\r\wedge\EOO_k^\r, \bob &|& \{\EOR_k^\s\}, \ttp \\
        \EOO^\r\wedge\EOO_k^\r, \bob &|& \{\EOR^\s,\EOR_k^\s\}, \ttp \\
        \neg\EOO^\r\wedge\neg\EOO_k^\r, \bob &|& \emptyset, \ttp \\
        \top, \ttp &|& \emptyset, \II \\
        m^\s, \II &|& \{m^\r\},\alice
        \end{array}\)%
     } (sTau);
     \end{tikzpicture}
     \caption[Mealy machine $\M_\tau$ for the Zhou-Gollmann Protocol]{Mealy machine $\M_\theta$ for the Zhou-Gollmann Protocol. $\theta$ specifies the number of steps each agent is allowed until the deadline $t$ is reached.}
     \label{fig:mealyZGprotocol}
 \end{figure*}

\fi

\begin{rk}[On previous verification efforts] The Zhou-Gollmann optimistic protocol was verified in~\cite{KremerRaskin2003}. The authors considered a list of verification conditions expressed in alternating time temporal logic (ATL)~\cite{DBLP:journals/jacm/AlurHK02}, an extension of branching time logics that allows reasoning about the strategic abilities of agents in alternating transition systems. Those alternating transition systems are equivalent to the game graphs used in our framework. The properties verified in that paper are logically equivalent to ensuring that the Zhou-Gollmann optimistic protocol enforces an SSE. Thus, their verification requirements align with the correctness condition implied by our framework.

However, our method operates at a higher level. We only need to express the correctness condition of the protocol and the wronging conditions for each participating agent. The semantics of the protocol challenge then automatically define the conditions that need to be verified on the game graph, relieving the specifier from reasoning about possible coalitions and their objectives. While our framework provides guarantees equivalent to those obtained through the verification tasks in~\cite{KremerRaskin2003}, it significantly reduces the specifier's burden by eliminating the need for ad hoc reasoning about agent rationality and potential coalitions that could threaten the protocol.
\end{rk}

\subsection{Beyond trust (continued)}\label{sec:motivatingExampleInfiniteFiniteModel}

    Recall the tripartite item exchange of Section~\ref{sec:motivatingExampleInfinite}.
    As explained in Example~\ref{ex:finiteModelTripartiteExchange}, this can be modeled as a finite arena as long as it is assumed that the process takes place in rounds.

In this setting, the goal of each player is still to exchange items infinitely often, and they can punish a player deviating from the predefined protocol by stopping exchanges with them.
Therefore the arena of the game consists of vertices remembering the current ownership of the six different items, regardless of their index.
A player's wronging condition is expressed by not having received an item after sending theirs (being \emph{scammed}) or not being able to continue the exchange when they have not scammed another player.
        Formally, for each player $i$, the wronging condition of $i$ is given by the following LTL formula:
        \begin{mathpar}
            \Wrong_i = \left(\exists j, \Scam_{j,i}\right) \vee (\exists j, \exists k, \neg \eventually \left( i \text{ has } \obj{j}{i}_k \right) \wedge \forall j' \neg \Scam_{i,j'})
            \\\text{where}\and\Scam_{i,j} = \exists k, \forall \ell\leq k, \eventually \left(i \text{ does not have } \obj{j}{i}_\ell \right) \wedge \neg \eventually j \text{ has } \obj{i}{j}_k.
        \end{mathpar}

In the game formalism, determining which players are wronged can be expressed through an automaton reading the infinite sequence of states of the arena.
A player, say Alice, wins by visiting infinitely often a vertex where the exchange of the current round is finished (i.e. when she owns $\iba$ and $\ica$) or when she successfully ``stole'' an item, for example owning $\iab$, $\iba$, $\iac$.
One can go one step further and externalize the winning conditions for each player into a separate automaton, in this case a B\"uchi\footnote{A \B\"uchi automaton accepts an infinite run if it visits at least one accepting state infinitely.} automaton.
The arena can then be reduced to a simple round-robin allowing players to alternatively send an item to any other participant, or skip their turn (see \figurename~\ref{fig:infiniteExchangeArena}).
Each B\"uchi automaton's state remembers what the player currently possesses: a player who sends an item that they do not currently have leads the parity automaton to a sink (not accepting) state $\bot$.
When a player has reached a state where the exchanges for this round have been performed, the automaton loops back to the initial state, meaning the next round is now being played.
The accepting states are the states where a player has finished the round or when they have stolen an item.
This is partially depicted in \figurename~\ref{fig:infiniteExchangeParityAuto}.
This separation does simplify the model as it separates the knowledge of each player, although it does not reduce the actual complexity of the analysis (see Section~\ref{sec:algorithms}).%
\begin{figure*}
    \centering
    \ifArxiv
    \begin{tikzpicture}[node distance=3cm]
        \tikzstyle{neutral}=[state,shape=rounded rectangle, inner xsep=2pt,inner ysep=3pt,draw]
        
        \node[state,player1] (va) at (0,0) {$v_\alice$};
        \node[state,player2] (vb) at (7,0) {$v_\bob$};
        \node[state,player3] (vc) at (3.5,-5) {$v_\charlie$};

        \path[->] (va) edge[bend left=90,distance=4cm] (vb);
        \path[->] (vb) edge[bend left] (vc);
        \path[->] (vc) edge[bend left] (va);

        \node[neutral,above right of=va,anchor=south west] (sABB) {$\alice$ sends $\iab$ to $\bob$};
        \node[neutral,below right of=va,anchor=north west] (sBCC) {$\alice$ sends $\ibc$ to $\charlie$};
            \node[neutral] (sABC) at (barycentric cs:sABB=4,sBCC=1) {$\alice$ sends $\iab$ to $\charlie$};
            \node[neutral] (sACB) at (barycentric cs:sABB=3,sBCC=2) {$\alice$ sends $\iac$ to $\bob$};
            \node[neutral] (sACC) at (barycentric cs:sABB=2,sBCC=3) {$\alice$ sends $\iac$ to $\charlie$};
            \node (sDots) at (barycentric cs:sACC=1,sBCC=1) {$\vdots$};
        \path[->] (va) edge[bend left=30] (sABB);
        \path[->] (sABB) edge[bend left=30] (vb);
        \path[->] (va) edge[bend left=20] (sABC);
        \path[->] (sABC) edge[bend left=20] (vb);
        \path[->] (va) edge[bend left=9] (sACB);
        \path[->] (sACB) edge[bend left=9] (vb);
        \path[->] (va) edge[bend right=9] (sACC);
        \path[->] (sACC) edge[bend right=9] (vb);
        \path[->] (va) edge[bend right=30] (sBCC);
        \path[->] (sBCC) edge[bend right=30] (vb);

        \node[neutral,dashed,shape=rectangle, rounded corners=9pt,anchor=west,text width=3.25cm,text centered,inner sep=5pt] (bobSends) at (7.5,-3) {$\bob$ sends $\obj{i}{j}$ to $k$\\$i \in\{\alice,\bob,\charlie\}$, $j,k \in \{\alice,\charlie\}$\\$i\neq j$};
        \path[->] (vb) edge[bend left] (bobSends);
        \path[->] (bobSends) edge[bend left] (vc);
        
        \node[neutral,dashed,shape=rectangle, rounded corners=9pt,anchor=east,text width=3.25cm,text centered,inner sep=5pt] (charlieSends) at (-0.5,-3) {$\charlie$ sends $\obj{i}{j}$ to $k$\\$i\in\{\alice,\bob,\charlie\}$, $j,k \in \{\alice,\bob\}$\\$i\neq j$};
        \path[->] (vc) edge[bend left] (charlieSends);
        \path[->] (charlieSends) edge[bend left] (va);
        \end{tikzpicture}
    \fi
    \caption{Arena for the exchange of items between three players (partial depiction).}
    \label{fig:infiniteExchangeArena}
\end{figure*}
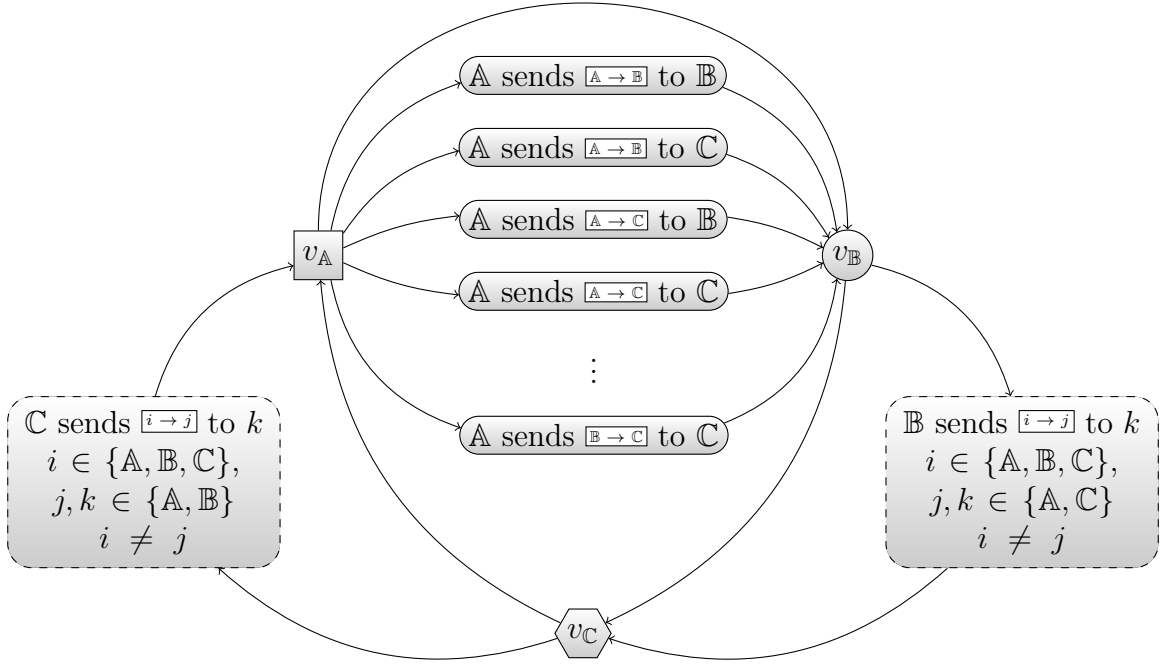

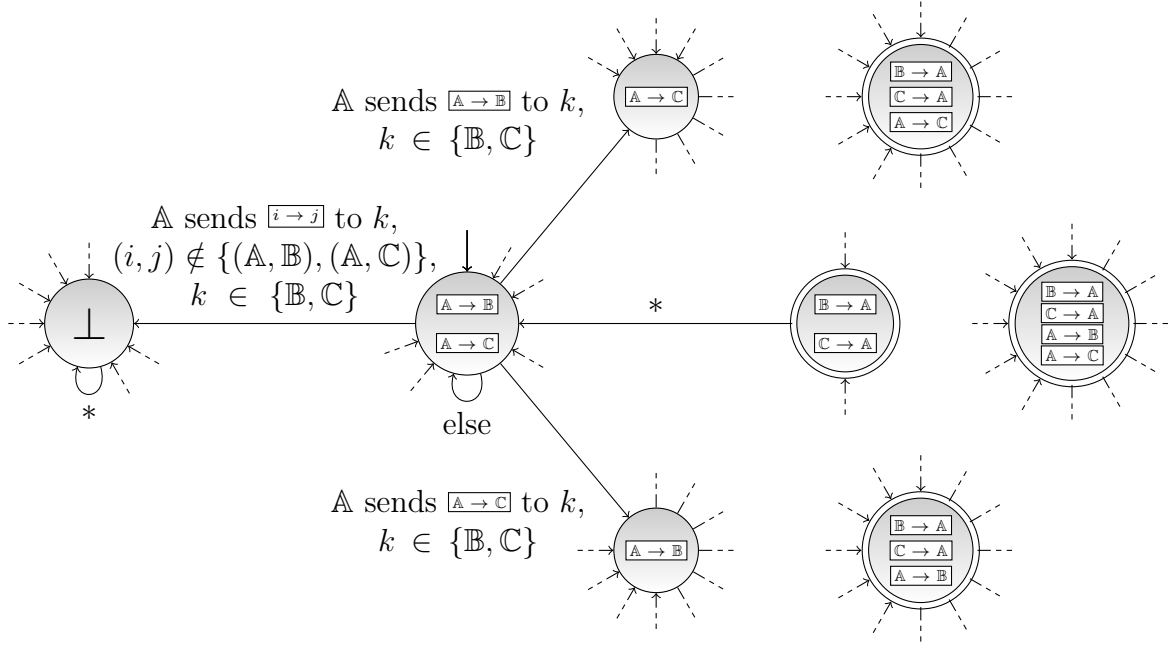
\begin{figure*}
    \centering
    \ifArxiv
    \begin{tikzpicture}[auto]
    \tikzstyle{parState}=[automatonState,text width=1cm,minimum width=1cm,inner sep=1pt,text centered]
    \tikzstyle{accepting}+=[double distance=2pt]
        \node[parState,font=\Large] (fail) at (-5,0) {$\bot$};
        \node[parState,initial,initial above] (ab-ac) at (0,0) {$\iab$\\$\iac$};
        \node[parState] (ac) at (2.5,3) {$\iac$};
        \node[parState] (ab) at (2.5,-3) {$\iab$};
        % \node[parState] (empty) at (5,0) {$\emptyset$};
        \node[parState,accepting] (ba-ca) at (5,0) {$\iba$\\$\ica$};
        \node[parState,accepting] (ba-ca-ac) at (6,3) {$\iba$\\[-5pt]$\ica$\\[-5pt]$\iac$};
        \node[parState,accepting] (ba-ca-ab) at (6,-3) {$\iba$\\[-5pt]$\ica$\\[-5pt]$\iab$};
        \node[parState,accepting] (ba-ca-ab-ac) at (8,0) {$\iba$\\[-6.5pt]$\ica$\\[-6.5pt]$\iab$\\[-6.5pt]$\iac$};

        \path[->] (ab-ac) edge node[swap,text width=4.375cm,text centered] {$\alice$ sends $\obj{i}{j}$ to $k$,\\$(i,j)\notin\{(\alice,\bob),(\alice,\charlie)\}$,\\$k\in\{\bob,\charlie\}$} (fail);
        \path[->] (fail) edge[loop below,distance=0.5cm] node {$*$} (fail);
        \path[->] (ab-ac) edge node[pos=0.75,text width=3.375cm,text centered] {$\alice$ sends $\iab$ to $k$,\\$k\in\{\bob,\charlie\}$} (ac);
        \path[->] (ab-ac) edge node[swap,pos=0.75,text width=3.375cm,text centered] {$\alice$ sends $\iac$ to $k$,\\$k\in\{\bob,\charlie\}$} (ab);
        \path[->] (ab-ac) edge[loop below,distance=0.5cm] node {else} (ab-ac);
        % \path[->] (ab) edge node[swap,pos=0.25,text width=3.375cm,text centered] {$\alice$ sends $\iab$ to $k$,\\$k\in\{\bob,\charlie\}$} (empty);
        % \path[->] (ac) edge node[pos=0.25,text width=3.375cm,text centered] {$\alice$ sends $\iac$ to $k$,\\$k\in\{\bob,\charlie\}$} (empty);
        \path[->] (ba-ca) edge node[swap] {$*$} (ab-ac);

        \outVanishingArrows{ac}{-90,-60,...,30}
        \inVanishingArrows{ac}{60,90,...,150}
        \outVanishingArrows{ab}{90,60,...,-60}
        \inVanishingArrows{ab}{-90,-120,...,-180}
        \inVanishingArrows{fail}{-60,-30,90,120,150,180,210}
        \outVanishingArrows{ba-ca-ac}{-90,-60,...,60}
        \inVanishingArrows[shorten <=1pt]{ba-ca-ac}{90,120,...,240}
        \outVanishingArrows{ba-ca-ab}{-90,-60,...,60}
        \inVanishingArrows[shorten <=1pt]{ba-ca-ab}{90,120,...,240}
        \outVanishingArrows{ba-ca-ab-ac}{-90,-60,...,60}
        \inVanishingArrows[shorten <=1pt]{ba-ca-ab-ac}{90,120,...,240}
        \inVanishingArrows{ab-ac}{200,230,30,-30,60}
        \inVanishingArrows[shorten <=1pt]{ba-ca}{90,-90}
    \end{tikzpicture}
    \fi
    \caption[B\"uchi automaton $\auto_\alice$ in the tripartite exchange of items protocol challenge.]{B\"uchi automaton $\auto_\alice$ in the tripartite exchange of items protocol challenge (partial depiction). Doubly rounded states are accepting.}
    \label{fig:infiniteExchangeParityAuto}
\end{figure*}

\section{Algorithms}\label{sec:algorithms}

In this section, we study decision problems about SSEs for games on graphs.
Therefore, from now on, by \emph{game}, we always mean game on graph.
We focus on the decision problems that arise from the links we have drawn with security protocols.
First, we study the SSE-checking problem, which appears when a given protocol has to be formally verified.
Second, we study the fixed-payoff SSE existence problem, which appears when one wants to know whether a positive answer exist for a given protocol challenge.

Since in all the examples we have seen before, the wrongness conditions of the agents and the compliance hypotheses, and therefore the objectives of the players in the corresponding games, were $\omega$-regular objectives that could be described by LTL formulae, we study those problems in two classes of games with $\omega$-regular objectives: \emph{$\omega$-regular games} in which the objectives are defined by a parity automaton (a canonical and succinct way to represent $\omega$-regular objectives), and \emph{parity games}, in which the objectives are simple enough to be described by labeling directly the vertices of the game instead of the states of an automaton.
We do not consider objectives that are directly described by an LTL formula, because such a representation is so succinct that classical algorithms to handle it (model-checking, synthesis) have a very high complexity, that would hide the complexity of the interesting part of our problems.

We also consider only perfect information games, even though actual protocols are of course often hiding information to agents.
We justify that choice with the following arguments.

\begin{itemize}
    \item The most important hidden or revealed information (typically, a cryptographic key) can be abstracted, and therefore captured by the arena itself, rather than by imperfect information.
    
    \item Protocol safety cannot rely on obscurity.
    Thus, relying on randomized strategies and imperfect information only makes an attack unlikely, instead of impossible.
    Therefore, from the attackers' perspective, information does not matter: one must consider the case where they would chose the right actions to succeed in their attack, be it by chance.
    In other words, they must be assumed to have perfect information.

    \item Thus, not relying on hidden information is only a constraint that the protocols must satisfy: it does not intervene in checking that a protocol is safe.
    The only decision problem where considering imperfect information might be a useful is therefore the SSE existence problem, if one wants to avoid the synthesis of protocols where the agents' actions rely on information they cannot be sure of.
    That is a track for possible future works ---~however, one can already assume that complexities will be much worse.
\end{itemize}

        \subsection{Definition of the problems}

For a given play $\alpha$, we write $\Occ(\alpha)$ (resp. $\Inf(\alpha)$) for the set of vertices that appear (resp. appear infinitely often) in $\alpha$.

\begin{defi}[Parity condition]\label{def:parityCondition}
    Let $(V, E)$, and let $\kappa: V \to \NN$ be a mapping, called \emph{color mapping}.
    The set $\Parity(\kappa)$ is the set of infinite paths $\alpha$ in $(V, E)$, such that $\min_{v \in \Inf(\alpha)} \kappa(v)$ is an even integer.
\end{defi}

\begin{defi}[Parity automaton]
    A {\em parity automaton} over alphabet $X$ is a tuple $\auto=(X,S,s_0,T,\kappa)$ where $S$ is a set of states, $T:S\times X \to S$ is a deterministic transition function, and $\kappa: S \to \NN$ is a color mapping.
    Given a word $w=w_0w_1\cdots\in X^\omega$, the run of $\auto$ on $w$ is $\rho_\auto(w)=s_0 s_1 \cdots$ where for $i\in \NN$, $s_{i+1}=T(s_i,w_i)$.
    The set of infinite words accepted by $\auto$ is $\Lang(\auto)= \left\{w\in X^\omega \mmid \rho_\auto(w)\in\Parity(\kappa)\right\}$.
\end{defi}

It is well-known that the set of languages recognized by parity automata is exactly the set of $\omega$-regular languages.
Now, we consider two classes of games, that fit with the games that are in practice obtained from protocol challenges: games in which the payoff awarded to a player is defined by a parity condition, or by a parity automaton.

\begin{defi}[Parity game]
    A \emph{parity game} is a multi-player game $\Game_{\|v_0}$ such that there exists an integer $k$ and for each player $i$, a color mapping $\kappa_{i}: V \to \{0, \dots, k-1\}$, such that for each play $\pi$, we have $\mu_i(\pi) = 1$ iff $\pi\in \Parity(\kappa_{i})\}$.
\end{defi}

\begin{defi}[$\omega$-regular game]
    An \emph{$\omega$-regular game} is a multi-player game $\Game_{\|v_0}$ such that there exists an integer $k$, and for each player $i$, a parity automaton $\auto_{i}$ over alphabet $V$ using at most $k$ colors, such that  for each play $\pi$, we have $\mu_i(\pi) = 1$ iff $\pi \in \Lang(\auto_{i})$.
\end{defi}

Remark that a parity game is a particular case of an $\omega$-regular game, where all automata $\auto_{i}$ have the same structure: one state $s_v$ per vertex $v$ of the game to remember the last visited vertex.
The transition function is $T(s_v,v')=s_{v'}$.
For each automaton $\auto_{ij}$, the color of $s_v$ is $\kappa_{ij}(v)$.

\begin{exa}
The model for the tripartite exchange of items as described in Section~\ref{sec:motivatingExampleInfiniteFiniteModel} uses B\"uchi automata, which are parity automata using only two colors: $0$ for accepting states and $1$ for other states.

It is possible to add an explicit Infrastructure agent, with the compliance hypothesis stating that every message sent is eventually received (following the model of resilient channels); or, equivalently, every message is infinitely not in waiting (since a message cannot be received if it has not been sent before).
The winning condition for the Infrastructure can therefore be encoded into a parity automaton $\auto_\II$ (of size $2^m\cdot m$ states, where $m$ is the number of different messages) with two priorities $0$ and $1$ (a B\"uchi automaton).
The parity automata for other players have to be modified in order to take into account that a player can win if the compliance hypothesis is not satisfied.
That means, \emph{e.g.} for Alice, combining the co-B\"uchi automaton obtained by negating the winning condition of $\auto_\II$ with the B\"uchi automaton of \figurename~\ref{fig:infiniteExchangeParityAuto}, into a product automaton with parities $0$, $1$, and $2$.
% That can be done by a cross product of $\auto_\alice$ and $\auto_\II$: a state which is accepting in $\auto_\alice$ will have priority $0$ (regardless of the state in $\auto_\II$; a state which is accepting in $\auto_\II$ (thus should be only visited finitely often for $\alice$ to win) gets priority $1$ (as long as it is not accepted with respect to $\auto_\alice$); other states get priority $2$.
\end{exa}

We can now define the decision problems we will study.
First, given a protocol, we want to be able to verify that it is safe --- from a game perspective, that amounts to decide, given the description of a strategy profile, whether all strategy profiles compatible with that description are SSEs.
Two approaches are possible regarding how the strategy profile is described: with one unique machine (monolithic), or one per agent (composite).

\begin{pb}[SSE-checking problem]
    Given a parity game (resp. $\omega$-regular game) $\Game_{\|v_0}$ and a Mealy machine $\Mach$ on $\Game$, is every strategy profile compatible with $\Mach$ an SSE?
\end{pb}

\begin{pb}[Compositional SSE-checking problem]
    Given a parity game (resp. $\omega$-regular game) $\Game_{\|v_0}$ and a family of Mealy machines $(\Mach_i)_{i \in \Pi}$ on $\Game$, is every strategy profile $\bsigma$, with $\sigma_i \in \Comp(\Mach_i)$ for each $i$, an SSE?
\end{pb}

In terms of protocols, solving the checking problem means verifying whether a protocol is safe.
Checking that a protocol is a positive answer to a protocol challenge also requires checking (independently) whether all implementations comply with the correctness condition, which is done through classical automata techniques (and with lower complexities than the problems at hand).
The compositional version can be used when protocols are given by a set of local descriptions of the procedure for each agent.

Later, we will also study a more complex problem: deciding the existence of a safe protocol in a given protocol challenge.
That amounts to search for an SSE with a given outcome in a given game.

\begin{pb}[Fixed-payoff SSE existence problem]
    Given a parity game (resp. $\omega$-regular game) $\Game_{\|v_0}$ and a payoff vector $\bx$, does there exist an SSE $\bsigma$ in $\Game_{\|v_0}$ such that $\mu\< \bsigma \> = \bx$?
\end{pb}

Since we need to ensure that the produced protocol, i.e. the SSE, satisfies the correctness property, one can specify additional constraints on $\bsigma$.
We assume that those constraints are given in the same formalism as the payoffs, i.e. through a parity condition or an automaton.

\begin{itemize}
    \item For a parity game, the correctness condition is given by a coloring mapping $\kappa_\Corr$.
    A path $\alpha$ satisfies $\Corr$, written $\alpha \vDash \Corr$, if $\alpha \in\Parity(\kappa_\Corr)$.
    \item For an $\omega$-regular game, the correctness condition is given by a parity automaton $\auto_\Corr$.
    A path $\alpha$ satisfies $\Corr$, written $\alpha \vDash \Corr$, if $\alpha \in \Lang(\auto_\Corr)$.
\end{itemize}
\begin{pb}[Fixed-payoff SSE constrained existence problem]
    Given a parity game (resp. $\omega$-regular game) $\Game_{\|v_0}$ and a payoff vector $\bx$, does there exist an SSE $\bsigma$ in $\Game_{\|v_0}$ such that $\mu\< \bsigma \> = \bx$ and $\< \bsigma \> \vDash \Corr$?
\end{pb}

The constrained existence problem is a more general case of the existence problem (with trivial parity coloring that goes to $0$ and single-state automaton with color $0$), and does not actually bring any additional complexity to the algorithms: from a complexity point of view, this is akin to adding an extra player; note that this player would need to be treated slightly differently as it does not count as being harmed by a deviation, and therefore we choose to treat it separately in the \ifPagesLimited{construction (see Definition~\ref{def:devGame})}\else{proofs  (Section~\ref{sec:existenceSSE})}\fi.

Solving the constrained existence problem allows to check whether there is a potential positive answer to a protocol challenge.

\ifArxiv\relax\fi

        \subsection{Results}
        
We give tight complexity bounds to all those problems, and also characterize their complexities when fixing the number of players or the number of colors (i.e. the degree of complexity of the wrongness conditions), which is in practice often small compared to the size of the game or of the automaton.
We recall that a problem is \emph{fixed-parameter tractable} ($\FPT$) when the parameter $p$ is fixed if it can be solved in time at most $O\left(f(p) P(n)\right)$, where $f: \NN \to \NN$ is some mapping, $P$ is a polynomial, and $n$ is the size of the instance.
Thus, a $\FPT$ problem is in practice polynomial, if we assume that $p$ remains small.

It is \emph{slice-wise polynomial} ($\XP$) when the parameter $p$ is fixed if it can be solved in time at most $O\left(f(p) n^{g(p)}\right)$, where $f,g: \NN \to \NN$ are some mapping, and $n$ is the size of the instance.
The class of $\XP$ problems strictly contains the $\FPT$ problems.
In between these classes lie the \emph{weft} hierarchy; while $\W[0]=\FPT$, $\W[1]$-hard problems are not thought\footnote{It is unknown whether $W[1]=\FPT$, but $W[1]\neq\FPT \Rightarrow \Poly\neq \NP$.} to be solvable with an $\FPT$ algorithm for said parameter. 

These results are gathered in \tablename~\ref{tab:complexityResults}; all the proofs are available in~\cite{arxivVersion}.
An important additional information is that the lower bounds for the SSE-checking problems all hold even when the given Mealy machines are deterministic --- or, in other words, that non-determinism does not bring additional complexity.

\begin{table*}[ht]
    \centering
    \begin{tabular}{|c|c|c|c|c|c|}\cline{3-6}
       \multicolumn{2}{c|}{} & \multicolumn{2}{c|}{Parity} & \multicolumn{2}{c|}{$\omega$-regular} \\ \cline{3-6}
       \multicolumn{2}{c|}{}                        & Gen. case    & Fixed params & Gen. case    & Fixed params \\ \hline
       \multirow{2}{*}{\ifArxiv\!Checking\!\else Checking\fi}       & Monolithic & $\coNP$-c.   & $n$: $\FPT$              & \multirow{2}{*}{$\PSpace$-c.} & \multirow{2}{*}{$n$: In $\Poly$, $\W[1]$-h.}\\ \cline{2-4}
                                       & Composite  & $\PSpace$-c. & $n$: In $\Poly$, $\W[1]$-h. &                               & \\ \hline
       \multicolumn{2}{|c|}{Fixed-payoff}           & \multirow{2}{*}{$\PSpace$-c.} & $n,k$: $\FPT$       & \multirow{2}{*}{$\ExpTime$-c.} & $n,k$: In $\Poly$, $\XP$-c. \\
       \multicolumn{2}{|c|}{\ifArxiv\!\!(constrained) existence\!\!\else (constrained) existence\fi}              &              & $n=2$: Parity$^\dagger$-h.     &               & $k=2$: $\PSpace$-h. \\ \hline
    \end{tabular}
    \medskip
    
    {\small
    \begin{tabular}{rl}
        $n$:&Number of players \\
        $k$:&Number of colors \\
        $^\dagger$Parity:& 2-players Parity Games (in $\NP\cap \coNP$)
    \end{tabular}}
    % \hfill
    % \emph{Unless otherwise specified, the fixed parameter is $n$.}
    \caption{Complexity results for SSE problems.}
    \label{tab:complexityResults}
\end{table*}
% \ifPagesLimited
% The details are omitted due to space constraints and can be found in Appendix~\ref{app:algorithms}.
% \fi

Let us note in particular that all those problems are decidable in complexity at most $\ExpTime$.
Therefore, we have defined a model that is expressive enough to capture usual protocols and protocol challenges, while opening new possibilities using rationality hypotheses; and that is, simultaneously, simple enough so that the natural decision problems that arise are all decidable in a reasonable complexity, in particular when the number of agents is fixed, as in these cases most of our algorithms execute in polynomial time.

\ifPagesLimited\relax\else\ifPagesLimited
\section{Algorithms: Theorem details and proofs}\label{app:algorithms}
\fi

        \subsection{SSE-checking problems}
            \subsubsection{A tool: the product graph}

Let us first present the method we will use for all the algorithms solving SSE-checking problems: construct the \emph{product graph} of the game and the Mealy machine(s), in order to search for a pair of paths (a \emph{witness pair}) that would correspond to a harmful deviation.

\begin{defi}[Product graph]\label{def:product_graph}
    Consider a graph game $\Game_{\|v_0}$, and a Mealy machine $\Mach$ on $\Game$.
    We define the product $\Game_{\|v_0} \otimes \Mach$ as the graph $(V \times S, E^\times)$ where $(u, s)(v, t) \in E^\times$ if and only if we have $uv \in E$, and if there exists $w \in V$ with $(s, u, t, w) \in \Delta$.
    If $(s, u, t, v) \in \Delta$, then the edge $(u, s)(v, t) \in E$ is called a \emph{non-deviating} edge, and otherwise, it is called an \emph{$i$-deviating} edge, where $i$ is the player controlling the vertex $u$.
    Given a path $\alpha = (u_0, s_0)(u_1, s_1) \dots$ in $\Game_{\|v_0} \otimes \Mach$, we write $\dalpha$ the projection $u_0u_1 \dots$.

    We call \emph{witness pair} in $\Game_{\|v_0} \otimes \Mach$ a pair $(\alpha, \beta)$ of paths, starting from the vertex $(v_0, s_0)$, such that:
        \begin{itemize}
            \item the path $\alpha$ has no deviating edge;

            \item if $k$ is the least index such that $\alpha_k \neq \beta_k$, then the edge $\alpha_{k-1}\beta_k$ is deviating;

            \item there is one player $i$ such that $\mu_i(\dalpha) > \mu_i(\dbeta)$;

            \item for every player $i$ such that $\beta$ contains an $i$-deviating edge, we have $\mu_i(\dalpha) \leq \mu_i(\dbeta)$.
        \end{itemize}
\end{defi}

Note that the player $i$ who has a strictly lower payoff in $\dbeta$ than in $\dalpha$ cannot deviate.
Then, as expected, the existence of a witness pair is equivalent to the existence of a harmful deviation in one of the strategy profiles compatible with the Mealy machine.

\begin{thm} \label{thm:product_graph}
    Let $\Game_{\|v_0}$ be a graph game, and let $\Mach$ be a Mealy machine on $\Game$.
    There exists a strategy profile $\bsigma$ compatible with the machine $\Mach$ that is not an SSE if and only if there exists a witness pair in $\Game_{\|v_0} \otimes \Mach$.
\end{thm}

\begin{proof}~
    \begin{itemize}
        \item Let $(\alpha, \beta)$ be such a witness pair, let $i$ be a player such that $\mu_i(\dalpha) > \mu_i(\dbeta)$, and let $C$ be the set of players $j$ such that $\beta$ has an $j$-deviating edge.
        Let us write $\alpha = (\pi_0, s_0) (\pi_1, s_1) \dots$ and $\beta = (\chi_0, t_0) (\chi_1, t_1) \dots$.
        Let $k$ be the least index such that $\alpha_k \neq \beta_k$.
        Let us define a strategy profile $\bsigma \in \Comp(\Mach)$ by defining the mapping $h \mapsto q_h$ as follows: for each $\l$, we define $q_{\dalpha_{\leq \l}} = s_\l$, and $q_{\dbeta_{\leq \l}} = t_\l$.
        We define the other states $q_h$ arbitrarily.

        Then, we have $\dalpha = \< \bsigma \>$, and $\dbeta$ is compatible with the strategy profile $\bsigma_{-C}$.
        Moreover, we have $\mu_C(\dbeta) \geq \mu_C(\dalpha)$, hence the coalition $C$ has a harmful deviation in the strategy profile $\bsigma$.

        \item Conversely, if a given coalition $C$ has a harmful deviation $\bsigma'_C$ to a strategy profile $\bsigma \in \Comp(\Mach)$, then let $h \mapsto q_h$ be the mapping corresponding to $\bsigma$.
        Let us define $\alpha = (\pi_0, q_{\pi_0}) (\pi_1, q_{\pi_0\pi_1}) (\pi_2, q_{\pi_0\pi_1\pi_2}) \dots$, and $\beta = (\chi_0, q_{\chi_0}) (\chi_1, q_{\chi_0\chi_1}) \dots$, where $\pi = \< \bsigma \>$ and $\chi = \< \bsigma_{-C}, \bsigma'_C \>$.
        Then, since $\bsigma'_C$ is a harmful deviation to $\bsigma$, there is a player $i \in C$ that has a strictly worse payoff in $\dalpha = \pi$ than in $\dbeta = \chi$, and every player $j$ such that there is a $j$-deviating edge along $\beta$ belongs to $C$, and has therefore at least the same payoff in $\dbeta$ than in $\dalpha$.
        The path $\alpha$ has also no deviating edge, and the path $\beta$ gets separated from it with a deviating edge: the paths $\alpha$ and $\beta$ form a witness pair.
        \qedhere
    \end{itemize}
\end{proof}

            \subsubsection{Parity games}

Using that preliminary result, we first provide an optimal algorithm in the case of parity games.
That algorithm consists in guessing and checking the optimal pair in polynomial time.

\begin{thm}\label{thm:SSEcheckingParity}
    In parity games, the SSE-checking problem is $\coNP$-complete, even if the Mealy machine is deterministic. 
    When the number of players is fixed, then this problem is fixed-parameter tractable.
\end{thm}

\begin{proof}~

    \begin{itemize}
        \item \emph{Fixed-parameter tractability.}

        By Theorem~\ref{thm:product_graph}, this problem reduces to deciding the non-existence of a witness pair in the graph $\Game_{\|v_0} \otimes \Mach$.
        We show that when the number of players is fixed, this amounts to solving a polynomial number of Streett emptiness problems.

        Let us first recall that a Streett condition is defined by a finite set $(R_j, G_j)_{j \in J}$ of pairs of subsets $R_j, G_j \subseteq V$, such that a play or path $\pi$ satisfies the condition $\pi \in \Streett\left((R_j, G_j)_j\right)$ if and only if for every $j$, if the set $R_j$ is visited infinitely often along $\pi$, then so is the set $G_j$.

        Let us now remark that when the number $n$ of players is fixed, there is a fixed number of possible payoff vectors (namely $2^n$).
        Similarly, there is also a fixed number of possible coalitions that can deviate in a witness pair (namely $2^n-2$ as the empty coalition and the coalition of all players cannot yield a witness pair).
        The algorithm enumerates all these possible payoff vectors and coalitions.

        For $F \subseteq \Pi$ let $\left[\Game_{\|v_0} \otimes \Mach\right]_C$ be the game $\Game_{\|v_0} \otimes \Mach$ where all $i$-deviating edges for $i \notin C$ and all unreachable states have been removed: that is the game where all players in $\Pi\setminus C$ remain faithful to strategies provided by $\Mach$.
        In particular, in $\left[\Game_{\|v_0} \otimes \Mach\right]_\emptyset$ there are no deviating edges, so in any witness pair the first component $\alpha$ is a path in $\left[\Game_{\|v_0} \otimes \Mach\right]_\emptyset$.
        And for a path $\beta$ where $C$ is the set of players $i$ such that $\beta$ contains an $i$-deviating edge, $\beta$ is a path in $\left[\Game_{\|v_0} \otimes \Mach\right]_{C}$.

        To check whether a payoff $\bx$ can be achieved when only players of $C$ can deviate from the strategies of $\Mach$, one can check the emptiness in $\left[\Game_{\|v_0} \otimes \Mach\right]_{C}$ of the language of infinite words specified by the Streett acceptance condition built as follows.
        Let $k$ be the maximal color appearing in the codomain of the coloring mappings.
        We start with an empty set of Streett pairs.
        \begin{itemize}
        \item For a player $i$ such that $x_i = 1$ (i.e. $i$ wins), and for $0\leq j\leq \lfloor\frac{k}{2}\rfloor$ add the Street pair $(R_j,G_j)$ defined as:
        \begin{itemize}
            \item $R_j= \left\{v \mmid \kappa_i(v)=2j+1\right\}$
            \item $G_j= \left\{v \mmid \exists j'\leq j, \kappa_i(v)=2j'\right\}$
        \end{itemize}
        \item For a player $i$ such that $x_i=0$ (i.e. $i$ loses), and for $0\leq j\leq \lceil\frac{k}{2}\rceil$ add the Street pair $(R_j,G_j)$ defined as:
        \begin{itemize}
            \item $R_j= \left\{v \mmid \kappa_i(v)=2j\right\}$
            \item $G_j= \left\{v \mmid \exists j'\leq j, \kappa_i(v)=2j'-1\right\}$
        \end{itemize}
        \end{itemize}
        This Streett condition simply contains the conjunction of the underlying parity conditions (or their negation).
        There are $n\cdot k$ pairs in this condition, and each pair is of size at most $|\Game|\cdot |\Mach|$ (as a vertex can only belong in one set for each pair).
        The emptiness check of such a Street automaton can be performed in time $O(n\cdot k\cdot |\Game|^2\cdot |\Mach|^2)$~\cite[Theorem~5.2]{LatvalaHeljanko99}.
        
        For a given vector $\bx$ and a coalition $\emptyset \subsetneq C \subsetneq \Pi$, let $\bx^{\uparrow_C}$ be the set of payoff vectors $\by$ such that for all $i\in C$, we have $y_i \geq x_i$ and for some $j\notin C$, we have $y_j < x_j$.
        We perform emptiness check for payoff vector $\bx$ in $\left[\Game_{\|v_0} \otimes \Mach\right]_\emptyset$ and for each $\by \in \bx^{\uparrow_C}$, we perform emptiness check for payoff $\by$ in $\left[\Game_{\|v_0} \otimes \Mach\right]_C$.
        If both these checks fail, i.e. if such $\bx$ and $\by$ can be achieved, then there is a witness pair and $\Mach$ does define profiles that are not SSE.

        The overall complexity is in $O(2^{2n}\cdot n\cdot k \cdot |\Game|^2 \cdot |\Mach|^2)$.
        It is therefore fixed-parameter tractable in the number of players.
        
        \item \emph{$\coNP$-easiness.}

        We now prove that if a witness pair $(\alpha, \beta)$ exists in the graph $\Game_{\|v_0} \otimes \Mach$, then there also exists a witness pair $(hc^\omega, h'd^\omega)$ of lassos whose sizes are polynomial in the size of $\Game_{\|v_0} \otimes \Mach$, and therefore in the size of the input.
        Let $n$ be the number of vertices in the graph $\Game_{\|v_0} \otimes \Mach$.

        Indeed, given the witness pair $(\alpha, \beta)$, let $\ell$ be the least index such that $\alpha_\ell \neq \beta_\ell$.
        The vertex $\alpha_{\ell-1}$ is accessible from $\alpha_0$: therefore, there exists a path $h_0$ from $\alpha_0$ to $\alpha_{\ell-1}$, with no deviating edge, that has length at most $n$.
        Then, let $m > \ell$ be an index such that every vertex in $\alpha_{\geq m}$ is visited infinitely often by $\alpha$.
        In particular, it is the case of the vertex $\alpha_m$, and there exists an index $m'$ such that $\alpha_{m'} = \alpha_m$, and $\Occ(\alpha_m \dots \alpha_{m' - 1}) = \Inf(\alpha)$.
        There is, therefore, a path from $\alpha_m$ to itself that visits all the vertices of the set $\Inf(\alpha)$ and only them, and that uses no deviating edges: therefore, by removing the cycles along that path, there exists a path $c$ that satisfies the same properties and that has length at most $n^2$.
        Moreover, there is a path $h_1$ from $\alpha_\ell$ to $\alpha_m$ with no deviating edges of size at most $n$.
        Then, the lasso $\alpha' = h_0 h_1 c^\omega$ is such that $\mu(\dalpha') = \mu(\dalpha)$, and it size is at most $2n + n^2$.
        We define analogously a lasso $\beta' = h_0 h_2 d^\omega$ such that $\mu(\dbeta') = \mu(\dbeta)$, and such that for every $i$, if there is an $i$-deviating edge in $\beta'$, then there is an $i$-deviating edge in $\beta$.
        Then, the pair $(\alpha', \beta')$ is a witness pair.

        Therefore, such a pair can be guessed and checked in polynomial time: negative instances of the SSE-checking problem can be recognized in non-deterministic polynomial time.

        \item \emph{$\coNP$-hardness.}
        
        We encode the negation of the satisfiability problem.
        Namely, for a propositional formula $\varphi$, we build a game $\Game$ and a machine $\Mach$ such that $\varphi$ is satisfiable if and only if there is a profile in $\Mach$ that is not an SSE.

        Let $\varphi=\bigwedge_{i=1}^{m'} \bigvee_{j=1}^m \l_{ij}$ with $\l_{ij}\in \textit{Lit}=\{x_1,\neg x_1,\dots,x_n,\neg x_n\}$ be a propositional formula.
        We build a game with $2n+1$ players: one player called \emph{Solver}, and written $\SS$, who tries to find a valuation, and the \emph{literal players}, i.e. $x_1, \neg x_1, \dots, x_n, \neg x_n$.
        That game is depicted in \figurename~\ref{fig:sat2sseCheck}.
        Solver controls circular vertices, while rectangular vertices are controlled by the player identified in the bottom right corner (if any; vertices with a single successor could belong to anyone).
        It is made of two parts: in the first part Solver chooses a valuation by selecting vertices (the \emph{setting} module), and in the second part the formula is checked by traversing all clauses and selecting one literal per clause (the \emph{checking} module).
        Note that each litteral player can, when given the opportunity to play, choose to exit a sink vertex $\blacktriangledown$ (duplicated for clarity); in that case the checking module is never reached.
        
        The coloring functions are defined as follows:
        \begin{itemize}
            \item For $\ell\in\textit{Lit}$:
            \begin{itemize}
                \item $\kappa_{\ell}(\blacktriangledown)=0$;
                \item for every $\l_{ij}=\neg\ell$ in the checking module, $\kappa_\ell(\l_{ij})=1$;
                \item $\kappa_\ell(\checkmark)=2$.
                \item for any other vertex $v$, $\kappa_\ell(v)=3$.
            \end{itemize}
            \item For Solver: for any vertex $v$, $\kappa_\SS(v)=0$.
        \end{itemize}
        So a player $\ell\in\textit{Lit}$ can achieve payoff $1$ by either reaching $\blacktriangledown$ or avoiding infinite visits to its negation in the checking module.
        On the other hand, Solver always achieves $1$.

        Let $\Mach_{\textit{Lit}}$ be the deterministic Mealy machine with a single state that encodes the strategy profile that has all players in $\textit{Lit}$ go to $\blacktriangledown$, and Solver do arbitrary actions.
        Let $\bsigma$ be a profile defined by this machine.
        Note that with $\bsigma$ all players get payoff $1$.
        We show that $\bsigma$ is an SSE if and only if $\varphi$ is not satisfiable.
        \medskip

        \begin{description}
        \item[Assume that $\varphi$ is satisfiable.]
        Then we construct a deviation to $\bsigma$.
        Let $\zeta$ be a valuation such that $\zeta \vDash \varphi$.
        Let $C$ be the coalition of players made of Solver and players selected by $\zeta$ as follows: if $\zeta(x_i)=\top$ then $x_i\in C$, if $\zeta(x_i)=\bot$ then $\neg x_i\in C$.
        Note that this coalition contains exactly one literal per variable.
        Let $\bsigma'$ be the profile that deviates from $\bsigma$ for players in $C$ as follows:
        \begin{itemize}
            \item for $x_i$ (resp. $\neg x_i$) $\in \textit{Lit} \cap C$, from $x_i^s$ (resp. $\neg x_i^s$) choose $?x_{i+1}$ (or $C_1$ if $i=n$)
            \item from $?x_i$, Solver chooses $x_i^s$ if $\zeta(x_i)=\top$ and $\neg x_i^s$ if $\zeta(x_i)=\bot$.
            \item from $C_i$, Solver chooses a literal that is $\top$ in $\zeta$ (since $\zeta\vDash \varphi$, this is always possible).
        \end{itemize}
        Note that with this profile no player outside of $C$ actually plays, so the $\blacktriangledown$ vertex is not reached.
        In addition, no negation of a literal in $C$ is ever reached in the checking module, so all players in $C$ still get payoff $1$.
        Consider a literal vertex $\l$ visited infinitely in the checking module (this assumes that $\varphi$ is not an empty formula).
        This literal being $\top$ in $\zeta$, its negation must be $\bot$ and therefore correspond to a player $\ell$ not in coalition $C$.
        This player will see a priority $1$ infinitely often and therefore lose, obtaining payoff $0$, which is strictly less than what was obtained with $\bsigma$.
        Therefore the deviation $\bsigma'$ of coalition $C$ is harmful and $\bsigma$ is not an SSE.

        \item[Now assume that $\varphi$ is not satisfiable.]
        Let $C$ be a coalition of players and $\bsigma'$ a profile that deviates for players in $C$.
        Since Solver cannot decrease her payoff, any coalition without Solver can also include her.
        So we can assume that $\SS\in C$.
        Furthermore, if for a variable $x_i$, neither player $x_i$ nor $\neg x_i$ are in $C$, then the play reaches $\blacktriangledown$ and the payoff remains the same for all players.
        So for the deviation to be harmful, at least one player per variable must be in $C$, and the corresponding vertex must be reached in the setting module (by decision of Solver) and they must actually deviate to avoid $\blacktriangledown$.
        Now since $\varphi$ is not satisfiable, in every loop of the checking module there must be at least a literal and its negation that are visited, so for some $i$ both $x_i$ and $\neg x_i$ see a color $1$ in this loop.
        As there are infinitely many turns of loops in the checking module, there is an index $i$ for which both players $x_i$ and $\neg x_i$ get payoff $0$.
        Since (at least) one of them is in the coalition $C$, the deviation would harm a deviating player, therefore $\bsigma$ is an SSE.\qedhere
        \end{description}
    \end{itemize}
\end{proof}
\begin{figure*}
    \centering
    
    \begin{subcaptionblock}{\textwidth}
    \centering
    \begin{tikzpicture}[node distance=50pt]
        % Setting the valuation
        \node[solverState,initial] at (0,0) (q1) {$?x_1$};
        \node[literalState=$x_1$,above right of=q1] (p1) {$x_1^s$};
        \node[literalState=$\neg x_1$,below right of=q1] (n1) {$\neg x_1^s$};
        \node[solverState,below right of=p1] (q2) {$?x_2$};
        \node[literalState=$x_2$,above right of=q2] (p2) {$x_2^s$};
        \node[literalState=$\neg x_2$,below right of=q2] (n2) {$\neg x_2^s$};
        \node[below right of=p2] (q3) {$\cdots$};
        \node[literalState=$x_n$,above right of=q3] (pn) {$x_n^s$};
        \node[literalState=$\neg x_n$,below right of=q3] (nn) {$\neg x_n^s$};

        % Checking all the clauses
        \node[solverState,below right of=pn] (c1) {$C_1$};
        \node[nobodyState,above right of=c1] (l11) {$\l_{1,1}$}; %\cobuchiLabel[inner sep=0pt,cloud puff arc=95]{l11}{$\neg \l_{1,1}$}
        \node[nobodyState,below right of=c1] (l1m1) {$\l_{1,m_1}$};  %\cobuchiLabel[inner sep=0pt,cloud puff arc=95]{l1m1}{$\neg \l_{1,m_1}$}
        \node (c1dots) at (barycentric cs:l11=1,l1m1=1) {$\myvdots$};
        \node[solverState,below right of=l11] (c2) {$C_2$};
        \node[above right of=c2] (l21) {$\myiddots$};
        \node[below right of=c2] (l2m2) {$\myddots$};
        \node (c2dots) at (barycentric cs:l21=1,l2m2=1) {$\myvdots$};
        \node[right of=l21,node distance=2.5cm] (lp1) {$\myddots$};
        \node[right of=l2m2,node distance=2.5cm] (lpmp) {$\myiddots$};
        \node (cpdots) at (barycentric cs:lp1=1,lpmp=1) {$\myvdots$};
        \node[yshift=7pt] (ldots1) at (barycentric cs:l21=1,lp1=1) {$\cdots$};
        \node[yshift=-7pt] (ldotsm) at (barycentric cs:l2m2=1,lpmp=1) {$\cdots$};
        \node (ldots) at (barycentric cs:ldots1=1,ldotsm=1) {$\cdots$};
        \node[nobodyState,below right of=lp1] (ph) {$\checkmark$}; %\cobuchiLabel[fill=green!80!black]{ph}{$\SS$}

        % Sink state
        \node[nobodyState,above of=p2] (sinkTop) {$\blacktriangledown$}; %\cobuchiLabelTop{sinkTop}{$\SS$}
        \node[nobodyState,below of=n2] (sinkBot) {$\blacktriangledown$}; %\cobuchiLabelBot{sinkBot}{$\SS$}

        \begin{pgfonlayer}{arrows}
            \path[->] (q1) edge (p1) edge (n1);
            \path[->] (p1) edge (q2) edge[bend left,out=-30] (sinkTop);
            \path[->] (n1) edge (q2) edge[bend right,out=30] (sinkBot);
            \path[->] (q2) edge (p2) edge (n2);
            \path[->] (p2) edge (q3) edge (sinkTop);
            \path[->] (n2) edge (q3) edge (sinkBot);
            \path[->] (q3) edge (pn) edge (nn);
            \path[->] (pn) edge (c1) edge[bend right,out=30] (sinkTop);
            \path[->] (nn) edge (c1) edge[bend left,out=-30] (sinkBot);
    
            \path[->] (c1) edge (l11) edge (c1dots) edge (l1m1);
            \path[<-] (c2) edge (l11) edge (c1dots) edge (l1m1);
            \path[->] (c2) edge (l21) edge (c2dots) edge (l2m2);
            \path[<-] (ph) edge (lp1) edge (cpdots) edge (lpmp);

            \path[->] (ph) edge[out=-90,in=-90,distance=4cm] (c1);
    
            \path[->] (sinkTop) edge[loop above] (sinkTop);
            \path[->] (sinkBot) edge[loop below] (sinkBot);
            \end{pgfonlayer}
    \end{tikzpicture}
    \caption{Game $\Game_\varphi$ to encode the satisfiability of propositional formula $\varphi$ into a game.}
    \label{fig:sat2sseCheckArena}
    \end{subcaptionblock}
    
    \begin{subcaptionblock}{0.4\textwidth}
    \centering
    \begin{tikzpicture}
        \node[state,initial] (q) {};
        \path[->] (q) edge [loop above] node {$x_i^s | \blacktriangledown$} (q);
        \path[->] (q) edge [loop below] node {$\neg x_i^s | \blacktriangledown$} (q);
        \path[->] (q) edge [loop right] node {else$| \ast$} (q);
    \end{tikzpicture}
    \caption{Mealy machine $\Mach_{\textit{Lit}}$.}
    \label{fig:sat2sseCheckMachine}
    \end{subcaptionblock}
    \caption{Encoding of SAT into an MVP SSE-checking problem.}
    \label{fig:sat2sseCheck}
\end{figure*}

        \subsubsection{\texorpdfstring{$\omega$}{ω}-regular games}

In the case of $\omega$-regular games, the existence of a witness pair with a polynomial description is no longer guaranteed in negative instances.
The witness pair must therefore be guessed step by step, yielding a $\PSpace$ algorithm.

\begin{thm}\label{thm:SSEcheckingOmegaRegular}
    In $\omega$-regular games, the SSE-checking problem is $\PSpace$-complete, even when the Mealy machine is deterministic.
    When the number of players is fixed, the problem can be solved in polynomial time but is $\W[1]$-hard.
\end{thm}

\begin{proof}~
    \begin{itemize}
        \item \emph{$\PSpace$-easiness.}

        We give a non-deterministic algorithm that decides the SSE-checking problem using polynomial space.
        By Theorem~\ref{thm:product_graph}, that problem reduces to search for a witness pair of paths in the product graph $\Game_{\|v_0} \otimes \Mach$.
        If such a pair $(\alpha, \beta)$ exists, then there also exists a witness pair $(h c^\omega, h' d^\omega)$ of lassos.
        Indeed, let $m$ be an index such that $\alpha_{\leq m} \neq \beta_{\leq m}$, and such that every vertex in $\alpha_{\geq m}$ is visited infinitely often by that path.
        Then, in particular, there is an index $m'$ such that $\alpha_{m'} = \alpha_m$.
        Let us then define $h = \alpha_{< m}$ and $c = \alpha_m \dots \alpha_{m'-1}$, and let us define similarly $h'$ and $d$.
        Since the projections of $hc^\omega$ and $h' d^\omega$ in $\Game_{\|v_0}$ give to all players the same payoffs as $\alpha$ and $\beta$, respectively, the pair $(hc^\omega, h'd^\omega)$ is also a witness pair.
        Let us therefore present an algorithm searching for such lassos.

        That algorithm can guess those lassos ``on the fly'', as follows:
        
        \begin{itemize}
            \item guess, vertex by vertex, the common prefix of $h$ and $h'$, with only non-deviating edges, and while doing so, run the automata $\auto_{i}$ (memorize nothing but the current vertex, and the current state in each $\auto_{i}$);

            \item guess the time of the first deviation: then, guess a non-deviating edge (to $u$) and a deviating edge (to $v$), and memorize $v$, and the current states of each automaton $\auto_{i}$ when reaching $v$;

            \item guess the end of $hc^\omega$: from $u$, guess vertex by vertex a path with no deviating edges, then guess which of those vertices will be repeated, memorize it, and guess the path from it to itself, still without deviating edge, and memorizing which states of each $\auto_{i}$ are visited while doing so;

            \item proceed similarly to guess the end of $h'd^\omega$ from $v$, using also deviating edges, and also memorizing which players deviate;

            \item then, use the data that has been memorized to compute $\mu_i(hc^\omega)$ and $\mu_i(h'd^\omega)$ for each $i$, and then check whether the pair $(hc^\omega, h'd^\omega)$ is a witness.
        \end{itemize}

        \item \emph{$\PSpace$-hardness.}

        We proceed by reduction from the following $\PSpace$-hard problem~\cite{Kozen77}: given $n$ deterministic finite automata $\auto_1, \dots, \auto_n$ on the alphabet $\{a, b\}$, is the intersection $\bigcap_i \Lang(\auto_i)$ empty?
        Given $n$ such finite automata, we construct a $\omega$-regular game $\Game_{\|v_0}$ and a Mealy machine $\Mach$ such that every strategy profile compatible with $\Mach$ is an SSE in $\Game_{\|v_0}$ if and only if $\bigcap_i \Lang(\auto_i) = \emptyset$.

        The game $\Game_{\|v_0}$ has $n+2$ players, namely players $1, \dots, n$, \emph{Solver}, written $\SS$, and \emph{Innocent}, written $\II$.
        For each $i \in \{1, \dots, n\}$, there is a vertex $i$ controlled by player $i$, with two outgoing edges: one to the sink vertex $\btd$, and one to the vertex $i+1$ (or $c$ if $i=n$).
        Solver controls two vertices, named $a$ and $b$, each of them having three outgoing edges: one to $a$, one to $b$, and one to a last sink vertex $\bt$.
        The vertices $\btd$ and $\bt$ have both only one outgoing edge, a self-loop.
        The initial vertex is $v_0 = 1$.
        That game is depicted by Figure~\ref{fig:autoInterEmptiness2omegarRsseCheck_arena}: each square vertex belongs to the number player indicated on the vertex, the round vertices belong to Solver, and the hexagon vertices to $\II$ (who actually has no choice at all).

        The parity automata that recognize the winning plays of each player are depicted by Figures~\ref{fig:autoInterEmptiness2omegarRsseCheck_solver}, \ref{fig:autoInterEmptiness2omegarRsseCheck_innocent} and~\ref{fig:autoInterEmptiness2omegarRsseCheck_i}.
        The number inside each state is its color.
        In Figure~\ref{fig:autoInterEmptiness2omegarRsseCheck_i}, each state of the automaton $\auto_i$ is given the color $1$, and an outgoing transition with label $\bt$ is added to each of them: to the sink state of color $0$ if it is an accepting state, to the sink state of color $1$ otherwise.
        Thus, Solver always wins; Innocent loses if and only if the vertex $c$ is reached; and each player $i$ wins if and only if either the vertex $\btd$ is reached, or the vertex $c$ is reached, Solver writes a word that is accepted by $\auto_i$, and then goes to $\bt$.

        Finally, the Mealy machine $\Mach$ is the deterministic one-state machine that makes all the number players go to the vertex $\btd$, and that makes Solver do arbitrary actions (Innocent can only choose self loops on sink vertices).

\begin{figure*}
    \centering

\begin{subcaptionblock}{0.9\textwidth}
    \centering
    \begin{tikzpicture}[node distance=2cm]
        \node[literalState=$1$,initial] (1) at (0,0) {$1$};
        \node[literalState=$2$,right of=1] (2) {$2$};
        \node[right of=2] (dots) {$\cdots$};
        \node[literalState=$n$,right of=dots] (n) {$n$};
        \node[state, below right of=n] (a) {$a$};
        \node[innocentState, below right of=a] (bt) {$\bt$};
        \node[state, above right of=bt] (b) {$b$};
        \node[innocentState,below of=dots] (btd) {$\btd$};
        \node[state, above right of=a] (c) {$c$};
\begin{pgfonlayer}{arrows}
        \path[->] (1) edge (2);
        \path[->] (2) edge (dots);
        \path[->] (dots) edge (q);
        \path[->] (n) edge (c);
        \path[->] (c) edge[bend right] (a);
        \path[->] (c) edge[bend left] (b);
        \path[->] (a) edge[bend right] (b);
        \path[->] (b) edge[bend right] (a);
        \path[->] (a) edge[loop left] (a);
        \path[->] (b) edge[loop right] (b);
        \path[->] (1) edge[bend right] (btd);
        \path[->] (2) edge[bend right, in=150] (btd);
        \path[->] (n) edge[bend left] (btd);
        \path[->] (btd) edge[loop below] (btd);
        \path[->] (a) edge[bend right] (bt);
        \path[->] (b) edge[bend left] (bt);
        \path[->] (bt) edge[loop below] (bt);
        \end{pgfonlayer}
    \end{tikzpicture}
    \caption{The arena.}
    \label{fig:autoInterEmptiness2omegarRsseCheck_arena}
    \end{subcaptionblock}

    \begin{subcaptionblock}{0.35\textwidth}
    \centering
    \begin{tikzpicture}[auto,node distance=3cm]
        \node[automatonState,initial] (init) at (0,0) {$0$};
        \path[->] (init) edge[loop above] node{$*$} (init);
    \end{tikzpicture}
    \caption{The parity automaton $\auto_{\SS}$.}
    \label{fig:autoInterEmptiness2omegarRsseCheck_solver}
    \end{subcaptionblock}
\hspace{2cm}
    \begin{subcaptionblock}{0.35\textwidth}
    \centering
    \begin{tikzpicture}[auto,node distance=3cm]
        \node[automatonState,initial] (init) at (0,0) {$0$};
        \node[automatonState,right of=init] (win) {$1$};
        \path[->] (init) edge[loop above] node{$1,\dots,n,\btd$} (init);
        \path[->] (init) edge node{$c$} (win);
        \path[->] (win) edge[loop above] node{$a,b,\bt$} (win);
    \end{tikzpicture}
    \caption{The parity automaton $\auto_{\II}$.}
    \label{fig:autoInterEmptiness2omegarRsseCheck_innocent}
    \end{subcaptionblock}
 \bigskip
 
    \begin{subcaptionblock}{0.7\textwidth}
    \centering
    \begin{tikzpicture}[auto,node distance=3cm]
        \node[automatonState,initial] (init) at (0,0) {$1$};
        \node[automatonState,below right of=init] (win) {$0$};
        \node[automatonState,rectangle,rounded corners=7pt,dotted,minimum size=3cm,above right of=win] (autom) {$\auto_{i}$};
        \node[automatonState,double] (qacc) at (3.5, -1) {1};
        \node[automatonState] (someq) at (5, -1) {1};
        \node[automatonState,below right of=autom] (lose) {$1$};

        \path[->] (init) edge[loop above] node{$1,\dots,n$} (init);
        \path[->] (init) edge node{$c$} (autom);
        \path[->] (init) edge node{$\btd$} (win);
        \path[->] (qacc) edge node{$\bt$} (win);
        \path[->] (someq) edge node{$\bt$} (lose);
        \path[->] (win) edge[loop below] node{$*$} (win);
        \path[->] (lose) edge[loop below] node{$*$} (lose);
    \end{tikzpicture}
    \caption{The parity automaton $\auto_{i}$ for player $i \in \{1, \dots, n\}$}
    \label{fig:autoInterEmptiness2omegarRsseCheck_i}
    \end{subcaptionblock}

    \caption{Encoding of the intersection emptiness problem into an $\omega$-regular SSE-checking problem.}
    \label{fig:autoInterEmptiness2omegarRsseCheck}
\end{figure*}

Thus, let us assume that there is a strategy profile $\bsigma$ that is compatible with $\Mach$, and that is not an SSE.
The play $\< \bsigma \>$ reaches the vertex $\btd$, and is therefore won by all the players.
Any harmful deviation must therefore be losing for some player who is not deviating.
A number player who is not deviating makes the play reach vertex $\btd$ and therefore still wins and is not harmed by the deviation.
Since Solver cannot lose, the only potentially harmed player is $\II$; that means the deviating play reach the vertex $c$.
As noted above, deviating coalition must necessarily include all the number players to avoid vertex $\btd$.
Consequently, a harmful deviation $\bsigma'$ of $\bsigma$ is necessarily a deviation of all the players $\{1,\cdots,n,\SS\}$, and must generate a play that is won by all of them.
Such a play must have the form $\< \bsigma' \> = 1 2 \dots n w \bt^\omega$, where $w$ is a word accepted by every automaton $\auto_i$: the intersection $\bigcap_i \Lang (\auto_i)$ is then nonempty.

Conversely, if that intersection is nonempty, then a harmful deviation for the coalition $\{1,\cdots,n,\SS\}$ consists in going to the vertex $c$, writing a word $w \in \bigcap_i \Lang (\auto_i)$, and then going to the sink vertex $\bt$.

Thus, the emptiness problem of the intersection of finite automata reduces to the SSE-checking problem in $\omega$-regular games, with a deterministic Mealy machine: that problem is therefore $\PSpace$-hard.

    \item\emph{Easiness (fixed number of players).}

    Let $\Game^\star_{\|v_0}$ be the game graph $\Game_{\|v_0}$ in synchronized product with the parity automata $\auto_1, \dots,\auto_n$.
    Let $\kappa_i'$ be the coloring function for player $i$ defined using the corresponding color in the original automaton $\auto_i$: $\kappa_i'(v,s_1,\dots,s_n)=\kappa_i(s_i)$.
    This game has a size $O(m^n)$ with \[m=\max\left\{2,|V|,\max_{i\in\Pi}\left|Q_i\right|\right\}\] the maximal size of the graph/automata provided as input (assumed to be at least 2 for complexity purposes).
    Using the procedure of Theorem~\ref{thm:SSEcheckingParity} on $\Game^\star_{\|v_0}$ yields an algorithm with complexity $O(m^{5n} \cdot |\Mach|^2 \cdot n \cdot k)$ with $k$ the maximal color in all parity conditions.
    This is therefore polynomial in the size of the input when $n$ is fixed (it falls into the $\XP$ complexity class with $n$ as a parameter).
    
    \item\emph{Hardness (fixed number of players).}
    The complexity of the above algorithm is not Fixed-Parameter-Tractable, as the base $m$ is not fixed.
    In fact, this problem is $\W[1]$-hard:  as shown above for the general case, we can reduce the emptiness problem for $n$ deterministic finite automata to SSE checking.
    This problem was shown to be $\W[1]$-hard~\cite[Corollary 5.9]{FernauHoffmannWehar2021}, which concludes the proof.
    \qedhere
    \end{itemize}
\end{proof}

\subsection{Compositional SSE-checking problems}

We now study the same problem, but where the strategy profile of the players is described \emph{compositionally}, that is, with one machine per player.
In many cases, that description is more succinct: it is therefore not surprising that the complexity is higher in the case of parity games.
Interestingly, it is not the case in $\omega$-regular games, where we can use the same algorithm as before.

    \subsubsection{In \texorpdfstring{$\omega$}{ω}-regular games}

\begin{thm} \label{thm:compCheckingOmegaReg}
    In $\omega$-regular games, the compositional SSE-checking problem is $\PSpace$-complete, even when the Mealy machines are deterministic.
    With a fixed number of players, the problem can be solved in polynomial time but is $\W[1]$-hard.
\end{thm}

\begin{proof}~
    \begin{itemize}
        \item \emph{Easiness.}

        Given a machine profile $(\Mach_i)_{i \in \Pi}$, let us consider the product machine $\Mach$ whose state space is the product of the $\Mach_i$'s state spaces, and whose transitions are naturally defined in order to have $\Comp(\Mach) = \prod_{i \in \Pi} \Comp(\Mach_i)$.

        Then, deciding the compositional SSE-checking problem on the input $\Game_{\|v_0}, (\Mach_i)_i$ is equivalent to decide the SSE-checking problem on the input $\Game_{\|v_0}, \prod_i \Mach_i$.
        That can be done using the non-deterministic algorithm that was presented in the proof of Theorem~\ref{thm:SSEcheckingOmegaRegular}, which does not require the effective construction of $\Game_{\|v_0} \otimes \Mach$, nor the construction of $\Mach$ itself.

        For a fixed number of players, combining the individual machines into one yields a machine of size $\prod_i|\Mach_i|$.
        Applying the same algorithm as in  the proof of Theorem~\ref{thm:SSEcheckingOmegaRegular} also provides a polynomial time algorithm.
        Once again, the complexity is exponential in $n$ and the base depends on the size of the inputs (game, Mealy machines and parity automata, in this case), so that is an $\XP$ but not an $\FPT$ algorithm.

        \item \emph{Hardness.}

        The reduction from deterministic finite automata that was given in the proof of Theorem~\ref{thm:SSEcheckingOmegaRegular} is still functional, provided that the Mealy machine $\Mach$ is replaced by a family of (deterministic) machines $\Mach_1, \dots, \Mach_n$ (which all go to the vertex $\btd$), and arbitrary one-state deterministic machines $\Mach_\SS$ and $\Mach_\II$.
        Therefore the general case is also $\PSpace$-hard and $\W[1]$-hard for a fixed number of players.
\qedhere
    \end{itemize}
\end{proof}

    \subsubsection{In parity games}

In parity games, the situation is quite similar: the complexity that was brought by the multiplicity of automata in $\omega$-regular games is now imposed by the Mealy machines.

\begin{thm}
    In parity games, the compositional SSE-checking problem is $\PSpace$-complete, even when the Mealy machines are deterministic.
    When the number of players is fixed, the problem becomes $\Poly$-easy and $\W[1]$-hard.
\end{thm}

\begin{proof}~
    \begin{itemize}
        \item \emph{Easiness.}

        Every parity game is an $\omega$-regular game (with one-state parity automata), hence the $\PSpace$-easiness and $\Poly$-easiness for a fixed number of players are an immediate consequence of Theorem~\ref{thm:compCheckingOmegaReg}.

        \item \emph{Hardness.}

        We proceed, once again, by reduction from the emptiness problem of the intersection of $n$ deterministic finite automata.
        This shows both the $\PSpace$-hardness and the $\W[1]$-hardness when the number of players is fixed.
        Let $\auto_1, \dots, \auto_n$ be $n$ such automata on the alphabet $\{a, b\}$.

        Let $\Game_{\|a}$ be the game depicted by Figure~\ref{fig:autoInterEmptiness2CompParitySseCheck_arena}.
        The players are Solver, controlling the round vertices, and $1, \dots, n$, controlling the square vertices (the player controlling each square vertex is indicated in its bottom right corner).
        The color mappings $\kappa_{1}, \dots, \kappa_{n}, \kappa_{\SS}$ are defined so that Solver always wins, and each number player wins if and only if the vertex $\btd$ is not reached.
        Namely $\kappa_{\SS}(v)=0$ for every vertex $v$, and for $i\in \{1,\dots,n\}$, $\kappa_{i}(\btd)=1$ and $\kappa_{i}(v)=0$ for any other vertex $v$.        

        \begin{figure*}
    \centering

\begin{subcaptionblock}{0.9\textwidth}
    \centering
    \begin{tikzpicture}[node distance=2cm]
        \node[state, initial, initial above] (a) at (0,0) {$a$};
        \node[state, right of=a] (b) {$b$};
        \node[literalState=$1$,below of=b] (1) {$1$};
        \node[literalState=$2$,right of=1] (2) {$2$};
        \node[right of=2] (dots) {$\cdots$};
        \node[literalState=$n$,right of=dots] (n) {$n$};
        \node[nobodyState, right of=n] (bt) {$\btd$}; %% Names don't match label because it was taken from an earlier proof for MVP games and SSE checking
        \node[nobodyState,below of=dots] (btd) {$\bt$};
\begin{pgfonlayer}{arrows}
        \path[->] (1) edge (2);
        \path[->] (2) edge (dots);
        \path[->] (dots) edge (n);
        \path[->] (n) edge (bt);
        \path[->] (a) edge[bend right] (b);
        \path[->] (b) edge[bend right] (a);
        \path[->] (a) edge[bend right] (1);
        \path[->] (b) edge (1);
        \path[->] (a) edge[loop left] (a);
        \path[->] (b) edge[loop right] (b);
        \path[->] (1) edge[bend right] (btd);
        \path[->] (2) edge[bend right, in=150] (btd);
        \path[->] (n) edge[bend left] (btd);
        \path[->] (btd) edge[loop below] (btd);
        \path[->] (n) edge (bt);
        \path[->] (bt) edge[loop right] (bt);
        \end{pgfonlayer}
    \end{tikzpicture}
    \caption{The arena.}
    \label{fig:autoInterEmptiness2CompParitySseCheck_arena}
    \end{subcaptionblock}

    \begin{subcaptionblock}{0.25\textwidth}
    \centering
    \begin{tikzpicture}[auto,node distance=3cm]
        \node[automatonState,initial] (init) at (0,0) {$1$};
        \path[->] (init) edge[loop above] node{$a|a, b|b, *$} (init);
    \end{tikzpicture}
    \caption{The machine $\Mach_{\SS}$.}
    \label{fig:autoInterEmptiness2CompParitySseCheck_solver}
    \end{subcaptionblock}
    \begin{subcaptionblock}{0.74\textwidth}
    \centering
    \begin{tikzpicture}[auto,node distance=2.5cm]
        \node[automatonState,rectangle,rounded corners=7pt,dotted,minimum size=3cm,initial] (autom) {$\auto_{i}$};
        \node[automatonState,double] (qacc) at (1, 1) {\phantom{b}};
        \node[automatonState] (someq) at (1, -1) {\phantom{b}};
        \node[automatonState,right of=qacc] (1) {\phantom{b}};
        \node[right of=1] (dots) {$\dots$};
        \node[automatonState,right of=dots] (i) {\phantom{b}};
        \node[automatonState,right of=someq] (1b) {\phantom{b}};
        \node[right of=1b] (dotsb) {$\dots$};
        \node[automatonState,right of=dotsb] (ib) {\phantom{b}};

        \path[->] (qacc) edge node{$1$} (1);
        \path[->] (1) edge node{$2$} (dots);
        \path[->] (dots) edge node{$i|i+1$} (i);
        \path[->] (i) edge[loop right] node{$*$} (i);
        \path[->] (someq) edge node{$1$} (1b);
        \path[->] (1b) edge node{$2$} (dotsb);
        \path[->] (dotsb) edge node{$i|\bt$} (ib);
        \path[->] (ib) edge[loop right] node{$*$} (ib);
    \end{tikzpicture}
    \caption{The machine $\Mach_i$ for player $i \in \{1, \dots, n\}$}
    \label{fig:autoInterEmptiness2CompParitySseCheck_i}
    \end{subcaptionblock}

    \caption{Encoding of the intersection emptiness problem into a compositional Parity SSE-checking problem.}
    \label{fig:autoInterEmptiness2CompParitySseCheck}
\end{figure*}

    Now, the deterministic Mealy machine $\Mach_\SS$ is depicted by Figure~\ref{fig:autoInterEmptiness2CompParitySseCheck_solver}: it captures a strategy that consists in looping on the vertex $a$ forever.
    Similarly, for each player $i \in \{1, \dots, n\}$, the Mealy machine $\Mach_i$ is depicted by Figure~\ref{fig:autoInterEmptiness2CompParitySseCheck_i} (when $i=n$, the reader should identify $i+1$ with $\btd$).
    It captures a strategy that consists in always going to the vertex $\bt$, except if the word constructed by the actions of Solver was a word recognized by the automaton $\auto_i$.

    Those Mealy machines define a strategy profile $\bsigma$.
    Let us prove that $\bsigma$ is an SSE if and only if the intersection $\bigcap_i \Lang(\auto_i)$ is empty.

    \begin{itemize}
        \item \emph{If $\bsigma$ is not an SSE.}

        Then, some coalition can induce a harmful deviation $\bsigma'$: since Solver cannot lose, he cannot be the player harmed by that deviation and that must be one number players $i$.
        That means $\<\bsigma'\>$ reaches vertex $\btd$, so all number players are actually harmed.
        Since having no deviation would not harm anyone, the coalition must be the singleton $\{\SS\}$.

        There is, then, a play $\pi$ that is compatible with the strategy profile $\bsigma_{-\SS}$, and that is lost by number players: that play has necessarily the form $w12\dots n \btd^\omega$, where $w \in \{a, b\}^*$.
        Moreover, for that play to be compatible with each strategy $\sigma_i$, the word $w$ must be accepted by the automaton $\auto_i$.
        Then, the intersection $\bigcap_i \Lang(\auto_i)$ is nonempty.

        \item \emph{If the intersection $\bigcap_i \Lang(\auto_i)$ is nonempty.}

        Then, let $w \in \bigcap_i \Lang(\auto_i)$: the play $\pi = w 12 \dots n \btd^\omega$ is compatible with the strategy profile $\bsigma_{-\SS}$.
        It is losing for all number players (and winning for Solver), hence the coalition $\{\SS\}$ has a profitable deviation. \qedhere
    \end{itemize}
    \end{itemize}
\end{proof}

\subsection{Fixed-payoff SSE (constrained) existence problem}\label{sec:existenceSSE}

We now switch to the fixed-payoff SSE (constrained) existence problem.
Unsurprisingly, that problem is often harder to solve: while the SSE-checking problems amounted to search for a harmful deviation, that problem consists in searching for a strategy profile that admits \emph{no} such deviation.

        \subsubsection{A tool: the deviator game}

Instead of the product graph, we need a new tool to solve that problem: it will take the form of a new game structure, in which one player, \emph{Prover}, tries to prove that an SSE generating the desired payoff exists, while another one, \emph{Challenger}, tries to prove that the strategy profile she constructs is actually not an SSE.
That game is the \emph{deviator game}, very similar to a construction with the same name proposed in~\cite{DBLP:conf/fossacs/Brenguier16}, in a slightly different context.

\begin{defi}[Deviator game]
    Let $\Game_{\|v_0}$ be a graph game, let $\bx$ be a payoff vector, and $\Corr$ a correctness condition.
    The \emph{deviator game} is the initialized graph game:
    \[\Dev_{\bx,\Corr} \Game_{\|v_0} = \left( \{\PP, \CC\}, V^\cd, E^\cd, \left(V^\cd_\PP, V^\cd_\CC\right), \mu^\cd\right)_{\|v_0^\cd},\]
    where:

     \begin{itemize}
         \item the player $\PP$ is called \emph{Prover}, and the player $\CC$ \emph{Challenger}.

         \item Prover controls the set $V^\cd_\PP = V \times 2^\Pi$, and Challenger the set $V^\cd_\CC = E \times 2^\Pi$.

         \item The initial vertex is $v_0^\cd = (v_0, \emptyset)$.

         \item The edge set $E^\cd$ is defined as follows: from the vertex $(u, D)$, Prover can go to every vertex of the form $(uv, D) \in E$ (she \emph{proposes} the edge $uv$).
         From the vertex $(uv, D)$, Challenger can go to the vertex $(v, D)$ (he \emph{accepts} the edge $uv$), or to every vertex $(w, D \cup \{i\})$, with $w \neq v$ and $uw \in E$, and where $i$ is the player controlling $u$ (he \emph{deviates} from Prover's proposal and $i$ is added to the set of deviators).

         \item Given a play $\chi$ is this game, we write $\dchi$ the play in $\Game$ constructed by the actions of Prover and Challenger: if $\chi = (u_0, D_0)(u_0v_0, D_0)(u_1, D_1) (u_1v_1, D_1) \dots$, then we define $\dchi = u_0 u_1 \dots$.
         We also define $\D(\chi) = \bigcup_k D_k$, the set of players who deviated along the play $\chi$.

         \item Then, $\mu^\cd$ is the Boolean payoff function defined as follows. A play $\chi$ is won by Challenger if and only if either:
         \begin{itemize}
             \item we have $\D(\chi) = \emptyset$ and $\mu(\dchi) \neq \bx$;
             \item we have $\D(\chi) = \emptyset$ and $\dchi \not\in \Corr$;
             \item or for every player $i \in \D(\chi)$, we have $\mu_i(\dchi) \geq x_i$, and there exists a player $j \notin \Pi$ such that $\mu_j(\dchi) < x_j$.
         \end{itemize}
     \end{itemize}
\end{defi}

\begin{thm} \label{thm_condev}
    Prover has a winning strategy in the game $\Dev_{\bx,\Corr} \Game_{\|v_0}$ if and only if there exists an SSE $\bsigma$ in $\Game_{\|v_0}$ with $\mu\< \bsigma \> = \bx$ and $\< \bsigma \>\vDash \Corr$.
\end{thm}

        \subsubsection{In parity games}

Let us first consider the case of parity games: then, similarly to what was done in~\cite{DBLP:conf/fossacs/Brenguier16}, a cautious way to solve the deviator game leads to an algorithm that uses only polynomial space.

\begin{thm}\label{thm:sseExistencePspaceEasy}
    In parity games, the fixed-payoff SSE (constrained) existence problem is $\PSpace$-complete.
    It is also $\FPT$ if both the number of colors and the number of players are fixed, but not if only one of them is ---~unless solving two-player zero-sum parity games is $\Poly$-easy.
\end{thm}

\begin{proof}~
    \begin{itemize}
        \item \emph{$\PSpace$-easiness.}

    By Theorem~\ref{thm_condev}, deciding the fixed-payoff SSE existence problem amounts to deciding which player, among Prover and Challenger, has a winning strategy in the corresponding deviator game.
    That game has an exponential size, but has a specific shape that makes it possible to solve it region by region, without using exponential space.

    Indeed, let us define recursively an algorithm that decides, given $\Game_{\|v_0}$, $\bx$ and a vertex $(u_0, D)$ of the game $\Dev_{\bx,\Corr} \Game$, whether Prover has a winning strategy from that vertex.
    Let $W = \{i\in\Pi ~|~ x_i=1\}$ be the set of players winning according to payoff vector $\bx$.
    First, construct the game $\Gameb^D_{\|(u_0, D)}$ as follows: construct the region of $\Dev_{\bx,\Corr} \Game$ that is made of vertices of the form $(u, D)$ or $(uv, D)$.
    Define on that region the color mappings so that player $i$'s payoff, in that region, is always equal to their payoff in the corresponding play in $\Game$: i.e., for all $i \in \Pi$, define $\kappa'_{i}(uv, D) = \kappa'_{i}(u, D) = \kappa_{i}(u)$ (and similarly for the correctness condition: $\kappa'_{\Corr}(uv, D) = \kappa'_{\Corr}(u, D) = \kappa_{i}(u)$).
    
    For each edge $(uv, D)(w, D \cup \{i\})$ that leaves that region, add to the constructed game the edge $(uv, D)(w, D \cup \{i\})$ and the loop $(w, D \cup \{i\})(w, D \cup \{i\})$, and decide recursively whether Prover has a winning strategy from the vertex $(w, D \cup \{i\})$.
    If she does, call the vertex $(w, D \cup \{i\})$ a \emph{Prover's leaf}, and define the color mappings so that only players in $W$ would win on that loop: for each $i' \in W$, define $\kappa'_{i'}(w, D \cup \{i\}) = 0$; for $i'\notin W$, $\kappa'_{i'}(w, D \cup \{i\}) = 1$; and $\kappa'_{\Corr}(w, D \cup \{i\}) = 0$.
    If she does not, call the vertex $(w, D \cup \{i\})$ a \emph{Challenger's leaf}, and define the color mappings so that only players in $D \cup \{i\}$ would win on that loop: for each $i' \in D \cup \{i\}$, define $\kappa'_{i'}(w, D \cup \{i\}) = 0$; for $i \notin D \cup \{i\}$, $\kappa'_{i'}(w, D \cup \{i\}) = 1$; and $\kappa'_{\Corr}(w, D \cup \{i\}) = 1$.
    
    Now, define the payoff functions as follows.
    If $D \neq \emptyset$, then a play $\pi$ is won by Challenger if and only if:
    \[\pi \in \left( \bigcup_{i \in W} \overline{\Parity(\kappa'_{i})} \right) \cap \left( \bigcap_{i \in D\cap W} \Parity(\kappa'_{i}) \right)\]
    (i.e., if $\pi$ gives to at least one player $i$ a payoff worse than $x_i$, and to every player who deviated at least the same payoff as in $\bx$).
    If $D = \emptyset$, then a play $\pi$ is won by Challenger if and only if:
    \[\pi \in \bigcup_{i \in W} \overline{\Parity(\kappa'_{i})} \cup \bigcup_{i \notin W} \Parity(\kappa'_{i})\cup \overline{\Parity(\kappa_\Corr')}\]
    (i.e., if the payoff vector generated by Prover, without any deviation, is not equal to $\bx$ or the correctness constraint is not met).
    Note that in both cases, Challenger wins if he reaches a Challenger's leaf, and loses if he reaches a Prover's leaf.
    
    Prover has a winning strategy in the game $\Gameb^D_{\|(u_0, D)}$ if and only if she has a winning strategy from the vertex $(u_0, D)$ in the game $\Dev_{\bx,\Corr} \Game$.
    Indeed, if Prover has a winning strategy in $\Gameb^D_{\|(u_0, D)}$, then she can follow that strategy in $\Dev_{\bx,\Corr} \Game$.
    Then, she either stays in the region where the deviating coalition is $D$, or she reaches a vertex of the form $(w, D \cup \{i\})$ from which she has a winning strategy, that she can then follow.
    Conversely, if she has a winning strategy from $(u_0, D)$ in $\Dev_{\bx,\Corr} \Game$, she can follow it in $\Gameb^D_{\|(u_0, D)}$: she will then either stay in the region where the deviating coalition is $D$, or reach a Prover's leaf (since the strategy she is following afterwards in $\Dev_{\bx,\Corr} \Game$ is still winning), which makes her win.

    This game can be seen as an Emerson-Lei game~\cite{EmersonLei87}: Prover's and Challenger's winning conditions are Boolean combinations of Büchi conditions.
    Indeed, if $\B(v^\cd)$ denotes, for each vertex $v^\cd$, the set of plays that visit infinitely often the vertex $v^\cd$, then we have the equality:
    \ifArxiv
    \[\Parity(\kappa'_{i}) = \bigcup_{2m < k} \left(\bigcup_{\kappa'_{i}(v^\cd) = 2m} \B(v^\cd) \cap \neg \bigcup_{\l < 2m} \bigcup_{\kappa'_{i}(w^\cd) = \l} \B(w^\cd)\right)\]
    \else
    \begin{mathpar}
        \Parity(\kappa'_{i}) = \bigcup_{2m < k} \Bigg(\bigcup_{\kappa'_{i}(v^\cd) = 2m} \B(v^\cd) \cap \neg \bigcup_{\l < 2m} \bigcup_{\kappa'_{i}(w^\cd) = \l} \B(w^\cd)\Bigg)
    \end{mathpar}
    \fi
    and $\overline{\Parity(\kappa'_{i})} = \neg \Parity(\kappa'_{i})$.
    It was shown in~\cite{HunterDawar05} that such games can be solved using a space polynomial in the size of the game and of the objectives.
    Here, both are themselves polynomial in the size of $\Game$.
    Thus, the last step of our recursive algorithm consists in solving $\Gameb^D_{\|(u_0, D)}$.

    Then, applying this recursive algorithm from the vertex $(v_0, \emptyset)$ solves our problem: there is an SSE $\bsigma$ in $\Game_{\|v_0}$ such that $\mu\< \bsigma \> = \bx$ if and only if Prover has a winning strategy from that vertex.
    Moreover, that algorithm uses polynomial space: each recursive call does, and the recursion stack has size at most $|\Pi|$, i.e. polynomial.
    Therefore, the fixed-payoff SSE existence problem in parity games is $\PSpace$-easy.

        \item \emph{Fixed-parameter tractability.}

        In~\cite{BruyereHautemRaskin18}, it has been shown that an Emerson-Lei game $\Game$ with state space $S$ and objective $\phi$ can be solved in time:
        \[O\left( 2^{2^{|\phi|}} |\phi| + \left(2^{|\phi| 2^{|\phi|}} |S|\right)^5\right).\]
        In our case, the size of the objective $\phi$ depends only on $k$ and $n$, while $|S| = 2^n |V|$, hence the desired result.
        
        \item \emph{$\PSpace$-hardness (with fixed number of colors).}

        We reduce the Quantified Satisfaction Problem (QSAT) to the (unconstrained) SSE fixed-payoff existence problem in co-Büchi games: that will then prove $\PSpace$-hardness for the SSE fixed-payoff existence problem, even when $m$ is fixed ---~hence it cannot be $\FPT$ with $m$ as parameter.
        
    Let \[\varphi=\exists x_1 \forall x_2  \cdots \exists x_{n-1} \forall x_n \bigwedge_{i=1}^p \bigvee_{j=1}^{m_i} \l_{i,j}\] be an instance of QSAT in conjunctive normal form where for all $i,j$, $\l_{i,j}$ is a literal of the form $x_k$ or $\neg x_k$ for some $k\in\{1,\dots,n\}$.
    We write $\psi(x_1,\dots,x_n)=\bigwedge_{i=1}^p \bigvee_{j=1}^{m_i} \l_{i,j}$ the non-quantified part of the formula.
    \medskip
    
    We construct a game with $2n+2$ players: one for each literal and two extra players \emph{Solver} (written $\solver$) and \emph{Opponent} (written $\opponent$) as depicted in \figurename~\ref{fig:qsat2SSEGame}.
    Round states belong to Solver, hexagonal states belong to Opponent, rectangular states belong to the literal indicated in the bottom right corner (if any: some states only have a single outgoing edge and therefore could belong to anyone without loss of generality).
    Red clouds indicate the list of players for which a given state is in the co-B\"uchi set: infinitely many visits to these states means losing (payoff $0$) while only finitely many visits means winning (payoff $1$).
    
    It is composed of three parts: in the first part, the \emph{setting} module, a valuation can be chosen by visiting $?x_i$ (setting $x_i=\top$) or $\neg x_i$ (setting $x_i=\bot$), the choice belonging to Solver if $x_i$ is existentially quantified and to Opponent if it is universally quantified.
    When the value of a variable is chosen, the corresponding literal player can choose to go to a sink state $\blacktriangledown$ (duplicated in the figure for clarity) or continue the game.
    
    If that does not happen then the game moves into the second part, the \emph{checking} module, where the formula is ``checked'': each disjunctive clause $C_i$ has a state owned by Solver who must choose a literal $\l_{i,j}$ from the clause.
    State $\l_{i,j}$ is in the co-B\"uchi set of the player $\neg \l_{i,j}$, and would lose if said state was visited infinitely often.
    When this ``checking'' is done, the last part, called \emph{punishing} module, visits the co-B\"uchi set of one literal per variable, the choice of which being left to Solver.

    The game continues back to the checking module.
    As the last state of these modules is in the co-B\"uchi set of Opponent, it will lose if the game enters these modules.
    Remark that there are no co-B\"uchi states for Solver who therefore always wins.
    \medskip

    We will show that $\varphi$ holds if and only if there can be no SSE with payoff vector $(1)_{i \in \Pi}$.
    The intuition is that a coalition made of Solver and all true literals can deviate to make Opponent and all false literals lose while not penalizing the players in the coalition.

    \begin{itemize}
        \item First, assume $\varphi$ does not hold.
    Let $\bsigma$ be the following profile of strategies: each literal goes to $\blacktriangledown$ whenever possible in the setting module, Solver plays any strategy, and Opponent plays the optimal strategy in response to Solver's choice of valuation so that $\psi(x_1,\dots,x_n)$ is false.
    The outcome of this profile is either path $?x_1\cdot x_1^s \cdot\blacktriangledown^\omega$ or $?x_1\cdot\neg x_1^s\cdot\blacktriangledown^\omega$, both with payoff vector $(1)_{i \in \Pi}$.
    
    Let us prove that it is an SSE.
    Let $C$ be a coalition of players and assume that there exists a harmful deviation for this coalition.
    As Solver always wins, she can be assumed to be in the coalition.
    And since no one loses in the setting module or $\btd$, any harmful deviation must reach reach the checking module.
    Because that means that Opponent would then lose, Opponent cannot be in the coalition.
    For each variable $x_i$, players $x_i$ and $\neg x_i$ cannot be both in the coalition as (at least) one of them loses through the infinite visits in the punishing module.
    To actually reach the checking module, all literal players whose state was visited in the setting module must have deviated from $\bsigma$, therefore are in the coalition.
    As a result the coalition is made of Solver and exactly one literal per variable, defining a valuation.
    Because Opponent played optimally, that valuation does not satisfy $\psi(x_1,\dots,x_n)$, so there is a clause $C_i$ where all literals $\l_{i,1}, \dots, \l_{i,m_i}$ are false, meaning that the player corresponding to their negation are all in the coalition.
    The choice by Solver of any of these literals therefore enforces a visit to the co-B\"uchi set of (at least) one player in the coalition.
    As this choice is infinitely repeated, that makes this player lose and therefore it should not be part of the coalition, which is a contradiction.
    So $\bsigma$ is an SSE.
    
    \item Now assume that $\varphi$ holds.
    Let $\bsigma$ be a strategy profile with payoff $(1,\dots,1,1,1)$.
    As noted above, that payoff requires the outcome of $\bsigma$ to end in the $\blacktriangledown$ state.
    Consider the following strategy for Solver: in the setting module, choose the literal that ensures satisfaction of $\psi(x_1,\dots,x_n)$, based on the already known choices of Opponent.
    This is possible since $\varphi$ holds.
    Let $C$ be the coalition made of Solver and all literal players corresponding to states $x_i^s$ or $\neg x_i^s$ visited in the setting module.
    These players deviate by reaching state $?x_{i+1}$ (or $C_1$ if $i=n$) instead of going to $\blacktriangledown$.
    Since the checking module is reached, Opponent (at least) will lose and be harmed by the deviation.
    The valuation thus built ensures that $\psi(x_1,\dots,x_n)$ is satisfied, so in the checking module, for every clause there is a true literal.
    The strategy of Solver consists in consistently choosing the state for these literals.
    Note that that means visiting a co-B\"uchi set for a player that is not in the coalition.
    In the punishing module, Solver visits states corresponding to players in the coalition (i.e. the exact same one that were visited in the setting module: if $x_i^s$ was visited, consistently choose $x_i^p$, if $\neg x_i^s$ was visited, consistently choose $\neg x_i^p$.
    This ensures that the co-B\"uchi set of no player in the coalition $C$ is ever visited, hence all players in this coalition still win, therefore the deviation is harmful and $\bsigma$ is not an SSE.
    \end{itemize}

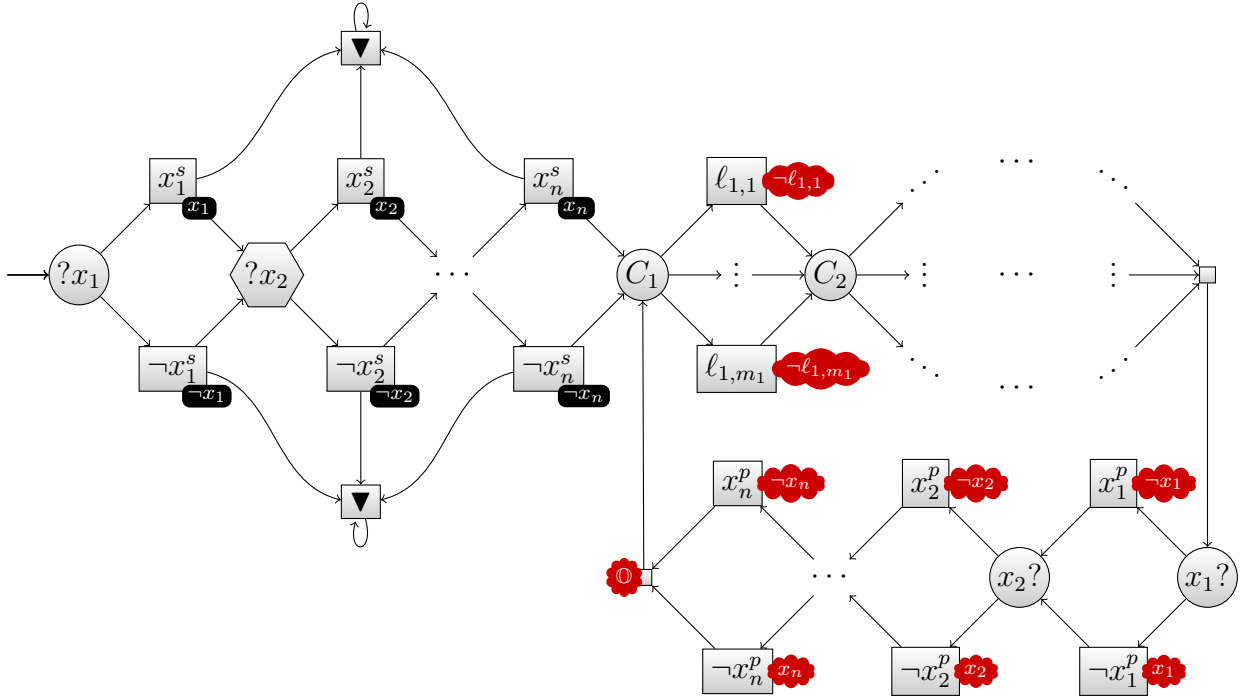
\begin{figure*}
    \centering
    \begin{tikzpicture}[node distance=50pt]
        % Setting the valuation
        \node[solverState,initial] at (0,0) (q1) {$?x_1$};
        \node[literalState=$x_1$,above right of=q1] (p1) {$x_1^s$};
        \node[literalState=$\neg x_1$,below right of=q1] (n1) {$\neg x_1^s$};
        \node[opponentState,below right of=p1] (q2) {$?x_2$};
        \node[literalState=$x_2$,above right of=q2] (p2) {$x_2^s$};
        \node[literalState=$\neg x_2$,below right of=q2] (n2) {$\neg x_2^s$};
        \node[below right of=p2] (q3) {$\cdots$};
        \node[literalState=$x_n$,above right of=q3] (pn) {$x_n^s$};
        \node[literalState=$\neg x_n$,below right of=q3] (nn) {$\neg x_n^s$};

        % Checking all the clauses
        \node[solverState,below right of=pn] (c1) {$C_1$};
        \node[nobodyState,above right of=c1] (l11) {$\l_{1,1}$}; \cobuchiLabel[inner sep=0pt,cloud puff arc=95]{l11}{$\neg \l_{1,1}$}
        \node[nobodyState,below right of=c1] (l1m1) {$\l_{1,m_1}$};  \cobuchiLabel[inner sep=0pt,cloud puff arc=95]{l1m1}{$\neg \l_{1,m_1}$}
        \node (c1dots) at (barycentric cs:l11=1,l1m1=1) {$\myvdots$};
        \node[solverState,below right of=l11] (c2) {$C_2$};
        \node[above right of=c2] (l21) {$\myiddots$};
        \node[below right of=c2] (l2m2) {$\myddots$};
        \node (c2dots) at (barycentric cs:l21=1,l2m2=1) {$\myvdots$};
        \node[right of=l21,node distance=2.5cm] (lp1) {$\myddots$};
        \node[right of=l2m2,node distance=2.5cm] (lpmp) {$\myiddots$};
        \node (cpdots) at (barycentric cs:lp1=1,lpmp=1) {$\myvdots$};
        \node[yshift=7pt] (ldots1) at (barycentric cs:l21=1,lp1=1) {$\cdots$};
        \node[yshift=-7pt] (ldotsm) at (barycentric cs:l2m2=1,lpmp=1) {$\cdots$};
        \node (ldots) at (barycentric cs:ldots1=1,ldotsm=1) {$\cdots$};
        \node[nobodyState,below right of=lp1] (ph) {};

        % Punition of literals
        \node[solverState,below of=ph,node distance=4cm] (v1) {$x_1?$};
        \node[nobodyState,above left of=v1] (vp1) {$x_1^p$}; \cobuchiLabel{vp1}{$\neg x_1$}
        \node[nobodyState,below left of=v1] (vn1) {$\neg x_1^p$}; \cobuchiLabel{vn1}{$x_1$}
        \node[solverState,below left of=vp1] (v2) {$x_2?$};
        \node[nobodyState,above left of=v2] (vp2) {$x_2^p$}; \cobuchiLabel{vp2}{$\neg x_2$}
        \node[nobodyState,below left of=v2] (vn2) {$\neg x_2^p$}; \cobuchiLabel{vn2}{$x_2$}
        \node[below left of=vp2] (vn) {$\cdots$};
        \node[nobodyState,above left of=vn] (vpn) {$x_n^p$};  \cobuchiLabel{vpn}{$\neg x_n$}
        \node[nobodyState,below left of=vn] (vnn) {$\neg x_n^p$}; \cobuchiLabel{vnn}{$x_n$}
        \node[nobodyState,below left of=vpn] (ph2) {}; \cobuchiLabelLeft{ph2}{$\opponent$}

        % Sink state
        \node[nobodyState,above of=p2] (sinkTop) {$\blacktriangledown$}; %\cobuchiLabelTop{sinkTop}{$\solver$}
        \node[nobodyState,below of=n2] (sinkBot) {$\blacktriangledown$}; %\cobuchiLabelBot{sinkBot}{$\solver$}

        \begin{pgfonlayer}{arrows}
            \path[->] (q1) edge (p1) edge (n1);
            \path[->] (p1) edge (q2) edge[bend left,out=-30] (sinkTop);
            \path[->] (n1) edge (q2) edge[bend right,out=30] (sinkBot);
            \path[->] (q2) edge (p2) edge (n2);
            \path[->] (p2) edge (q3) edge (sinkTop);
            \path[->] (n2) edge (q3) edge (sinkBot);
            \path[->] (q3) edge (pn) edge (nn);
            \path[->] (pn) edge (c1) edge[bend right,out=30] (sinkTop);
            \path[->] (nn) edge (c1) edge[bend left,out=-30] (sinkBot);
    
            \path[->] (c1) edge (l11) edge (c1dots) edge (l1m1);
            \path[<-] (c2) edge (l11) edge (c1dots) edge (l1m1);
            \path[->] (c2) edge (l21) edge (c2dots) edge (l2m2);
            \path[<-] (ph) edge (lp1) edge (cpdots) edge (lpmp);

            \path[->] (ph) edge (v1);

            \path[->] (v1) edge (vp1) edge (vn1);
            \path[<-] (v2) edge (vp1) edge (vn1);
            \path[->] (v2) edge (vp2) edge (vn2);
            \path[<-] (vn) edge (vp2) edge (vn2);
            \path[->] (vn) edge (vpn) edge (vnn);
            \path[<-] (ph2) edge (vpn) edge (vnn);

            \path[->] (ph2) edge (c1);
    
            \path[->] (sinkTop) edge[loop above] (sinkTop);
            \path[->] (sinkBot) edge[loop below] (sinkBot);
            \end{pgfonlayer}
    \end{tikzpicture}
    \caption{Game to encode QSAT  instance $\exists x_1 \forall x_2  \cdots \exists x_{n-1} \forall x_n \bigwedge_{i=1}^p \bigvee_{j=1}^{m_i} \l_{i,j}$ into the existence of an SSE with payoff $(1,\dots,1,1,1)$. Round states belong to Solver, hexagonal states belong to Opponent, rectangular states belong to the literal indicated in the bottom right corner (if any). The red cloud lists the set of players for which this state is in the co-B\"uchi condition.}
    \label{fig:qsat2SSEGame}
\end{figure*}

    \item \emph{$\PSpace$-hardness (with fixed number of players).}

    We show here that when $n$ is fixed to $2$, the SSE fixed-payoff existence problem is at least as hard as solving a two-player zero-sum parity game.
    The latter problem is known to belong to $\NP \cap \coNP$; it is not known whether it can be solved in deterministic polynomial time.

    Let $\Game$ be a parity game $\Game_{\|v_0}$ with two players, \emph{Eve} (written $\eve$) and \emph{Adam} (written $\adam$), and such that $\kappa_\adam = \kappa_\eve + 1$ (i.e., Adam's winning condition is the complement of Eve's one).

    Let us now consider the game $\Game'_{\|v_0}$, with the same arena, and with color mappings $\kappa'_\eve = \kappa_\eve$ and $\kappa_\adam$ constantly equal to $2$.

    Let us assume Eve has a winning strategy $\sigma_\eve$ in $\Game_{\|v_0}$: then, for any strategy $\sigma_\adam$, the strategy profile $\bsigma$ is, in $\Game'_{\|v_0}$, an SSE where both Eve and Adam get the payoff $1$.
    Indeed, Adam gets the payoff $1$ as he always does, Eve gets the payoff $1$ because $\sigma_\eve$ wins against $\sigma_\adam$, and that strategy profile is an SSE because no player can deviate to harm the other one: Adam wins every play, and Eve is playing a winning strategy.

    Conversely, if $\bsigma$ is an SSE in $\Game'_{\|v_0}$ in which both Eve and Adam get the payoff $1$, then consider some alternative strategy $\sigma'_\adam$ for Adam: like every play, the play $\< \sigma_\eve, \sigma'_\adam \>$ is won by Adam in $\Game'$.
    Consequently, it is also won by Eve, otherwise it would be a harmful deviation for Adam.
    It is therefore also won by Eve in $\Game$.
    This proves that $\sigma_\eve$ is a winning strategy in $\Game_{\|v_0}$.

    As a consequence, the problem of solving a two-player zero-sum parity game reduces to the SSE fixed-payoff existence problem in parity games with only two players; and therefore, the SSE fixed-payoff existence problem in parity games cannot be solved in polynomial time if only $n$ is fixed (unless parity games are solvable in $\Poly$).
    \qedhere
\end{itemize}
\end{proof}

\subsubsection{In \texorpdfstring{$\omega$}{ω}-regular games}

When the payoff functions are defined by parity automata, those ones must be incorporated in the arena of the game, entailing an exponential blowup.
However, there is no exponential blowup in the formula that defines the winning conditions of Prover and Challenger, hence the game can be solved by using algorithms that exists in the literature, in exponential time.

\begin{thm}
    In $\omega$-regular games, the fixed-payoff SSE (constrained) existence problem is $\ExpTime$-complete.
    When both  the number of players and colors are fixed, the problem becomes $\Poly$-easy but is $\XP$-complete.
    When only the number of colors is fixed, the problem remains $\PSpace$-hard.%
\end{thm}

\begin{proof}~
    \begin{itemize}
        \item \emph{$\ExpTime$-easiness and $\XP$-easiness (fixed number of players and colors).}
        
Let $\Game_{\|v_0}$ be an $\omega$-regular game, and let $\bx \in \NN$.
    Let $W = \{i\in\Pi ~|~ x_i=1\}$ be the set of players winning according to payoff vector $\bx$.
Let $\left(\auto_{i}\right)_{i\in\Pi\cup\{\Corr\}}$ be the parity automata defining the payoff of player $i$ and correctness constraint over $\Game_{\|v_0}$.
Let $S_{i}$ (resp. $T_{i}$, $\kappa_{i}$, $s_{0,i}$) be the set of states (resp. transition function, coloring mapping, initial state) of $\auto_{i}$.
We assume all these automata use at most $k$ colors, i.e. the co-domain of every $\kappa_{i}$ is included in $\{0,\dots,k\}$.
% We identify $\Pi$ with the set $\{1,\dots,n\}$ for simplicity in this proof.
Let \[m=\max\left(2,|V|,\max_{i\in\Pi\cup\{\Corr\}}\left|Q_{i}\right|\right)\] be the maximal size of the graph/automata provided as input (assumed to be at least 2 for complexity purposes).

We build the Extended deviator Game as the product of the Deviator Game with all automata $\auto_{i}$.
Formally:
    The \emph{extended deviator game} is the initialized graph game:
    \[\EDev_{\bx} \Game_{\|v_0} = \left( \{\PP, \CC\}, V^\ecd, E^\ecd, \left(V^\ecd_\PP, V^\ecd_\CC\right), \mu^\ecd\right)_{\|v_0^\ecd},\]
    where:

     \begin{itemize}
         \item the players are Prover and Challenger.

         \item Prover controls the set $V^\ecd_\PP$ and Challenger the set $V^\ecd_\CC$ with
         \begin{mathpar}
             V^\ecd_\PP = V \times 2^\Pi \times \prod_{i\in\Pi\cup\{\Corr\}} Q_{i}
             \and\text{and}\and
             V^\ecd_\CC= E \times 2^\Pi \times \prod_{i\in\Pi\cup\{\Corr\}} Q_{i}
         \end{mathpar}
         In the sequel we assume, for convenience, that we have $\Pi = \{1, \dots, n\}$, and we write $\bs = (s_{1},\dots,s_{n},s_\Corr)$ for elements of $\prod_{i\in\Pi\cup\{\Corr\}} Q_{i}$.

         \item The initial vertex is $v_0^\ecd = (v_0, \emptyset, s_{0,1},\dots,s_{0,n},s_{0,\Corr})=(v_0,\emptyset,\bs_0)$.

         \item The edge set $E^\ecd$ is defined as follows: from the vertex $(u, D, s_{1},\dots, s_{n},s_\Corr)$, Prover can go to every vertex of the form $(uv, D, s_{1},\dots, s_{n},s_\Corr)$ with $uv \in E$ (she \emph{proposes} the edge $uv$).
         From the vertex $(uv, D, s_{1},\dots, s_{n},s_\Corr)$, Challenger can go to the vertex $(v, D, s_{1}',\dots, s_{n}',s_\Corr')$ where $s_{i}'=T_{i}(s_{i},v)$ (he \emph{accepts} the edge $uv$), or to every vertex $(w, D \cup \{\eta\}, s_{1}',\dots, s_{n}',s_\Corr')$, with $w \neq v$, $uw \in E$, and $s_{i}'=T_{i}(s_{i},w)$, and where $\eta$ is the player controlling $u$ (he \emph{deviates} from Prover's proposal).
         
         \item We define the coloring functions $\kappa'_{i}$ for $i\in \Pi$ as $\kappa'_{i}(u, D, s_{1},\dots, s_{n},s_\Corr)=\kappa_{i}(s_{i})$ and $\kappa'_{i}(uv, D, s_{1},\dots, s_{n},s_\Corr)=k+1$ (thus effectively ignoring proposition vertices).

         \item Given a play $\chi=(u_0, D_0, \bs_0)(u_0v_0, D_0,\bs_0)(u_1, D_1,\bs_1) (u_1v_1, D_1,\bs_1)\dots$ in this game, we write $\D(\chi) = \bigcup_k D_k$, the set of players who deviated along the play $\chi$.
         % $\hchi$ the play in $\Game$ constructed by the actions of Prover and Challenger: if $\chi =  \dots$, then we define $\dchi = u_0 u_1 \dots$.

         \item Then, $\mu^\ecd$ is the Boolean payoff function defined as follows.
         A play $\chi$ is won by Challenger if and only if either:
         \begin{itemize}
    \item If $D(\chi) \neq \emptyset$, then a play $\chi$ is won by Challenger if and only if:
    \begin{mathpar}
        \chi \in \left( \bigcup_{i \in W} \overline{\Parity(\kappa'_{i})} \right) \cap \left( \bigcap_{i \in D(\chi)\cap W} \Parity(\kappa'_{i}) \right)
    \end{mathpar}
    (i.e., if $\chi$ gives to at least one player $i$ a payoff lower than $x_i$, and to every player who deviated at least the same payoff as in $\bx$).
    \item If $D(\chi) = \emptyset$, then a play $\chi$ is won by Challenger if and only if:
    \begin{mathpar}
        \chi \in \left(\bigcup_{i \in W} \overline{\Parity(\kappa'_{i})}\right) \cup \left(\bigcup_{i \notin W} \Parity(\kappa'_{i})\right) \cup \overline{\Parity(\kappa_\Corr')}
    \end{mathpar}
    (i.e. if $\chi$ gives to at least one player $i$ a payoff different from $x_i$ or the correction constraint is not satisfied).
    \end{itemize}
    \end{itemize}
    Note that we can convert all parity conditions into a chain of Rabin conditions: for $i\in \Pi\cup\{\Corr\}, 0\leq c \leq k$, let:
    \[\Xi^{i}_c=\left\{(u, D, s_{1},\dots, s_{n})\mmid \kappa_{i}(s_{i})\leq c\right\}.\]
    Then $\Parity(\kappa'_{i})=\bigcup_{\substack{0\leq c \leq k\\c \text{ is even}}} \Inf(\Xi^{i}_c) \cap \Fin(\Xi^{i}_{c-1})$ is the set of runs where $\Xi^{i}_c$ appears infinitely often and $\Xi^{i}_{c-1}$ only finitely often for an even color $c$.
    Similarly, $\overline{\Parity(\kappa'_{i})}=\bigcup_{\substack{0\leq c \leq k\\c \text{ is odd}}} \Inf(\Xi^{i}_c) \cap \Fin(\Xi^{i}_{c-1})$.
    Remark that we can (and do) ignore here the color $k+1$ added to ignore intermediate proposing vertices.
    We add extra sets:
    \begin{mathpar}
        \noDev=\left\{(u, \emptyset, \bs)\mmid u\in V, \bs \in \prod_{i\in\Pi\cup{\Corr}} Q_{i}\right\}
        \and\text{and}\and
        \Delta_i=\left\{(u, D, \bs)\mmid u\in V, i \in D, \bs \in \prod_{i\in\Pi\cup\Corr} Q_{i}\right\}
    \end{mathpar} that tracks whether no one (resp. player $i$) deviated.
    The winning condition for Challenger therefore becomes:
    \ifArxiv
    \begin{mathpar}
       \left(\Fin(\noDev) \cap \left( \bigcup_{\substack{i \in W\\0\leq c \leq k\\c \text{ is odd}}} \Inf(\Xi^{i}_c) \cap \Fin(\Xi^{i}_{c-1}) \right) \cap \left( \bigcap_{i\cap W} \Fin(\Delta_i) \cup \bigcup_{\substack{0\leq c \leq k\\c \text{ is even}}} \Inf(\Xi^{i}_c) \cap \Fin(\Xi^{i}_{c-1}) \right) \right)
        \and\cup\and
        \left(\Inf(\noDev) \cap \left(
            \left(\bigcup_{\substack{i \in W\\0\leq c \leq k\\c \text{ is odd}}} \Inf(\Xi^{i}_c) \cap \Fin(\Xi^{i}_{c-1})\right)
            \cup \left( \bigcup_{\substack{i \notin W\\0\leq c \leq k\\c \text{ is even}}} \Inf(\Xi^{i}_c) \cap \Fin(\Xi^{i}_{c-1})\right)
            \right.\right.\\\left.\left.\phantom{\bigcup_{\substack{W\\k\\d}}} % Line break within parenthesis has to be dealt with in an ad hoc manner
            \cup \left(\bigcup_{\substack{0 \leq c\leq k\\c \text{ is odd}}} \Inf(\Xi^{\Corr}_c) \cap \Fin(\Xi^{\Corr}_{c-1})\right)
        \right)\right)
    \end{mathpar}
    \else
    \begin{mathpar}
       \Biggg(\Fin(\noDev) \cap \left( \bigcup_{\substack{i \in W\\0\leq c \leq k\\c \text{ is odd}}} \Inf(\Xi^{i}_c) \cap \Fin(\Xi^{i}_{c-1}) \right)
       \and\cap\\ \left( \bigcap_{i\cap W} \Fin(\Delta_i) \cup \bigcup_{\substack{0\leq c \leq k\\c \text{ is even}}} \Inf(\Xi^{i}_c) \cap \Fin(\Xi^{i}_{c-1}) \right) \Biggg)
        \and\cup\and
        \Biggg(\Inf(\noDev) \cap \Biggg(
            \left(\bigcup_{\substack{i \in W\\0\leq c \leq k\\c \text{ is odd}}} \Inf(\Xi^{i}_c) \cap \Fin(\Xi^{i}_{c-1})\right)
            \cup \left( \bigcup_{\substack{i \notin W\\0\leq c \leq k\\c \text{ is even}}} \Inf(\Xi^{i}_c) \cap \Fin(\Xi^{i}_{c-1})\right)
            %\right.\right.\\\left.\left.\phantom{\bigcup_{\substack{W\\k\\d}}} % Line break within parenthesis has to be dealt with in an ad hoc manner
            \\
            \cup \left(\bigcup_{\substack{0 \leq c\leq k\\c \text{ is odd}}} \Inf(\Xi^{\Corr}_c) \cap \Fin(\Xi^{\Corr}_{c-1})\right)
        \Biggg)\Biggg)
    \end{mathpar}
    \fi

The extended deviator game therefore has $2\cdot m\cdot 2^{n+1} \cdot m^{n}=O(m^n)$ vertices and a winning condition that is provided by an Emerson-Lei condition with $(n+1)\cdot k$ colors, and a formula of size polynomial in $n$ and $k$.
Solving this game with the $\PSpace$ algorithm used in the proof of Theorem~\ref{thm:sseExistencePspaceEasy} would provide an $\ExpSpace$ algorithm.
However, using the algorithm from~\cite[Corollary 24]{HausmannLehautPiterman23}, this game can be solved in $O((n\cdot k)!\cdot (m^n)^{n\cdot k+2})=O((n\cdot k)!\cdot m^{n^2\cdot k+2n})$.
As a result the fixed-payoff SSE existence problem for multi-valued $\omega$-regular games is in $\ExpTime$.
When both $n$ and $k$ are fixed, this complexity boils down to $\Poly$, but note that it is an $\XP$ complexity as the size of the arena and automata ($m$) appears only as the exponent's base.

    \item \emph{$\ExpTime$-hardness and $\XP$-hardness (fixed number of players and colors).}

    We reduce the problem of existence of a winning strategy for two player zero-sum games with winning condition defined by a conjunction of parity automata (see Lemma~\ref{lem:conjunctionParityGameExptimeHard} below) to the SSE fixed-payoff unconstrained existence problem.
    Let $\Game_{\|v_0}=(V,E)$ be a two player game for players Adam ($\adam$) and Eve ($\eve$): $V=V_\adam \uplus V_\eve$.
    Let $\auto_1,\dots,\auto_n$ be $n$ parity automata over alphabet $V$.
    A play $\pi$ is winning for $\eve$ if for $i\in \{1,\dots,n\}$, $\pi\in\Lang(\auto_i)$.

    We build the $n+2$-player game $\Game^\dagger$ that has an SSE with payoff $(1,\dots,1,1,1)$ if and only if $\adam$ has a winning strategy in $\Game$.
    The $n+2$ players are $\{1,\dots,n,\adam,\eve\}$ (used in that order for payoff description).
    The game arena consists of the original game $\Game_{\|v_0}$ with a prepended path where each player from $1$ to $n$ can either choose to continue or reach a sink state $\blacktriangledown$.
    Players $i$ wins (payoff $1$) if either $\auto_i$ accepts in $\Game_{\|v_0}$ or the play ends up in $\blacktriangledown$; otherwise they get payoff $0$.
    This is ensured in $\auto^\dagger_{i}$ by adding two states to $\auto_i$ to take into account this prefix (see \figurename~\ref{fig:conjParity2sseAutomataRegular}).
    Player $\adam$ gets payoff $1$ if the play ends up in $\blacktriangledown$ and gets $0$ if it reaches $\Game_{\|v_0}$.
    $\eve$ always wins in this game.

    Assume that $\adam$ has a winning strategy in $\Game_{\|v_0}$.
    Let $\bsigma$ be the following profile in $\Game^\dagger$: for $i\in \{1,\dots,n\}$, from $q_i$ go to $\blacktriangledown$; $\adam$ plays his winning strategy in $\Game_{\|v_0}$, and $\eve$ play an arbitrary strategy.
    Since $\blacktriangledown$ is reached in the play induced by $\bsigma$, the payoff vector is $(1,\dots,1,1,1)$.
    We will show that $\bsigma$ is an SSE.

    Let $C$ be a coalition of players and assume $\bsigma'$ is a harmful deviation from players in $C$.
    First, since $\eve$ cannot be harmed by a deviation (she always wins), we can assume that $\eve\in C$.
    In addition, if one of the players $i\in\{1,\dots,n\}$ does not effectively deviate, then the play ends up in $\blacktriangledown$ and the payoff vector does not change; therefore all of these players must be in $C$.
    On the other hand, if the play changes the payoff by not ending in $\blacktriangledown$, then $\adam$ strictly decreases his payoff so cannot be part of the coalition.
    Therefore, any harmful deviation must be from the coalition $C=\{1,\dots,n,\eve\}$.
    As the deviation reaches $\Game_{\|v_0}$, $\adam$ is indeed harmed; for the deviation to be deemed harmful it remains to show that no player in the coalition is themself harmed.
    However, since $\adam$ has a winning strategy and plays it in $\Game_{\|v_0}$, the play is not accepted by (at least) one automaton $\auto_i$, so $i$ loses and thus decreases their payoff.
    As a result $i$ cannot be part of the coalition, which is a contradiction, and $\bsigma$ is an SSE.

    Now assume that $\eve$ has a winning strategy in $\Game_{\|v_0}$.
    Let $\bsigma$ be a profile with payoff vector $(1,\dots,1,1,1)$.
    We will show this profile is not an SSE.
    For this payoff to occur, that means the play ends up in $\blacktriangledown$.
    Consider the coalition $C=\{1,\dots,n,\eve\}$ and the deviation (which may not be actually effective for some players) $\bsigma'$ defined as follows: for $i \in \{1,\dots,n\}$, from $q_i$ go to $q_{i+1}$ (or $v_0$ if $i=n$); $\eve$ plays optimally in $\Game_{\|v_0}$, i.e. a winning strategy.
    The play thus built will reach $\Game_{\|v_0}$ where the created play is accepted by all $\auto_i$, so the payoff for player $i$ remains $1$ while the payoff for $\adam$ strictly decreases from $1$ to $0$.
    Therefore $\bsigma'$ is a harmful deviation and $\bsigma$ is not an SSE.    

    As a result there is an SSE with payoff $(1,\dots,1,1,1)$ if and only if $\adam$ has a winning strategy in $\Game_{\|v_0}$ for the winning condition $\bigcap_{i=1}^n \Lang(\auto_i)$.
    
    \item \emph{$\PSpace$-hardness (fixed number of colors)}% and $\W[1]$-hardness (fixed number of players and colors).}

    We reduce the problem of emptiness of the intersection of $n$ finite automata, which is $\PSpace$-hard. % and $\W[1]$-hard when $n$ is fixed.
    The arena is the same as the one given in \figurename~\ref{fig:autoInterEmptiness2CompParitySseCheck_arena} for the encoding into compositional Parity SSE-checking problem: Solver chooses letters between $a$ and $b$ in the \emph{word} module then gets into an \emph{approval} module where all players from $1$ to $n$ can either reject, reaching $\bt$ or accept and yield choice to the next player.
    If all players accept then vertex $\btd$ is reached.
    In that case Solver wins, but a player $i$ wins if and only if the prefix in $\{a,b\}^*$ was accepted by $\auto_i$.
    If $\bt$ is reached then all players lose; note that it is also the case if the game remains forever in the word module.
    We show that this game has an SSE where all players win iff there is a word accepted by all automata $\auto_i$.

    Assume there is an SSE with payoff $1$ for everyone.
    Then this strategy profile starts by generating a word on $\{a,b\}^*$ that is accepted by all automata $\auto_i$, otherwise they cannot all win.
    So there is a word over $\{a,b\}$ accepted by all automata.

    Now assume there is a word $w\in \{a,b\}^*$ accepted by all the $\auto_i$s.
    Let $\sigma_\solver$ be the strategy for Solver that starts by playing this word before entering the approval module.
    For player $1$ to $n$, let $\sigma_i$ be the strategy that chooses $\bt$ if the prefix on $\{a,b\}$ was not accepted by $\auto_i$ and goes to vertex $i+1$ (or $\btd$ for $n$) otherwise.
    When all these strategies are played together, the play indeed reaches $\btd$ and all players win: automata $\auto_i$ each reach the top right state with priority $0$.
    So this profile indeed achieves the desired payoffs.
    It remains to be shown that it is indeed an SSE.
    Assume there is a coalition $C$ of players who can create an harmful deviation without any player in $C$ losing.
    Assume a player $i\in\{1,\dots,n\}$ is in the deviating coalition.
    If $i$ goes to $\bt$, then he loses, which contradicts the assumption that no player in $C$ loses.
    So $C=\{\solver\}$.
    A deviation of Solver means either remaining in the word module forever (then he loses along with the others), or producing a word $w'\neq w$ before reaching vertex $1$.
    If $w'$ is also accepted by all automata, then the payoff does not actually change for anyone.
    If $w'$ is not accepted by $\auto_i$, then player $i$ will choose to go to $\bt$ so everyone, including Solver loses, once again in contradiction with the assumption.

    Note that only two colors were used in the parity automata.
    % Therefore the fixed-payoff SSE existence problem is:
    % \begin{itemize}
    %     \item $\PSpace$-hard when the number of colors is fixed.
    %     \item $\W[1]$-hard when both the number of players and colors are fixed.
    % \qedhere
    % \end{itemize}
    Therefore the fixed-payoff SSE existence problem is $\PSpace$-hard when the number of colors is fixed.
\end{itemize}
\end{proof}
\begin{figure*}
    \centering
    \begin{tikzpicture}[node distance=2cm]
        \node[literalState=$1$,initial] (q1) at (0,0) {$q_1$};
        \node[literalState=$2$,right of=q1] (q2) {$q_2$};
        \node[right of=q2] (qdots) {$\cdots$};
        \node[literalState=$n$,right of=qdots] (qn) {$q_n$};
        \node[state,rectangle,rounded corners=7pt,dotted,minimum size=1.5cm,right of=qn,node distance=3cm] (orig) {$\Game_{\|v_0}$};
        \node[nobodyState,below of=qdots] (sink) {$\blacktriangledown$};
\begin{pgfonlayer}{arrows}
        \path[->] (q1) edge (q2);
        \path[->] (q2) edge (qdots);
        \path[->] (qdots) edge (qn);
        \path[->] (qn) edge (orig);
        \path[->] (q1) edge[bend right] (sink);
        \path[->] (q2) edge[bend right,in=150] (sink);
        \path[->] (qn) edge[bend left] (sink);
        \path[->] (sink) edge[loop below] (sink);
        \end{pgfonlayer}
    \end{tikzpicture}
    \caption{Game $\Game^\dagger$. The player owning the vertex is given in the bottom right corner of each state.}
    \label{fig:conjParity2SSEGame}
\end{figure*}
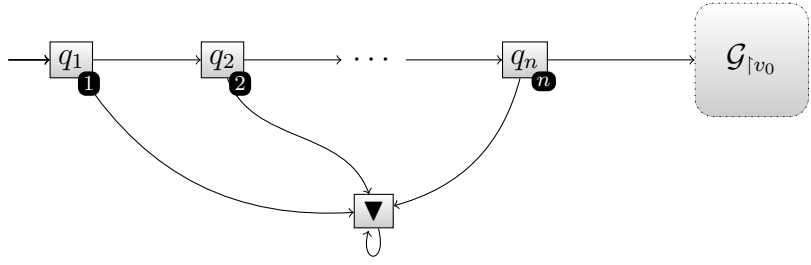

\begin{figure*}
    \centering
    \begin{subcaptionblock}{0.9\textwidth}
    \centering
    \begin{tikzpicture}[auto,node distance=3cm]
        \node[automatonState,initial] (init) at (0,0) {$0$};
        \node[automatonState,below right of=init] (sink) {$1$};
        \node[automatonState,rectangle,rounded corners=7pt,dotted,minimum size=1.5cm,above right of=sink] (autom) {$\auto_{i}$};
        \path[->] (init) edge[loop above] node{$q_1,\dots,q_n,\blacktriangledown$} (autom);
        \path[->] (init) edge node{$v_0$} (autom);
        \path[->] (init) edge node[swap]{$v\in V\setminus\{v_0\}$} (sink);
        \path[->] (autom) edge node{$q_1,\dots,q_n,\blacktriangledown$} (sink);
        \path[->] (sink) edge[loop above] node{$*$} (sink);
    \end{tikzpicture}
    \caption{Parity automata $\auto^\dagger_{i}$ for player $i$ to win. Numbers in state indicate the value of the coloring.}
    \label{fig:conjParity2sseAutomataRegular}
    \end{subcaptionblock}
    
    \begin{subcaptionblock}{0.4\textwidth}
    \centering
    \begin{tikzpicture}[auto,node distance=3cm]
        \node[automatonState,initial] (init) at (0,0) {$0$};
        \node[automatonState,right of=init] (lose) {$1$};
        \path[->] (init) edge[loop above] node{$q_1,\dots,q_n,\blacktriangledown,v\in V\setminus\{v_0\}$} (init);
        \path[->] (init) edge node{$v_0$} (lose);
        \path[->] (lose) edge[loop above] node{$*$} (lose);
    \end{tikzpicture}
    \caption{Parity automaton $\auto^\dagger_{\adam}$.}
    \label{fig:conjParity2sseAutomataAdam}
    \end{subcaptionblock}
    \hfill
    \begin{subcaptionblock}{0.4\textwidth}
    \centering
    \begin{tikzpicture}[auto,node distance=3cm]
        \node[automatonState,initial] (init) at (0,0) {$1$};
        \node[automatonState,right of=init] (lose) {$0$};
        \path[->] (init) edge[loop above] node{$q_1,\dots,q_n,\blacktriangledown,v\in V\setminus\{v_0\}$} (init);
        \path[->] (init) edge node{$v_0$} (lose);
        \path[->] (lose) edge[loop above] node{$*$} (lose);
    \end{tikzpicture}
    \caption{Parity automaton $\auto^\dagger_{\eve}$.}
    \label{fig:conjParity2sseAutomataEve}
    \end{subcaptionblock}

    \caption{Parity automata for payoffs in $\Game^\dagger$.}
    \label{fig:conjParity2sseAutomata}
\end{figure*}

\begin{figure*}
    
    \begin{subcaptionblock}{0.35\textwidth}
    \centering
    \begin{tikzpicture}[auto,node distance=3cm]
        \node[automatonState,initial] (init) at (0,0) {$1$};
        \node[automatonState] (ok) at (2,0) {$0$};
        \node[automatonState] (err) at (4,0) {$1$};
        \path[->] (init) edge[loop above] node{$a,b$} (init);
        \path[->] (ok) edge[loop above] node{$2,\cdots,n,\btd$} (ok);
        \path[->] (init) edge node {$1$} (ok);
        \path[->] (ok) edge node {$\bt$} (err);
        \path[->] (err) edge[loop above] node{$*$} (err);
        
        \node[anchor=north,text width=0.8\textwidth,font=\scriptsize] at (2,-10pt) {It is assumed seeing any other vertex leads to the rightmost state.};

    \end{tikzpicture}
    \caption{The parity automaton $\auto_{\solver}$.}
    \label{fig:autoInterEmptiness2omegaRegSseEx_solver}
    \end{subcaptionblock}
    \begin{subcaptionblock}{0.64\textwidth}
    \centering
    \begin{tikzpicture}[auto,node distance=3cm]
        \node[automatonState,rectangle,rounded corners=7pt,dotted,minimum size=3cm,initial] (autom) {$\auto_{i}$};
        \node[automatonState,double] (qacc) at (1, 1) {$1$};
        \node[automatonState] (someq) at (1, -1) {$1$};
        \node[automatonState,right of=qacc] (1) {$0$};
        \node[automatonState,right of=someq] (1b) {$1$};        
        
        \path[->] (qacc) edge node{$1$} (1);
        \path[->] (1) edge[loop right] node{$a,b,1,\dots,n,\btd$} (1);
        \path[->] (1) edge node{$\bt$} (1b);
        \path[->] (someq) edge node{$1$} (1b);
        \path[->] (1b) edge[loop right] node{$*$} (1b);
    \end{tikzpicture}
    \caption{The parity automaton $\auto_i'$ for player $i \in \{1, \dots, n\}$}
    \label{fig:autoInterEmptiness2omegaRegSseEx_i}
    \end{subcaptionblock}

    \caption{Encoding of the intersection emptiness problem into an $\omega$-regular fixed-payoff SSE existence problem.}
    \label{fig:autoInterEmptiness2omegaRegSseEx}
\end{figure*}

\begin{lm}\label{lem:conjunctionParityGameExptimeHard}
    The problem of existence of a winning strategy for two player zero-sum games with winning condition defined by a conjunction of parity automata is $\ExpTime$-hard and $\XP$-complete in the number of automata.
\end{lm}

\begin{proof}
We encode the emptiness problem for the intersection of deterministic top-down tree automata (DTDTA), which is $\ExpTime$-hard~\cite[Theorem 1.7.5]{TATA07} and $\XP$-complete in the number of automata~\cite[Theorem 4.16]{Wehar16}.

A DTDTA takes a tree labeled with function symbols from a set $F$ (each function $f$ having a specified arity $a(f)$) as an input and labels it with states of the automaton, according to rewriting rules in the transition function $\Delta$.
If $f$ is a symbol of arity $n$ then $\Delta(f,q)=(q_1,\dots,q_n)$ means that if $q$ is the state of the current node then the $n$ children will be labeled $q_1,\dots,q_n$, respectively.
Initially, the root of the tree is labeled by the initial state $q_0$.
A tree is accepted by a DTDTA if all leaves (\emph{i.e.} symbols of arity $0$) are labeled by an accepting state.

We transform a DTDA $\mathcal{T}$ into a two player game on a graph where vertices of $\adam$ are $F \uplus \{\bot\}$ and vertices of $\eve$ are $\left\{(f,i) \mmid f \in F, 1\leq i \leq a(f)\right\} \uplus\{\top\}$.
From vertex $f$ there are transitions to every $(f,i)$ and from vertex  $(f,i)$ (and $\top$) there are transitions to each $f'\in F$.
If a symbol has arity $0$ then there is a single transition from $f$ to $\bot$ and $\bot$ only has looping transitions to itself.
The intial vertex is $\top$.
So $\eve$ chooses the label of the tree node and $\adam$ chooses a branch.

We define a parity automaton matching $\Delta$ as follows.
The set of states is the same as the states of $\mathcal{T}$ plus two sink states $\blacktriangledown$ and $\blacktriangle$, with $q_0$ being the initial state.
Reading $f$ does not change the state, but reading $(f,i)$ in state $q$ makes the automaton move from state $q$ to $q_i$ when $\Delta(f,q)=(q_1,\dots,q_n)$.
Reading $\bot$ from an accepting state changes the state to $\blacktriangle$; from any other state, that leads automaton into $\blacktriangledown$.
State $\blacktriangle$ has priority $0$ while all other states have priority $1$.
As a result a path in the graph is accepted if and only if it ended in $\bot$ with an accepting state of $\mathcal{T}$.

If a tree is accepted by $\mathcal{T}$, one can devise a (memoryful) strategy for $\eve$ that consists in choosing the next function symbol according to this tree and the path chosen by $\adam$ in the tree.

For example, with symbols $x,y,y'$ (arity $2$) and $z,z'$ (arity $0$), if tree $x(y(z,z'),y'(z',z))$ is accepted, then from $\top$ choose $x$, then after $\top x (x,1)$ choose $y$, and from $\top x (x,1) y (y,2)$ choose $z'$.
That creates the path $\top x (x,1) y (y,2) z' \bot^\omega$.
The labeling of that branch in $\mathcal{T}$ will be mimicked by the automaton; since it is accepted it must end in an accepting state of $\mathcal{T}$, so the automaton will end up in state $\blacktriangle$ and the run will be accepted.

On the other hand, if the tree is not accepted then $\adam$ can choose a branch that is rejected by $\mathcal{T}$ and that would be rejected by the parity automaton as well.

As a result an accepted tree corresponds to a winning strategy for $\eve$.
So the emptiness problem can be encoded into deciding the corresponding two-player game between $\adam$ and $\eve$.%
\medskip

As in the above conversion the function symbols were converted into the arena and the states and transition function into a parity automaton, requiring the tree to be conjointly accepted by several DTDTA can be translated into having the path conjointly accepted by several parity automata.
\end{proof}\fi

\section{Conclusion}\label{sec:conclusion}

% \begin{itemize}
%     \item Multivalued games? (can we find an application for that?) \Le{Don't think so, je focusserais plutôt sur l'info imparfaite comme future work. D'ailleurs il faut qu'on explique quelque part pourquoi on se permet de regarder l'info parfaite pour l'instant.}
% \end{itemize}

In this paper, we present a game-theoretic approach to formalize fair exchange protocols under assumption of rational but possibly malicious agents. This framework is more expressive than previous frameworks for rational fairness.

We establish a strong link between our framework and concepts in game theory. More precisely, we show that the correctness conditions induced by our framework are strongly related to the concept of Strong Secure Equilibrium. This link with game theory allows us to develop efficient algorithms for the verification problems and we show that our algorithms have tractable complexity when the number of players is fixed, which is relevant in practice as the number of agents often remains small.

In future work, we plan to investigate the influence of imperfect information on the complexity of our problems and on the modeling.

\bibliography{bibli}

% Generated by IEEEtranS.bst, version: 1.14 (2015/08/26)
\begin{thebibliography}{10}
\providecommand{\url}[1]{#1}
\csname url@samestyle\endcsname
\providecommand{\newblock}{\relax}
\providecommand{\bibinfo}[2]{#2}
\providecommand{\BIBentrySTDinterwordspacing}{\spaceskip=0pt\relax}
\providecommand{\BIBentryALTinterwordstretchfactor}{4}
\providecommand{\BIBentryALTinterwordspacing}{\spaceskip=\fontdimen2\font plus
\BIBentryALTinterwordstretchfactor\fontdimen3\font minus
  \fontdimen4\font\relax}
\providecommand{\BIBforeignlanguage}[2]{{%
\expandafter\ifx\csname l@#1\endcsname\relax
\typeout{** WARNING: IEEEtranS.bst: No hyphenation pattern has been}%
\typeout{** loaded for the language `#1'. Using the pattern for}%
\typeout{** the default language instead.}%
\else
\language=\csname l@#1\endcsname
\fi
#2}}
\providecommand{\BIBdecl}{\relax}
\BIBdecl

\bibitem{AbrahamDolevGonenHalpern06}
I.~Abraham, D.~Dolev, R.~Gonen, and J.~Halpern, ``Distributed computing meets
  game theory: robust mechanisms for rational secret sharing and multiparty
  computation,'' in \emph{Proceedings of the 25th Annual ACM Symposium on
  Principles of Distributed Computing (PODC'06)}.\hskip 1em plus 0.5em minus
  0.4em\relax New York, NY, USA: Association for Computing Machinery, 2006, p.
  53–62.

\bibitem{EmersonLei87}
\BIBentryALTinterwordspacing
E.~{Allen Emerson} and C.-L. Lei, ``Modalities for model checking: branching
  time logic strikes back,'' \emph{Science of Computer Programming}, vol.~8,
  no.~3, pp. 275--306, 1987. [Online]. Available:
  \url{https://www.sciencedirect.com/science/article/pii/0167642387900360}
\BIBentrySTDinterwordspacing

\bibitem{DBLP:journals/jacm/AlurHK02}
R.~Alur, T.~A. Henzinger, and O.~Kupferman, ``Alternating-time temporal
  logic,'' \emph{J. {ACM}}, vol.~49, no.~5, pp. 672--713, 2002.

\bibitem{Aumann60}
R.~J. Aumann, ``{Acceptable points in games of perfect information.}''
  \emph{Pacific Journal of Mathematics}, vol.~10, no.~2, pp. 381 -- 417, 1960.

\bibitem{BackesCiobotaruKrohmer10}
M.~Backes, O.~Ciobotaru, and A.~Krohmer, ``Ratfish: A file sharing protocol
  provably secure against rational users,'' in \emph{Proceedings of the 5th
  European Symposium on Research in Computer Security (ESORICS'10)},
  D.~Gritzalis, B.~Preneel, and M.~Theoharidou, Eds.\hskip 1em plus 0.5em minus
  0.4em\relax Berlin, Heidelberg: Springer, 9 2010, pp. 607--625.

\bibitem{BaumWaidner00}
B.~Baum-Waidner and M.~Waidner, ``Round-optimal and abuse free optimistic
  multi-party contract signing,'' in \emph{Proceedings of the 27th
  International Colloquium on Automata, Languages and Programming
  ({ICALP}'00)}.\hskip 1em plus 0.5em minus 0.4em\relax Berlin, Heidelberg:
  Springer-Verlag, 2000, p. 524–535.

\bibitem{DBLP:conf/fossacs/Brenguier16}
\BIBentryALTinterwordspacing
R.~Brenguier, ``Robust equilibria in mean-payoff games,'' in \emph{Proceedings
  of the 19th International Conference on Foundations of Software Science and
  Computation Structures ({FOSSACS}'16)}, ser. Lecture Notes in Computer
  Science, B.~Jacobs and C.~L{\"{o}}ding, Eds., vol. 9634.\hskip 1em plus 0.5em
  minus 0.4em\relax Springer, Apr. 2016, pp. 217--233. [Online]. Available:
  \url{https://doi.org/10.1007/978-3-662-49630-5\_13}
\BIBentrySTDinterwordspacing

\bibitem{arxivVersion}
\BIBentryALTinterwordspacing
L.~Brice, J.-F. Raskin, M.~Sassolas, G.~Scerri, and M.~van~den Bogaard,
  ``Pessimism of the will, optimism of the intellect: Fair protocols with
  malicious but rational agents,'' 2024. [Online]. Available:
  \url{https://arxiv.org/abs/2405.18958}
\BIBentrySTDinterwordspacing

\bibitem{BruyereHautemRaskin18}
\BIBentryALTinterwordspacing
V.~Bruy{\`{e}}re, Q.~Hautem, and J.~Raskin, ``Parameterized complexity of games
  with monotonically ordered omega-regular objectives,'' in \emph{29th
  International Conference on Concurrency Theory, {CONCUR} 2018, September 4-7,
  2018, Beijing, China}, ser. LIPIcs, S.~Schewe and L.~Zhang, Eds., vol.
  118.\hskip 1em plus 0.5em minus 0.4em\relax Schloss Dagstuhl -
  Leibniz-Zentrum f{\"{u}}r Informatik, 2018, pp. 29:1--29:16. [Online].
  Available: \url{https://doi.org/10.4230/LIPIcs.CONCUR.2018.29}
\BIBentrySTDinterwordspacing

\bibitem{DBLP:conf/csl/BruyereMR14}
V.~Bruy{\`{e}}re, N.~Meunier, and J.~Raskin, ``Secure equilibria in weighted
  games,'' in \emph{{CSL-LICS}}.\hskip 1em plus 0.5em minus 0.4em\relax {ACM},
  2014, pp. 26:1--26:26.

\bibitem{ButtyanHubauxCapkun04}
L.~Butty{\'{a}}n, J.~Hubaux, and S.~Capkun, ``A formal model of rational
  exchange and its application to the analysis of syverson's protocol,''
  \emph{Journal of Computer Security}, vol.~12, no. 3-4, pp. 551--587, 2004.

\bibitem{ChatterjeeHenzingerJurdzinski05}
K.~Chatterjee, T.~A. Henzinger, and M.~Jurdzi{\'{n}}ski, ``Games with secure
  equilibria,'' in \emph{Formal Methods for Components and Objects}, F.~S.
  de~Boer, M.~M. Bonsangue, S.~Graf, and W.-P. de~Roever, Eds.\hskip 1em plus
  0.5em minus 0.4em\relax Berlin, Heidelberg: Springer Berlin Heidelberg, 2005,
  pp. 141--161.

\bibitem{ClarksonSchneider10}
M.~R. Clarkson and F.~B. Schneider, ``Hyperproperties,'' \emph{Journal of
  Computer Security}, vol.~18, no.~6, pp. 1157--1210, Sep. 2010.

\bibitem{TATA07}
\BIBentryALTinterwordspacing
H.~Comon{-}Lundh, M.~Dauchet, R.~Gilleron, C.~L{\"o}ding, F.~Jacquemard,
  D.~Lugiez, S.~Tison, and M.~Tommasi, \emph{Tree Automata Techniques and
  Applications}.\hskip 1em plus 0.5em minus 0.4em\relax Université de Lille,
  ENS Cachan, RWTH Aachen, Université Aix-Marseille, INRIA, CNRS, Nov. 2007.
  [Online]. Available:
  \url{https://inria.hal.science/hal-03367725v1/file/tata.pdf}
\BIBentrySTDinterwordspacing

\bibitem{FernauHoffmannWehar2021}
H.~Fernau, S.~Hoffmann, and M.~Wehar, ``Finite automata intersection
  non-emptiness: Parameterized complexity revisited,'' 2021.

\bibitem{GaleStewart53}
D.~Gale and F.~M. Stewart, ``Infinite games with perfect information,'' in
  \emph{Contributions to the Theory of Games (AM-28), Volume II}, H.~W. Kuhn
  and A.~W. Tucker, Eds.\hskip 1em plus 0.5em minus 0.4em\relax Princeton:
  Princeton University Press, 1953, pp. 245--266.

\bibitem{HausmannLehautPiterman23}
D.~Hausmann, M.~Lehaut, and N.~Pitermann, ``Symbolic reactive synthesis for the
  safety and {EL}-fragment of {LTL},'' 2023.

\bibitem{HunterDawar05}
P.~Hunter and A.~Dawar, ``Complexity bounds for regular games,'' in
  \emph{Mathematical Foundations of Computer Science 2005}, J.~J{\c
  e}drzejowicz and A.~Szepietowski, Eds.\hskip 1em plus 0.5em minus 0.4em\relax
  Berlin, Heidelberg: Springer Berlin Heidelberg, 2005, pp. 495--506.

\bibitem{Kozen77}
D.~Kozen, ``Lower bounds for natural proof systems,'' in \emph{18th Annual
  Symposium on Foundations of Computer Science ({SFCS}'77)}, 1977, pp.
  254--266.

\bibitem{KremerRaskin2003}
\BIBentryALTinterwordspacing
S.~Kremer and J.-F. Raskin, ``A game-based verification of non-repudiation and
  fair exchange protocols,'' \emph{Journal of Computer Security}, vol.~11,
  no.~3, pp. 399--429, 2003. [Online]. Available:
  \url{http://www.lsv.ens-cachan.fr/Publis/PAPERS/PS/Kremer-gameNRextended.ps}
\BIBentrySTDinterwordspacing

\bibitem{LatvalaHeljanko99}
T.~Latvala and K.~Heljanko, ``Coping with strong fairness,'' \emph{Fundamenta
  Informaticae}, vol.~43, no. 1-4, pp. 175--193, 2000.

\bibitem{Nas50}
J.~F. Nash, ``Equilibrium points in $n$-person games,'' in \emph{PNAS},
  vol.~36.\hskip 1em plus 0.5em minus 0.4em\relax National Academy of Sciences,
  1950, pp. 48--49.

\bibitem{pagnia1999impossibility}
H.~Pagnia and F.~C. G{\"a}rtner, ``On the impossibility of fair exchange
  without a trusted third party,'' Darmstadt University of Technology, Tech.
  Rep. TUD-BS-1999-02, 1999.

\bibitem{Pnueli77}
A.~Pnueli, ``The temporal logic of programs,'' in \emph{Proceedings of the 18th
  Annual Symposium on Foundations of Computer Science (FoCS'77)}.\hskip 1em
  plus 0.5em minus 0.4em\relax IEEE Computer Society, Oct. 1977, pp. 46--57.

\bibitem{RakotonirinaBartheSchneidewind24}
I.~Rakotonirina, G.~Barthe, and C.~Schneidewind, ``Decision and complexity of
  {D}olev-{Y}ao hyperproperties,'' \emph{Proceedings of ACM on Programming
  Languages}, vol.~8, no. POPL, Jan. 2024.

\bibitem{Wehar16}
M.~Wehar, ``On the complexity of intersection non-emptiness problems,'' Ph.D.
  dissertation, State Universit of New York at Buffalo, 2016.

\bibitem{ZhouGollmann97}
J.~Zhou and D.~Gollmann, ``An efficient non-repudiation protocol,'' in
  \emph{Proceedings of The 10th Computer Security Foundations Workshop}.\hskip
  1em plus 0.5em minus 0.4em\relax IEEE Computer Society, Jun. 1997, pp.
  126--132.

\end{thebibliography}
\ifDraft
\label{LastPage}
\noindent\textcolor{black!60}{\dotfill Pagecount: \pageref{LastPage}\dotfill}
\fi
%\newpage

\ifArxiv\newpage\fi
\appendix
\ifArxiv\section{Different approaches for fair exchange}\label{app:first_examples}

The naive exchange model is depicted in \figurename~\ref{fig:ongoingExampleTwoPlayers}.

\begin{figure*}[h!]
    \centering
    \begin{tikzpicture}[auto,node distance=4cm]
        \node[abstractstate] (init) at (0,0) {$\begin{array}{rl}%
            \alice:&\textit{item}_1 \\
            \bob:&\textit{item}_2
        \end{array}$};%$\alice:\textit{item}_1$\\$\bob:\textit{item}_2$};
        \node[abstractstate] at (6,2) (aliceOwnsAll) {$\begin{array}{rl}%
            \alice:&\textit{item}_1,\textit{item}_2 \\
            \bob:&\emptyset
        \end{array}$};
        \node[abstractstate] at (6,-2) (bobOwnsAll) {$\begin{array}{rl}%
            \alice:& \emptyset\\
            \bob:&\textit{item}_1,\textit{item}_2
        \end{array}$};
        \node[abstractstate] at (12,0) (end) {$\begin{array}{rl}%
            \alice:&\textit{item}_2 \\
            \bob:&\textit{item}_1
        \end{array}$};

        \path[->,aliceMove] (init) edge node[swap,sloped] {$\alice$ sends $\textit{item}_1$} (bobOwnsAll);
        \path[->,bobMove] (init) edge node[sloped] {$\bob$ sends $\textit{item}_2$} (aliceOwnsAll);
        \path[->,aliceMove] (aliceOwnsAll) edge node[sloped] {$\alice$ sends $\textit{item}_1$} (end);
        \path[->,bobMove] (bobOwnsAll) edge node[swap,sloped] {$\bob$ sends $\textit{item}_2$} (end);
    \end{tikzpicture}
    \caption{Exchange between Alice and Bob only.}
    \label{fig:ongoingExampleTwoPlayers}
\end{figure*}
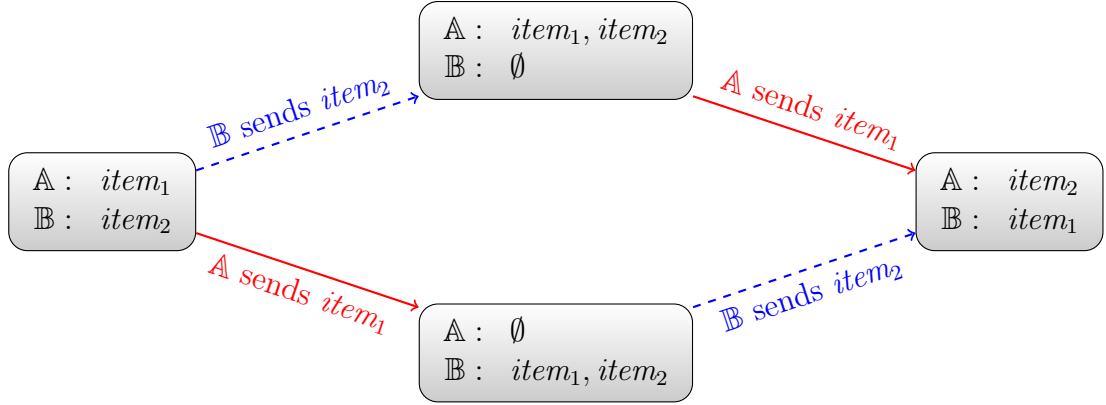

The exchange model with optimistic use of the TTP is given in \figurename~\ref{fig:ongoingExampleWithTTP}.

\begin{figure*}[h!]
    \centering
    \begin{tikzpicture}[auto,node distance=4cm]
        \node[abstractstate] (init) at (1,0) {$\begin{array}{rl}%
            \alice:&\textit{item}_1 \\
            \bob:&\textit{item}_2
        \end{array}$};%$\alice:\textit{item}_1$\\$\bob:\textit{item}_2$};
        \node[abstractstate] at (5,3) (aliceOwnsAll) {$\begin{array}{rl}%
            \alice:&\textit{item}_1,\textit{item}_2 \\
            \bob:&\emptyset
        \end{array}$};
        \node[abstractstate] at (5,-3) (bobOwnsAll) {$\begin{array}{rl}%
            \alice:& \emptyset\\
            \bob:&\textit{item}_1,\textit{item}_2
        \end{array}$};
        \node[abstractstate] at (14,0) (end) {$\begin{array}{rl}%
            \alice:&\textit{item}_2 \\
            \bob:&\textit{item}_1
        \end{array}$};
        \node[abstractstate] at (12,3) (aliceOwnsAllTtpAlerted) {$\begin{array}{rl}%
            \alice:&\textit{item}_1,\textit{item}_2 \\
            \bob:&\emptyset \\
            & \ttp \text{ alerted}
        \end{array}$};
        \node[abstractstate] at (12,-3) (bobOwnsAllTtpAlerted) {$\begin{array}{rl}%
            \alice:& \emptyset\\
            \bob:&\textit{item}_1,\textit{item}_2\\
            & \ttp \text{ alerted}
        \end{array}$};

        \path[->,aliceMove] (init) edge node[sloped,pos=0.55] {$\alice$ sends $\textit{item}_1$} (bobOwnsAll);
        \path[->,bobMove] (init) edge node[swap,sloped,pos=0.55] {$\bob$ sends $\textit{item}_2$} (aliceOwnsAll);
        \path[->,aliceMove] (aliceOwnsAll) edge node[swap,sloped] {$\alice$ sends $\textit{item}_1$} (end);
        \path[->,bobMove] (bobOwnsAll) edge node[sloped] {$\bob$ sends $\textit{item}_2$} (end);
        \path[->,aliceMove] (bobOwnsAll) edge node[swap,sloped] {$\alice$ alerts $\ttp$} (bobOwnsAllTtpAlerted);
        \path[->,bobMove] (aliceOwnsAll) edge node[sloped] {$\bob$ alerts $\ttp$} (aliceOwnsAllTtpAlerted);
        \path[->,ttpMove] (aliceOwnsAllTtpAlerted) edge node[text width=2.5cm,pos=0.8] {$\ttp$ forces $\alice$ to send $\textit{item}_1$} (end);
        \path[->,ttpMove] (bobOwnsAllTtpAlerted) edge node[swap,text width=2.5cm,pos=0.8] {$\ttp$ forces $\bob$ to send $\textit{item}_2$} (end);
    \end{tikzpicture}
    \caption{Exchange between Alice and Bob with a trusted third party.}
    \label{fig:ongoingExampleWithTTP}
\end{figure*}
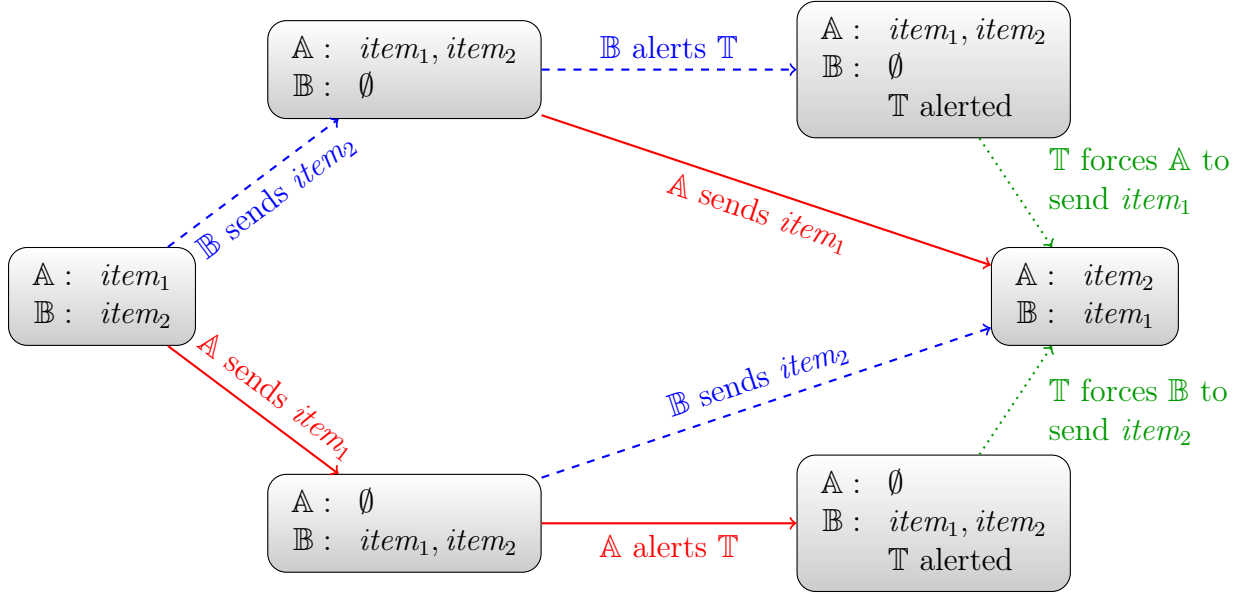

% \section{Exchange with an explicit infrastructure and TTP}
\label{app:ongoingExampleWithTTPandInfra}

The model of item exchange with explicit modeling of the infrastructure is depicted in \figurename~\ref{fig:ongoingExampleWithTTPandInfra}.

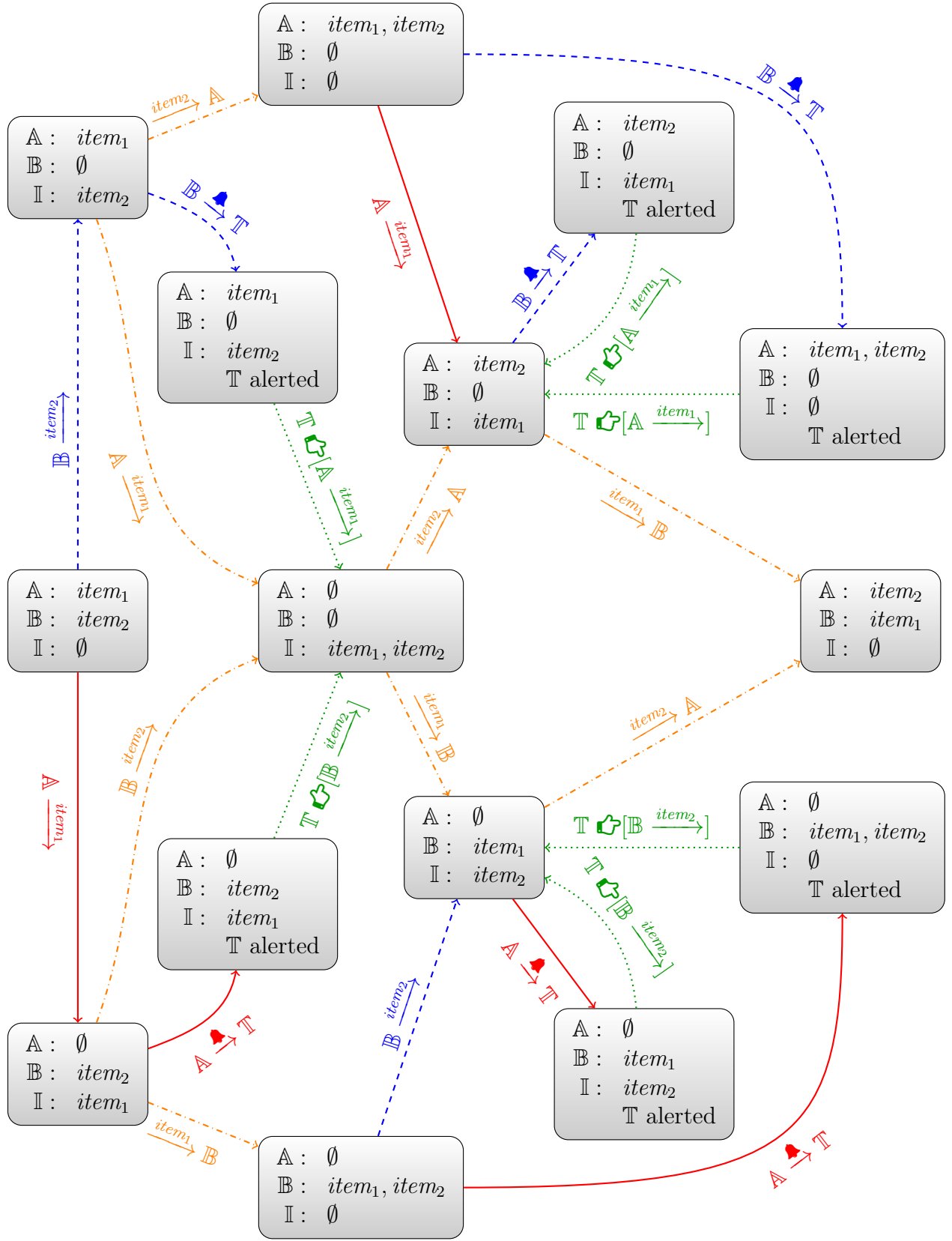
\begin{figure*}[hp!]
    \begin{tikzpicture}[auto,node distance=4cm]
        % \tikzstyle{abstractstate}+=[inner ysep=2.5pt] % Nécessaire si on ajoute une phrase dans la section
        \node[abstractstate] (init) at (0,0) {$\begin{array}{rl}%
            \alice:&\textit{item}_1 \\
            \bob:&\textit{item}_2 \\
            \II:&\emptyset
        \end{array}$};
        \node[abstractstate] at (0,8) (bobSent) {$\begin{array}{rl}%
            \alice:&\textit{item}_1 \\
            \bob:&\emptyset \\
            \II:&\textit{item}_2
        \end{array}$};
        \node[abstractstate] at (3,5) (bobSentTtpAlerted) {$\begin{array}{rl}%
            \alice:&\textit{item}_1 \\
            \bob:&\emptyset \\
            \II:&\textit{item}_2 \\
            & \ttp \text{ alerted}
        \end{array}$};
        \node[abstractstate] at (7,4) (received2aliceSent) {$\begin{array}{rl}%
            \alice:&\textit{item}_2 \\
            \bob:&\emptyset \\
            \II:&\textit{item}_1
        \end{array}$};
        \node[abstractstate] at (10,8) (received2aliceSentTtpAlerted) {$\begin{array}{rl}%
            \alice:&\textit{item}_2 \\
            \bob:&\emptyset \\
            \II:&\textit{item}_1 \\
            & \ttp \text{ alerted}
        \end{array}$};
        \node[abstractstate] at (5,10) (aliceOwnsAll) {$\begin{array}{rl}%
            \alice:&\textit{item}_1,\textit{item}_2 \\
            \bob:&\emptyset \\
            \II:&\emptyset
        \end{array}$};
        \node[abstractstate] at (0,-8) (aliceSent) {$\begin{array}{rl}%
            \alice:& \emptyset\\
            \bob:&\textit{item}_2 \\
            \II:&\textit{item}_1
        \end{array}$};
        \node[abstractstate] at (3,-5) (aliceSentTtpAlerted) {$\begin{array}{rl}%
            \alice:& \emptyset\\
            \bob:&\textit{item}_2 \\
            \II:&\textit{item}_1 \\
            & \ttp \text{ alerted}
        \end{array}$};
        \node[abstractstate] at (7,-4) (received1bobSent) {$\begin{array}{rl}%
            \alice:& \emptyset\\
            \bob:&\textit{item}_1 \\
            \II:&\textit{item}_2
        \end{array}$};
        \node[abstractstate] at (10,-8) (received1bobSentTtpAlerted) {$\begin{array}{rl}%
            \alice:& \emptyset\\
            \bob:&\textit{item}_1 \\
            \II:&\textit{item}_2 \\
            & \ttp \text{ alerted}
        \end{array}$};
        \node[abstractstate] at (5,-10) (bobOwnsAll) {$\begin{array}{rl}%
            \alice:& \emptyset\\
            \bob:&\textit{item}_1,\textit{item}_2 \\
            \II:&\emptyset
        \end{array}$};
        \node[abstractstate] at (5,0) (allSent) {$\begin{array}{rl}%
            \alice:& \emptyset\\
            \bob:& \emptyset\\
            \II:&\textit{item}_1,\textit{item}_2
        \end{array}$};
        \node[abstractstate] at (14,0) (end) {$\begin{array}{rl}%
            \alice:&\textit{item}_2 \\
            \bob:&\textit{item}_1 \\
            \II:&\emptyset
        \end{array}$};
        \node[abstractstate] at (13.5,4) (aliceOwnsAllTtpAlerted) {$\begin{array}{rl}%
            \alice:&\textit{item}_1,\textit{item}_2 \\
            \bob:&\emptyset  \\
            \II:&\emptyset\\
            & \ttp \text{ alerted}
        \end{array}$};
        \node[abstractstate] at (13.5,-4) (bobOwnsAllTtpAlerted) {$\begin{array}{rl}%
            \alice:& \emptyset\\
            \bob:&\textit{item}_1,\textit{item}_2 \\
            \II:&\emptyset\\
            & \ttp \text{ alerted}
        \end{array}$};

        \path[->,aliceMove] (init) edge node[swap,sloped,pos=0.4] {$\alice \xrightarrow{\textit{item}_1}$} (aliceSent);
        \path[->,infraMove] (aliceSent) edge node[swap,sloped,pos=0.4] {$\xrightarrow{\textit{item}_1}\bob$} (bobOwnsAll);
        
        \path[->,bobMove] (init) edge node[sloped,pos=0.4] {$\bob\xrightarrow{\textit{item}_2}$} (bobSent);
        \path[->,infraMove] (bobSent) edge node[sloped,pos=0.4] {$\xrightarrow{\textit{item}_2}\alice$} (aliceOwnsAll);

        \path[->,infraMove] (aliceSent) edge[out=70, in=-160,out distance=3cm] node[sloped,pos=0.6] {$\bob \xrightarrow{\textit{item}_2}$} (allSent);
        \path[->,infraMove] (bobSent) edge[out=-70, in=160,out distance=3cm] node[swap,sloped,pos=0.6] {$\alice \xrightarrow{\textit{item}_1}$} (allSent);
        \path[->,infraMove] (allSent) edge node[sloped] {$\xrightarrow{\textit{item}_1}\bob$} (received1bobSent);
        \path[->,infraMove] (allSent) edge node[swap,sloped] {$\xrightarrow{\textit{item}_2}\alice$} (received2aliceSent);

        \path[->,aliceMove] (aliceOwnsAll) edge node[swap,sloped] {$\alice \xrightarrow{\textit{item}_1}$} (received2aliceSent);
        \path[->,infraMove] (received2aliceSent) edge node[swap,sloped,pos=0.4] {$\xrightarrow{\textit{item}_1}\bob$} (end);
        \path[->,bobMove] (bobOwnsAll) edge node[sloped] {$\bob\xrightarrow{\textit{item}_2}$} (received1bobSent);
        \path[->,infraMove] (received1bobSent) edge node[sloped] {$\xrightarrow{\textit{item}_2}\alice$} (end);
        
        \path[->,aliceMove] (bobOwnsAll) edge[out=0,in=-90,out distance=6cm,in distance=4cm] node[swap,sloped] {$\alice\xrightarrow{\text{\warning}}\ttp$} (bobOwnsAllTtpAlerted);
        \path[->,aliceMove] (received1bobSent) edge node[swap,sloped] {$\alice\xrightarrow{\text{\warning}}\ttp$} (received1bobSentTtpAlerted);
        \path[->,aliceMove] (aliceSent) edge[out=20,in=-100] node[swap,sloped] {$\alice\xrightarrow{\text{\warning}}\ttp$} (aliceSentTtpAlerted);
        \path[->,bobMove] (aliceOwnsAll) edge[out=0,in=90,out distance=6cm,in distance=4cm] node[sloped] {$\bob\xrightarrow{\text{\warning}}\ttp$} (aliceOwnsAllTtpAlerted);
        \path[->,bobMove] (received2aliceSent) edge node[sloped] {$\bob\xrightarrow{\text{\warning}}\ttp$} (received2aliceSentTtpAlerted);
        \path[->,bobMove] (bobSent) edge[out=-20,in=100] node[sloped] {$\bob\xrightarrow{\text{\warning}}\ttp$} (bobSentTtpAlerted);
        
        \path[->,ttpMove] (bobOwnsAllTtpAlerted) edge node[sloped] {$\ttp$ \force$[\bob\xrightarrow{\textit{item}_2}]$} (received1bobSent);
        \path[->,ttpMove] (received1bobSentTtpAlerted) edge[bend right] node[sloped] {$\ttp$ \force$[\bob\xrightarrow{\textit{item}_2}]$} (received1bobSent);
        \path[->,ttpMove] (aliceSentTtpAlerted) edge node[swap,sloped] {$\ttp$ \force$[\bob\xrightarrow{\textit{item}_2}]$} (allSent);
        \path[->,ttpMove] (aliceOwnsAllTtpAlerted) edge node[sloped,swap] {$\ttp$ \force$[\alice\xrightarrow{\textit{item}_1}]$} (received2aliceSent);
        \path[->,ttpMove] (received2aliceSentTtpAlerted) edge[bend left] node[sloped,swap] {$\ttp$ \force$[\alice\xrightarrow{\textit{item}_1}]$} (received2aliceSent);
        \path[->,ttpMove] (bobSentTtpAlerted) edge node[sloped] {$\ttp$ \force$[\alice\xrightarrow{\textit{item}_1}]$} (allSent);
    \end{tikzpicture}
    \caption{Exchange between $\alice$ and $\bob$ with explicit infrastructure modeling and TTP.}
    \label{fig:ongoingExampleWithTTPandInfra}
\end{figure*}\fi
\ifArxiv
\section{Case study: Baum-Waidner Optimistic Multi-Party Contract Signing}\label{app:BaumWaidner}

In \cite{BaumWaidner00}, the authors present a protocol for $n$ agents $\{\player_1,\dots,\player_n\}$ to sign a contract, only relying on a trusted third-party (TTP, $\ttp$) to recover from timeout (or potential cheating attempt).
The protocol is carried out over $n+1$ rounds.
In the first round, all agents send a level-1 promise of signature to all other agents to start the protocol.
All agents wait to have received all such promises to move to round $2$.
In subsequent rounds ($2$ to $n+1$), agents send level-$r$ promises as well as acknowledgment of success for level $r-1$.
Once an agent has received all level-$r$ promises and acknowledgment of success for level $r-1$, they move to round $r+1$.
A level-$r$ promise from $\player_i$ is a message signed by $i$ that contains all the acknowledgments from all agents for round $r-2$, that were received at the end of the previous round.
Note that at level $2$, the promises of level $1$ serve as acknowledgments that ``round $0$'' was successfully carried out.
The level-$n+1$ promise is considered an actual signature of the contract.

If at any point an agent suspects the protocol is not going as expected (for example if they fail to receive a message from another agent), they can contact the TTP to request a resolution at round $r$.
The resolution request must contain the promises of level $r-1$ and acknowledgments of level $r-2$ received from all the other agents: these should have been received from everyone before even moving to round $r$.
The TTP can simply abort if the level of promises in the resolution request is not high enough; this abortion can later be recovered if another agent provides a resolution request with level of promises that proof that level $r$ was actually successfully performed.
Otherwise, they produce a resolution token that will be transmitted to any subsequent agent requesting a resolution, as it acts as a proof of signature.
Any given agent is only given one chance to contact the TTP (all further requests are ignored); therefore an agent that asks for a resolution at an early stage, receiving an abort token, but still participates in the protocol would never be able to obtain a resolution token from the TTP if the protocol does not reach its normal end.

Remark that in this case the channels are assumed perfect and therefore are not part of the model: any message sent is immediately received by the intended recipient.

\subsection{Modeling of the protocol}

\subsubsection{Messages}

The possible messages of the protocol are as follows.
They all are assumed to include the necessary metadata identifying the contract being signed and the instance of the protocol.
For each agent $i,j \in \{1,\dots,n\}$ and $r\in \{1,\dots,n+1\}$, $\Promise(i,j,r)$ is the sending of signature promise of level $r$ from $\player_i$ to $\player_j$.
For each agent $i,j,r \in \{1,\dots,n\}$, $\Ack(i,j,r)$ is the sending of acknowledgment of the completion of level $r$, from $\player_i$ to $\player_j$.
The TTP can send $\Aborted(r,i)$ which is an abort token sent to $\player_i$ and $\Signed(r,i)$, a proof of the contract signature.

As the level-$r$ promise includes the acknowledgments received from other agents and is signed by the sender, $\Promise(i,j,r)$ must first be sent:
\begin{itemize}
 \item by agent $\player_i$,
 \item if $r>1$, for all $j\neq i$, $\Ack(j,i,r-1)$ has been sent.
\end{itemize}
The acknowledgments are signed so $\Ack(i,j,r)$ must first be sent by $\player_i$, but nothing prevents them from sending it at any time (even though it may be a bad idea).
Nothing prevents promises, acknowledgments, or a proof of signature from the TTP to be forwarded to other players:
\begin{itemize}
    \item Once $\Promise(i,j,r)$ has been sent, $\player_j$ can send $\Promise(i,k,r)$ for any $k\in\{1,\dots,n\}$.
    \item Once $\Ack(i,j,r)$ has been sent, $\player_j$ can send $\Ack(i,k,r)$ for any $k\in\{1,\dots,n\}$.
    \item Once $\Signed(r,i)$ has been sent, $\player_i$ can send $\Signed(r,j)$ for any $j\in\{1,\dots,n\}$.
\end{itemize}

For each agent $i \in \{1,\dots,n\}$ and $r\in \{1,\dots,n+1\}$, $\Resolve(i,r)$ is the sending of resolution request to the TTP at round $r$.
It can only be sent:
\begin{itemize}
 \item by agent $\player_i$,
 \item if $r>2$, for all $j\neq i$, $\Ack(j,i,r-2)$ has been sent.
 \item if $r>1$, for all $j\neq i$, $\Promise(j,i,r-1)$ has been sent.
\end{itemize}

Message $\Signed(r,i)$ contains information from a $\Resolve$ request at level $r$ so one must have been sent to $\ttp$ before: for $r\in\{1,\dots,n+1\}$, $\Signed(r,i)$ can be sent only by the TTP and if for some $j$, $\Resolve(j,r)$ was sent before.

The $\Signed(r,i)$ message can act for $\player_i$ as a proof of the signature of the contract by all parties.
We can therefore define, for each player $i$, the state of having a signed contract as:
\[O_{\Signed,i} \eqdef \bigwedge_{j\neq i} \Promise(j,i,n+1) \vee \Signed(r,i)\]
A correct execution of the protocol, and therefore the collaborative goal of the agents is \[\Corr=\bigwedge_i O_{\Signed,i}.\]

\subsubsection{Arena}
The game arena is composed of states $2^\actions\times\{\ttp,\player_0,\dots,\player_n\}$ where
\ifArxiv\[\begin{array}{rcl}
\actions &=&
\left\{\Promise(i,j,r) \ \middle|\  r\in \{1,\dots,n+1\}, i,j \in \{1,\dots,n\}\right\}\\&\cup&
\left\{\Ack(i,j,r) \ \middle|\  r\in \{1,\dots,n+1\}, i,j \in \{1,\dots,n\}\right\}\\&\cup&
\left\{\Resolve(i,r) \ \middle|\  r\in \{1,\dots,n+1\}, i \in \{1,\dots,n\}\right\}\\&\cup&
\left\{\Signed(r,i) \ \middle|\  r\in \{1,\dots,n+1\}, i \in \{1,\dots,n\}\right\}\\&\cup&
\left\{\Aborted(r,i) \ \middle|\  r\in \{1,\dots,n+1\}, i \in \{1,\dots,n\}\right\}
\end{array}\]
\else
$\actions$ is given in \tablename~\ref{tab:actionsBW}.
\begin{table*}
\[\begin{array}{rclcl}
\actions &=&
\left\{\Promise(i,j,r) \ \middle|\  r\in \{1,\dots,n+1\}, i,j \in \{1,\dots,n\}\right\} &\cup&
\left\{\Ack(i,j,r) \ \middle|\  r\in \{1,\dots,n+1\}, i,j \in \{1,\dots,n\}\right\}\\&\cup&
\left\{\Resolve(i,r) \ \middle|\  r\in \{1,\dots,n+1\}, i \in \{1,\dots,n\}\right\} &\cup&
\left\{\Signed(r,i) \ \middle|\  r\in \{1,\dots,n+1\}, i \in \{1,\dots,n\}\right\}\\&\cup&
\left\{\Aborted(r,i) \ \middle|\  r\in \{1,\dots,n+1\}, i \in \{1,\dots,n\}\right\}
\end{array}\]
\caption{Actions in the Baum-Weidner protocol.}
\label{tab:actionsBW}
\end{table*}
\fi
\subsubsection{Mealy machines for strategy following the protocol}

A Mealy machine $\M_i^\tau$ for the strategy of $\player_i$ that allows up to $\tau$ turns per round before starting the resolution process with the TTP can be defined as follows.
In this description, we may abusively use $\player_0$ to denote $\ttp$, in order to facilitate the round-robin nature of the game.
Machine $\M_i^\tau$ has $(n+1)\times (\tau+1)+2$ states: $(r,t)\in\{1,\dots,n+1\}\times\{0,\dots,\tau\}$, a recovery state $R$ and a stop state $s$.

From $(1,t)$, if the game is in a state $(v,\player_i)$ that satisfies $\bigwedge_{j\neq i} \Promise(j,i,1)$, then there is a transition to $(2,0)$ with output:
$$(v\cup \{\Promise(i,j,2),\Ack(i,j,1) ~|~ j\neq i\},\player_{i+1\mod n}).$$
Otherwise, if $t<\tau$ then reading $(v,\player_i)$ outputs $(v,\player_{i+1 \mod n})$ and moves to $(1,t+1)$.
From $(1,\tau)$, reading $(v,\player_i)$ that does not satisfy $\bigwedge_{j\neq i} \Promise(j,i,1)$ leads to state $R$ with output: $$(v\cup \{\Resolve(i,1)\},\player_{i+1 \mod n}).$$

From $(r,t)$ with $1<r\leq n$, if the game is in a state $(v,\player_i)$ that satisfies $\bigwedge_{j\neq i} \Promise(j,i,n) \wedge \Ack(j,i,n-1)$, then there is a transition to $(r+1,0)$ with output: $$(v\cup\{\Promise(i,j,n+1),\Ack(i,j,n) ~|~ j\neq i\}.$$
Otherwise, if $t<\tau$ then reading $(v,\player_i)$ outputs $(v,\player_{i+1 \mod n})$ and moves to $(r,t+1)$.
From $(r,\tau)$, reading $(v,\player_i)$ that does not satisfy $\bigwedge_{j\neq i} \Promise(j,i,n) \wedge \Ack(j,i,n-1)$ leads to state $R$ with output: $$(v\cup \{\Resolve(i,r)\},\player_{i+1 \mod n}).$$

From $R$, reading $(v,\player_i)$ where $v$ satisfies $\Aborted(r,i) \vee \Signed(r,i)$  from some $r$ leads to $S$.
All other values lead to $R$ with output $(v,\player_{i+1\mod n})$.
From $S$, reading $(v,\player_i)$ outputs $(v,\player_{i+1\mod n})$ and remains in $S$.

In all states, if $(v,\player_j)$ is read with $j\neq i$, the machine does not move and outputs $(v,\player_{j+1\mod n})$.

\bigskip
A Mealy machine for the TTP $\M_\ttp$ has states of the form $(c,a,s)$ where $c\subseteq \{1,\dots,n\}$ is the set of players who have contacted the TTP, $a\subseteq \{1,\dots,n\}\times\{1,\dots,n+1\}$ is the set of abort tokens having been sent, $s\in\{0,\dots,n+1\}$ indicates that a signing token has been sent at level $s$, with $0$ indicating no signing token has been sent.
From state $(c,a,0)$, if the game is in state $(v,\ttp)$ where $v$ satisfies $\Resolve(i,1)$ and $i\notin c$, the machine goes to state $(c\cup\{i\},a\cup\{(i,1)\},0)$ and outputs $(v\cup\{\Aborted(r,i)\},\player_1)$.
From state $(c,a,0)$, if the game is in state $(v,\ttp)$ where $v$ satisfies $\Resolve(i,r)$ for $r>1$, $i\notin c$
\begin{itemize}
    \item if for any $(p,\ell)\in a$, $\ell<r$, the machine goes to state $(c\cup\{i\},a,r)$ and outputs $(v\cup\{\Signed(r,i)\},\player_1)$.
    \item otherwise the machine goes to state $(c\cup\{i\},(a\cup\{(i,r)\},0)$ and outputs transition $(v\cup\{\Aborted(r,i)\},\player_1)$.
\end{itemize}
From state $(c,a,\rho)$, if the game is in state $(v,\ttp)$ where $v$ satisfies $\Resolve(i,r)$ for $r>1$, and $i\notin c$, then the machine goes to state $(c\cup\{i\},a,\rho)$ and outputs $(v\cup\{\Signed(\rho,i)\},\player_1)$.
All other cases of input $(v,\player_i)$ have the machine remain in the same state and output $(v,\player_{i+1\mod n})$.

\subsection{Objectives and attacks models}

In general, agents consider that they are in a bad situation if they have given their signature to another party while not having the full set of signatures.
The actual individual objectives of agents may however vary with the attack model being considered.

\subsubsection{Selfish agents}
In this model, agents are selfish and only care about the objectives informally described above.
We define:
\begin{mathpar}
%     \Ind_i \eqdef O_{\Signed,i} \wedge \bigvee_{j\neq i}\neg O_{\Signed,j}
% \and
    \Wrong_i \eqdef \neg O_{\Signed,i} \wedge \bigvee_{j\neq i} \Promise(i,j,n+1)
\end{mathpar}

\subsubsection{Sacrifices for the greater ``good''}

In the original model, it is assumed that up to $n-1$ agents can be dishonest and not follow the protocol in order to cheat the last (and honest) agent.
But one can also assume that these dishonest may even be willing to sacrifice themselves by being wronged as long as one member of the group does achieve its individual objective.
In order to differentiate ``losing by sacrifice'' from ``losing by being cheated'', we introduce a $\Boast$ family of actions that allows cheaters achieving their individual objective to exclude other agents from the coalition: if he is wronged that would count as a loss, not a sacrifice.
Message $\Boast(i,j)$ is the announcement by $\player_i$ that if $\player_j$ is wronged it should not count as a sacrifice performed to help $i$.
It can be sent only by $\player_i$ at any time.

We can now define the objectives for all players:
\begin{mathpar}
%     \Ind_i \eqdef O_{\Signed,i} \wedge \bigvee_{j\neq i} \left(\Boast(i,j) \wedge\neg O_{\Signed,j} \right)
% \and
    \Wrong_i \eqdef   \bigvee_{j\neq i} \left(\Boast(j,i) \wedge O_{\Signed,j}\right)  \wedge \neg O_{\Signed,i} \wedge\bigvee_{j\neq i} \Promise(i,j,n+1)
\end{mathpar}
These objective mean that to achieve an attack, an agent must have a signed contract and be able to point another agent who does not.
By being pointed (only by such an agent), and being in a bad situation of having sent its signature while not having received a signed contract, a user is wronged.
Accomplices that sacrifice themselves to allow for another agent to achieve its individual goal are not considered wronged if they were not pointed.
It allows these agents to actually participate in a dishonest coalition without officially being wronged.

\subsubsection{A ghost in the protocol}
The above approach can be generalized to include finer specifications on the way coalitions are allowed to form.
That can be achieved by using a \emph{Ghost} agent $\ghost$ whose only actions are to $\Curse$ players who, then, cannot sacrifice themselves in order to form a harmful coalition.
Agent $\ghost$ has no action besides $\Curse(i)$ for $i\in\{1,\dots,n\}$, that indicates that $\player_i$ is a potential target of an attack: if he is wronged that would count as a loss, not a sacrifice.

For example, the above case can be formalized using the following objectives for players:
\begin{mathpar}
%     \Ind_i \eqdef \bot
% \and
    \Wrong_i \eqdef \Curse(i) \wedge \neg O_{\Signed,i} \wedge \bigvee_{j\neq i} \Promise(i,j,n+1)
\end{mathpar}
And for the ghost $\Wrong_\ghost \eqdef \bot$.

As it has no action (besides cursing), this can only be achieved with the help of other agents.

Also note that the above definition allows up to $n-1$ agents to team up in the attack, but finer models could also be used by modifying the ghost's wronging condition.
For example, to allow only $k$ attackers to sacrifice themselves to cheat, one could define
% \[\Ind_\ghost'=\Ind_\ghost \wedge \bigvee_{\substack{A \subseteq \{1,\dots,n\}\\|A|\geq\left\lfloor\frac{n}{2}\right\rfloor}} \bigwedge_{i \in A} \Curse(i)\]
\[\Wrong_\ghost'=\bigvee_{i \in \Pi} \Curse(i) \wedge \bigvee_{\substack{A \subseteq \{1,\dots,n\}\\|A|\geq k}} \bigwedge_{j \in A} \neg\Curse(j).\]

\fi
% \ifPagesLimited\input{appendixAlgorithmsBooleanSSE}\fi
\end{document}